\documentclass[11pt,a4paper,reqno,french,tikz]{amsart}
\usepackage[utf8]{inputenc} \usepackage[T1]{fontenc} 
\usepackage{amsthm,amsmath,amsfonts,amssymb,amsxtra,appendix,bookmark,dsfont,bm,mathrsfs,amstext,amsopn,mathrsfs,mathtools,comment,cite,hyperref,color,xcolor}
\usepackage{graphicx} 
\usepackage{enumitem}
\usepackage{tikz,tikz-cd}
\usetikzlibrary{intersections}
\usepackage{pgfplots}
\usepackage{rotating}
\usepackage{cite}
\usepackage{color}
\usepackage{epsf}
\usetikzlibrary{datavisualization}
\usetikzlibrary{datavisualization.formats.functions}
\usetikzlibrary{shapes.callouts}
\tikzset{
  level/.style   = { ultra thick, blue },
  connect/.style = { dashed, red },
  notice/.style  = { draw, rectangle callout, callout relative pointer={#1} },
  label/.style   = { text width=2cm }
}
\usepackage{bbm} \newcommand{\bB}{\mathbbmss{B}}

\newtheorem{theorem}{Theorem}[section]\newtheorem{lemma}[theorem]{Lemma}\newtheorem{proposition}[theorem]{Proposition}\newtheorem{corollary}[theorem]{Corollary}\newtheorem{conjecture}[theorem]{Conjecture}\newtheorem{remark}[theorem]{Remark}

\let\C\relax\newcommand{\C}{\mathbb{C}}\newcommand{\Z}{\mathbb{Z}}\newcommand{\R}{\mathbb{R}}\newcommand{\N}{\mathbb{N}}

\newcommand\cC{\mathcal{C}}\newcommand\cD{\mathcal{D}}\newcommand\cE{\mathcal{E}}\newcommand\cF{\mathcal{F}}\newcommand\cH{\mathcal{H}}\newcommand\cI{\mathcal{I}}\newcommand\cL{\mathcal{L}}\newcommand\cM{\mathcal{M}}\newcommand\cP{\mathcal{P}}\newcommand\cR{\mathcal{R}}\newcommand\cS{\mathcal{S}}\newcommand\cT{\mathcal{T}}\newcommand\cV{\mathcal{V}}

\DeclareMathOperator{\tr}{Tr}\DeclareMathOperator{\ran}{Ran}\DeclareMathOperator{\Ker}{Ker}\DeclareMathOperator{\re}{Re}
\def\d{{\rm d}}

\renewcommand{\ge}{\geqslant}\renewcommand{\le}{\leqslant}

\newcommand{\pa}[1]{\left( #1 \right)} 
\newcommand{\acs}[1]{\left\{ #1 \right\}} 
\newcommand{\seg}[1]{\left[ #1 \right]} 
\newcommand{\ab}[1]{\left|#1\right|} 
\newcommand{\ps}[1]{\left< #1 \right>} 
\newcommand{\proj}[1]{\big| #1 \big> \big< #1 \big|} 
\newcommand{\nor}[2]{ \left| \! \left| #1 \right| \! \right|_{#2} } 

\newcommand\vp{\varphi} 
\newcommand{\ep}{\varepsilon} 
\let\p\relax\newcommand{\p}{\Psi} 
\newcommand{\na}{\nabla} 

\newcommand{\f}[2]{\frac{#1}{#2}} 
\newcommand{\bul}{$\bullet$ \hspace{0.1cm}} 
\newcommand{\mymax}[1]{\underset{\substack{#1}}{\text{\normalfont{max}}}\quad} 
\newcommand{\mymin}[1]{\underset{\substack{#1}}{\text{\normalfont{min}}}\quad} 
\newcommand{\ind}[1]{_{\textup{#1}}} 
\newcommand{\apo}[1]{``#1''} 

\newcommand{\empt}{\varnothing}
\DeclareMathOperator{\supp}{{\rm supp }}
\newcommand{\mysum}[2]{\sum_{\substack{#1}}^{#2}}
\newcommand{\st}{\hspace{0.1cm} \bigr\vert \hspace{0.1cm}}
\newcommand{\mylim}[1]{\underset{\substack{#1}}{\text{\normalfont{lim}}} \hs}
\newcommand{\indic}{\mathds{1}}
\newcommand{\ra}{\rightarrow}
\newcommand{\ro}{\rho}
\newcommand{\expo}[1]{^{\textup{#1}}} 
\newcommand{\ii}{\infty}
\newcommand{\tx}[1]{\textup{#1}} 
\newcommand{\bhs}{\hspace{1cm}}
\newcommand{\hs}{\hspace{0.1cm}}
\newcommand{\sch}{\mathfrak{S}} 
\newcommand{\bpa}[1]{\big( #1 \big)}
\newcommand\restr[2]{#1_{\mkern 1mu \vrule height 2ex\mkern2mu #2}}
\newcommand{\ud}{\frac{1}{2}}
\newcommand{\mysup}[1]{\underset{\substack{#1}}{\text{\normalfont{sup}}}\hs}
\newcommand{\smallsum}{\begingroup\textstyle\sum\endgroup}
\newcommand{\myinf}[1]{\underset{\substack{#1}}{\text{\normalfont{inf}}}\hs}
\newcommand{\bigst}{\hspace{0.1cm} \Big| \hspace{0.1cm}} 
\newcommand{\ex}[1]{^{(#1)}}
\newcommand{\us}[2]{\underset{#1}{#2}}
\DeclareMathOperator{\sgn}{sgn}
\DeclareMathOperator{\diam}{diam}
\DeclareMathOperator{\vect}{Span}
\newcommand{\mat}[1]{\begin{pmatrix} #1 \end{pmatrix}} 
\newcommand{\floor}[1]{\left\lfloor #1 \right\rfloor} 
\newcommand{\ov}[1]{\overline{#1}} 
\newcommand{\mediumint}{\begingroup\textstyle\int\endgroup}
\newcommand{\wra}{\rightharpoonup}
\newcommand{\inv}[1]{\frac{1}{#1}}
\newcommand{\vv}[1]{\bpa{\ssum #1_i}} 
\newcommand{\ssum}{\begingroup\textstyle\sum\endgroup_i}
\newcommand{\ketbra}[2]{\left| #1 \right> \left< #2 \right|}
\def\Td{{\rm T}}

\newcommand{\exc}[1]{E^{(#1)}}
\newcommand{\ger}[1]{G_{\ro}^{(#1)}}

\newcommand{\lpi}{L^p+L^{\ii}}
\newcommand{\expa}[1]{^{(#1)}}

\newcommand{\cans}{\cS\ind{mix}^N} 
\newcommand{\canp}{\cS} 
 
\newcommand{\rodii}{\kappa} 
\newcommand{\come}{c_{\Omega}}
\newcommand{\wei}{\alpha} 
\newcommand{\weig}{\boldsymbol{\alpha}} 
\newcommand{\edi}[1]{E^{(#1)}\ind{dis}} 
\newcommand{\matix}{\cM_{\vp}} 
\newcommand{\matixx}[1]{\cM_{\vp,#1}} 
 
\newcommand{\gr}{Q} 
\newcommand{\hn}{H_N}

\newcommand{\fll}[1]{F^{\weig,(#1)}} 
\newcommand{\flln}[2]{F^{\weig_{#2},(#1)}}
\newcommand{\fl}[1]{F\ind{mix}^{\weig,(#1)}}
\newcommand{\fln}[2]{F\ind{mix}^{\weig_{#2},(#1)}}
\newcommand{\gew}[1]{G_{r,\weig}^{(#1)}}
\newcommand{\kerr}{\Ker \bpa{\hn(v)-\exc{k}(v)}}
\newcommand{\kerrr}{\Ker_{\R} \bpa{\hn(v)-\exc{k}(v)}}
\newcommand{\kerw}{\Ker \bpa{\hn^{w=0}(v)-\exc{k}(v)}}
\newcommand{\kerrw}{\Ker_{\R} \bpa{\hn^{w=0}(v)-\exc{k}(v)}}

\newcommand{\spf}[1]{\cV_{N,\partial}^{(#1)}} 

\makeatletter\def\moverlay{\mathpalette\mov@rlay}
\def\mov@rlay#1#2{\leavevmode\vtop{%
   \baselineskip\z@skip \lineskiplimit-\maxdimen
   \ialign{\hfil$\m@th#1##$\hfil\cr#2\crcr}}}
\newcommand{\charfusion}[3][\mathord]{
    #1{\ifx#1\mathop\vphantom{#2}\fi
        \mathpalette\mov@rlay{#2\cr#3}
      }
    \ifx#1\mathop\expandafter\displaylimits\fi}
\makeatother
\newcommand{\cupdot}{\charfusion[\mathbin]{\cup}{\cdot}}

\usepackage{ifthen}\newboolean{versionlongue}\setboolean{versionlongue}{false} 
\usepackage{}
\usepackage{scalerel,stackengine,wasysym} 

\title[Building inverse potentials]{Building Kohn-Sham potentials\\for ground and excited states}
\author[L. Garrigue]{Louis Garrigue}
\address{CERMICS, \'Ecole des ponts ParisTech, 6 and 8 av. Pascal, 77455 Marne-la-Vallée, France} 
\email{louis.garrigue@enpc.fr}
\date{\today}

\begin{document}
\begin{abstract}
	We analyze the inverse problem of Density Functional Theory using a regularized variational method. First, we show that given $k$ and a target density $\rho$, there exist potentials having $k^{\text{th}}$ bound mixed states which densities are arbitrarily close to $\rho$. The state can be chosen pure in dimension $d=1$ and without interactions, and we provide numerical and theoretical evidence consistently leading us to conjecture that the same pure representability result holds for $d=2$, but that the set of pure-state $v$-representable densities is not dense for $d=3$. Finally, we present an inversion algorithm taking into account degeneracies, removing the generic blocking behavior of standard ones.
\end{abstract}

\maketitle
\flushbottom
\setcounter{tocdepth}{1} 

\section{Introduction}

In 1965, Kohn and Sham postulated the existence of effective one-body potentials which would replace the electronic interaction while keeping the same ground state density, and stated their relations to the exchange-correlation functionals \cite{KohSha65}. Physical quantities of this new effective non-interacting system provide approximations of the exact ones. This led to the developement of very successful techniques enabling to predict properties of microscopic systems in quantum chemistry and physics. The existence of such a potential producing a prescribed ground state density $\ro$ is called the $v$-representability problem, and its search is the inverse problem of Density Functional Theory. There are few works addressing the mathematical aspects of this problem, although several numerical studies were carried out. In \cite{Lieb83b}, Lieb proved that any density can be \textit{approximately} represented by a ground mixed state in some external potential $v$, and introduced a dual variational method enabling to find the Kohn-Sham potential. The ground state $v$-representability problem was studied by variational methods in the cases of classical DFT at positive temperature \cite{ChaLie84} and for quantum lattices \cite{ChaChaRus85}. 

In this document, we address the problem of $v$-representability in the quantum case at zero temperature, with ground or excited states, in pure and mixed settings, both theoretically and numerically. 

In the first part, we present a mathematical investigation. As shown by Lieb \cite{Lieb83b}, the exact inverse potential of a density $\ro \ge 0$ maximizes the functional $v \mapsto \exc{0}(v) - \int_{\R^d} v \ro$ in $(L^p+L^{\ii})(\R^d,\R)$, where $\exc{0}(v)$ is the $N$-particle ground state energy and $p$ is defined in \eqref{dims}. Nevertheless, this functional is not locally coercive in this space as we will see in Section \ref{boil}. To circumvent this ill-posedness, we regularize the problem by discretizing the space of potentials, more precisely we restrict our attention to potentials of the form $\sum_i v_i \wei_i$ where the $\wei_i$ are fixed weight functions and $v_i \in \R$ are real parameters. The discretization amounts to integrating the problematic short-distance degrees of freedom, and implements an ultra-violet cut-off. Our approach enables to show that for any $k \in \N := \{0,1,\dots\}$, the regularizations of the functionals $v \mapsto \exc{k}(v) - \int_{\R^d} v \ro$ are coercive, where $\exc{k}(v)$ denotes the $k\expo{th}$ bound state energy, and this implies the approximate representability of densities by $k\expo{th}$ mixed bound states, with arbitrary precision. For $(d,w)=(1,0)$, where $w$ is the two-body interaction potential, we show that we can take a pure state. Correspondingly, we define pseudo-discrete regularized Levy-Lieb and Lieb functionals by relaxing the condition $\ro_{\p} = \ro$ to $\int_{\R^d} \ro_{\p}\wei_i = \int_{\R^d} \ro \wei_i$ for any $i$, a similar approach was applied to optimal transport in \cite{AlfCoyEhrLom19,CoyaudPhD}, where the numerical efficiency seems promising. 

Computing inverse potentials is used in the Optimal Effective Potential method to develop exchange-correlation functionals which perform better than standard functionals in some configurations \cite[Chapter 6]{EngDre11}. Numerically, this problem received significant attention for $k=0$, in \cite{WuYan03,GonDan18,JenWas18,SchNeu18,KanZimGav19,NaiOhaLia19,KumSinHar19,CalLatGid20,AccBraMar20,WagBakSto14,MorCarGeo20} using the dual formulation, in \cite{JenWas18} using the PDE-constrained optimization, and in \cite{PenLaeTelRug19} using derivatives of the Moreau-Yosida regularized Levy-Lieb functional, where degeneracy is discussed in \cite{ErrPenLaeTel20}. However, degeneracy issues are also critical in the dual approach, except when $(d,k)=(1,0)$, and they were not taken into account in the literature to the best of our knowledge. Indeed, the standard algorithm breaks down when eigenvalues cross or when the inverse potential is degenerate. Hence in the second part, we present an algorithm which converges to a potential having a $k\expo{th}$ bound mixed state with the target density. With $w=0$, numerical results indicate that for $d=2$, densities are $v$-representable by $k\expo{th}$ bound pure states, whereas this depends on the target density for $d = 3$. We also numerically remark that degeneracies are generic for inverse potentials, and as in the SCF procedure \cite{CanMou14} that perturbation of target densities does not lift degeneracies. Finally, we confirm the study \cite{GauBur04}, which indicates that for excited states, many potentials lead to the same density.


\subsection*{Acknowledgement}
I warmly thank Mathieu Lewin, for having advised me during this work, and \'Eric Cancès for useful comments. This project has received funding from the European Research Council (ERC) under the European Union's Horizon 2020 research and innovation programme (grant agreements MDFT No 725528 and EMC2 No 810367). Data sharing not applicable to this article as no datasets were generated or analysed during the current study.

\section{Properties of the dual problem}\label{repsm}
\subsection{Definitions}\label{defss}
Let $d \in \N \backslash \acs{0}$ and in all the document, $\Omega \subset \R^d$ denotes a (bounded or unbounded) connected open set with Lipschitz boundary, representing the space in which our quantum system lives. We do not consider spin degrees of freedom but our results can be extended in this way without complications. We define
\begin{equation}\label{dims}
p = 1 \tx{ if } d = 1, \bhs p > 1 \tx{ if } d = 2, \bhs p = d/2 \tx{ if } d \ge 3.
\end{equation}
In all this work, we consider an even non-negative interaction potential $w \in (L^{p}+L^{\ii})(\R^d,\R_+)$, where we recall that $\pa{L^p + L^\ii}(\Omega)$ is the Banach space of functions $f = f_p + f_\infty$, where $f_p \in L^p(\Omega)$ and $f_\infty \in L^\infty(\Omega)$, endowed with the norm
\begin{align*}
\nor{f}{\lpi(\Omega)} = \myinf{f_p \in L^p(\Omega), f_\infty \in L^\infty(\Omega) \\ f_p + f_\infty = f} \nor{f_p}{L^p(\Omega)} + \nor{f_\infty}{L^\infty(\Omega)}.
\end{align*}
We take external electric potentials $v \in (L^{p}+L^{\ii})(\Omega,\R)$, and consider the self-adjoint $N$-particle Schr\"odinger operator 
\begin{equation}\label{exc}
H_N(v) := \sum_{i=1}^N -\Delta_i  + \sum_{1\leq i < j \leq N} w(x_i-x_j) + \sum_{i=1}^N v(x_i),
\end{equation}
acting on the antisymmetric $N$-particle space $L^2\ind{a}(\Omega^N) := \wedge^N L^2(\Omega)$ with homogeneous Dirichlet boundary conditions. The one-body density of a state $\p \in L^2\ind{a}(\Omega^N)$ is defined as
\begin{align*}
\ro_{\p}(x) := N \int_{\Omega^{N-1}} \ab{\p}^2(x,x_2,\dots,x_N)\d x_2 \cdots \d x_N.
\end{align*}
For vector subspaces $A \subset L\ind{a}^2(\Omega^N)$, we define the sets of mixed states
\begin{align*}
\cans(A,\Omega) := \sch_1\bpa{L\ind{a}^2(\Omega^N)} \cap \acs{\Gamma = \Gamma^*\ge 0, \tr (-\Delta) \Gamma < + \ii, \restr{\Gamma}{A^{\perp}} = 0},
\end{align*}
where $\sch_1(B)$ is the space of trace-class operators on the vector space $B$, and when $\Omega \neq \R^d$, $\tr (-\Delta) \Gamma := \tr (-\Delta_D)^{\ud} \Gamma (-\Delta_D)^{\ud}$, where $-\Delta_D = -\Delta$ is the Dirichlet Laplacian. We will also use $\cans(\Omega) := \cans\bpa{L\ind{a}^2(\Omega^N),\Omega}$. The one-body density of such a mixed state $\Gamma$ is
\begin{align*}
\ro_{\Gamma}(x) := N \int_{\Omega^{N-1}} \Gamma (x,x_2,\dots,x_N ; x,x_2 ,\dots,x_N) \d x_2 \cdots \d x_N,
\end{align*}
where $\Gamma(x_1,\dots,x_N ; y_1,\dots,y_N)$ is the integral kernel of the operator $\Gamma$.

Let us denote by $\cE_v(\p) := \ps{\p,\hn(v)\p}$ the energy functional for pure states, and by $\cE_v(\Gamma) := \tr \hn(v) \Gamma$ the one for mixed states. We recall \cite[Section 12.1]{LieLos01} that the ground ($k=0$) and excited ($k \ge 1$) energies are
\begin{align}\label{defen}
	&\exc{k}(v)  = \hspace{-0.1cm}\mysup{A \subset L\ind{a}^2(\Omega^N) \\ \dim_{\C} A = k} \myinf{\p \in A^{\perp} \\ \int_{\Omega^N} \ab{\p}^2 = 1 \\ \p \in H^1\ind{a}(\Omega^N)} \hspace{-0.3cm} \ps{\p,\hn(v)\p} =\hspace{-0.1cm} \myinf{A \subset L\ind{a}^2(\Omega^N) \\ \dim_{\C} A = k+1} \mymax{\p \in A \\ \int_{\Omega^N} \ab{\p}^2 = 1 \\ \p \in H^1\ind{a}(\Omega^N)} \hspace{-0.5cm} \ps{\p,\hn(v)\p} \nonumber \\
	& \hs\hs\hs\hs\hs = \mysup{A \subset L\ind{a}^2(\Omega^N) \\ \dim_{\C} A = k} \myinf{\Gamma \in \cans(A^{\perp},\Omega) \\ \tr \Gamma = 1} \tr \hn(v) \Gamma =\hspace{-0.25cm} \myinf{A \subset L\ind{a}^2(\Omega^N) \\ \dim_{\C} A = k+1} \mymax{\cans(A,\Omega) \\ \tr \Gamma = 1} \tr \hn(v) \Gamma,
 \end{align}
where $H^1\ind{a}(\Omega^N) := L^2\ind{a}(\Omega^N) \cap H^1_0(\Omega^N)$. We also define the open set of potentials able to bind $N$ particles in a ground ($k=0$) or a $k\expo{th}$ excited state
\begin{align*}
	\spf{k} := \acs{v \in \pa{L^p+L^{\ii}}(\Omega) \bigst \exc{k}(v) < \inf \sigma\ind{ess}(\hn(v))},
\end{align*}
see \cite{Garrigue21} for more properties on it. If $v \in \spf{k}$ and $\Gamma$ is an optimizer of \eqref{defen}, we say that it is a $k\expo{th}$ bound mixed state, and then it is supported on the $k\expo{th}$ eigenspace, that is $\ran \Gamma \subset \Ker\pa{\hn(v)-\exc{k}(v)}$.

 \subsection{The primal problem}
For $\ro \ge 0$ with $\int_\Omega \ro = N$ and $\sqrt{\ro} \in H^1(\Omega)$, the exact ground and excited Levy-Lieb (or pure) and Lieb (or mixed) functionals \cite{Levy79,Lieb83b,Lieb85,LewLieSei20}, are
 \begin{align}\label{lly}
	 F\ex{k}(\ro) & := \mysup{A \subset H^1\ind{a}(\Omega^N) \\ \dim_{\C} A = k}  \myinf{\p \in A^{\perp} \\ \ro_{\p} = \ro} \ps{ \p, H_N\pa{0} \p}, \\
	 F\ex{k}\ind{mix}(\ro)&  := \mysup{A \subset H^1\ind{a}(\Omega^N) \\ \dim_{\C} A = k}  \myinf{\Gamma \in \cans(A^{\perp},\Omega) \\ \ro_{\Gamma} = \ro} \tr H_N\pa{0} \Gamma. \nonumber
 \end{align}
For any $k$, $F\ex{k}\ind{mix}$ is convex and lower semi-continuous \cite{Lieb85} on $L^1 \cap L^{p'}(\Omega)$, where $p' := p/(p-1) \in \R \cup \acs{+\ii}$. Moreover, we have $F\ex{k}\ind{mix} \le F\ex{k}$, $F\ex{k}\ind{mix} \le F\ind{mix}\ex{k+1}$ and $F\ex{k} \le F\ex{k+1}$. We know that in the ground state case $k=0$, they are finite and have optimizers \cite{Lieb83b}, and that they enable to compute the ground state energy in the sense that
 \begin{align*}
	 \exc{0}(v)\hspace{-0.1cm} =\hspace{-0.1cm} \myinf{\ro \in L^1(\Omega,\R_+) \\ \sqrt{\ro} \in H^1(\Omega) \\  \int_\Omega \ro =N} \pa{F\ex{0}(\ro) + \int_\Omega v\ro}\hspace{-0.1cm}= \hspace{-0.1cm} \myinf{\ro \in L^1(\Omega,\R_+) \\ \sqrt{\ro} \in H^1(\Omega) \\ \int_\Omega \ro =N} \pa{F\ind{mix}\ex{0}(\ro) + \int_\Omega v\ro}.
 \end{align*}
 As noted by Lieb in \cite{Lieb85}, for $k \ge 1$ we cannot recover $\exc{k}(v)$ by minimizing $\ro \mapsto F\ex{k}(\ro) + \int_\Omega v \ro$ or $F\ex{k}\ind{mix}(\ro)+ \int_\Omega v \ro$, or even any such functional of $\ro$, because this would lead to a convex functional of $v$ while $\exc{k}$ is not so.

 \subsection{The dual problem}

Take $\ro \in L^1(\Omega,\R_+)$ such that $\sqrt{\ro} \in H^1(\Omega)$ and $\int_\Omega \ro = N$, and take $k \in \N$. For $v \in (L^p+L^{\ii})(\Omega)$ where $p$ is as in \eqref{dims}, the dual functional is
 \begin{align*}
	 \ger{k}(v) := \exc{k}(v) - \int_{\Omega} v \ro, \hs\hs\hs\hs\hs \tx{verifying} \hs\hs\hs\hs 
 \mysup{v \in (L^p+L^{\ii})(\Omega,\R)} \ger{k}(v) = F\ind{mix}\ex{k}(\ro),
 \end{align*}
as showed in \cite{Lieb85}. To prepare the approximate representability of ground and excited densities by potentials, we want to explore the maximization of $\ger{k}$.

\subsection{Solution of the local dual problem}
The potential-to-energy map $v \mapsto \exc{k}(v)$ is neither Fréchet nor Gâteaux differentiable, but it is Dini differentiable, as presented in \cite[Theorem 1.6]{Garrigue21} and \eqref{expg}. Let us denote by ${^+}\delta_v \ger{k}(u)$ the Dini differential in the direction $u$, which is the right derivative of $\lambda \mapsto \ger{k}(v + \lambda u)$ at $\lambda = 0$. The local first order problem is to find the optimal direction(s) in which the functional $\ger{k}$ increases the most, that is solving 
 \begin{align*}
\mysup{u \in (L^p+L^{\ii})(\Omega,\R) \\ \nor{u}{\lpi}=1} {^+}\delta_v \ger{k}(u).
 \end{align*}
 For our analysis, we will see in the proof of Theorem \ref{kspropi} that it will be sufficient to find the maximizing direction $u \in L^p$ of the problem
 \begin{align*}
\mysup{u \in L^p(\Omega,\R) \\ \nor{u}{L^p}=1} {^+}\delta_v \ger{k}(u).
 \end{align*}
We now show that this linearized problem can be \apo{solved}, that is transformed into a simple low-dimensional problem. For a finite dimensional real vector space $Q_{\R} \subset L^2\ind{a}(\Omega^N,\R)$ formed by real-valued wavefunctions, real mixed states with range in $Q_{\R}$ will be identified with symmetric matrices in
 \begin{align*}
	 \cS(Q_{\R}) := \acs{\Gamma \in  \R^{\dim Q_{\R} \times \dim Q_{\R}}, \Gamma = \Gamma^{\tx{T}}}.
 \end{align*}
 For $k \in \N$ and potentials $v \in \spf{k}$, we define the integers $m_k^v, M_k^v \in \N$ by
\begin{align}\label{cdff}
 \exc{m_k^v - 1}(v) < \exc{m_k^v}(v) = \dots = \exc{k}(v) = \dots = \exc{M_k^v}(v) < \exc{M_k^v+1}(v),
 \end{align}
with $\exc{-1}(v) := -\ii$ by convention. We denote by $\Ker_{\R} (\hn(v)-\exc{k}(v))$ the real vector space of real $k\expo{th}$ bound eigenfunctions.


\begin{proposition}[The local problem]\label{sollin}
Take $\Omega$ an open connected domain with Lipschitz boundary. Let $p$ be as in \eqref{dims}, take $w \in \pa{L^p+L^{\infty}}(\R^d)$, $w \ge 0$, $\ro \in L^1(\Omega)$ such that $\sqrt{\ro} \in H^1(\Omega)$ and $\int_\Omega \ro = N$. Take $s \ge p$ with $s > 1$, and take $v \in \spf{k}$. Then $\ro \in L^{\f{s}{s-1}}(\Omega)$ and we have
\begin{align}\label{lao}
	\mysup{u \in L^s(\Omega,\R) \\ \nor{u}{L^s}=1} {^+}\delta_v \ger{k} (u) & = \mymax{Q \subset \Ker_{\R} (\hn(v)-\exc{k}(v)) \\ \dim_{\R} Q = M_k^v-k+1} \mymin{\Gamma \in \canp\pa{Q} \\\Gamma \ge 0, \tr \Gamma =1} \nor{\ro_{\Gamma} - \ro}{L^{\f{s}{s-1}}(\Omega)},
 \end{align}
 and the supremum is attained by
\begin{align}\label{maxu}
u^* =\ab{\f{\ro_{\Gamma^*}-\ro}{\nor{\ro_{\Gamma^*}-\ro}{L^{\f{s}{s-1}}(\Omega)}}}^{\f{1}{s-1}} \sgn(\ro_{\Gamma^*}-\ro), 
 \end{align}
where $\Gamma^*$ is an optimizer of the right hand side of \eqref{lao}.
\end{proposition}
A proof is provided in Section \ref{ssub:proofs_sollin_kspropi}, also, see \eqref{ineq:sob_inj} for a presentation of the Sobolev injections at stake for $\ro$. In particular, when $\kerr$ is non-degenerate, we call $\p^{(k)}(v)$ the $k^{\text{th}}$ eigenfunction of $H_N(v)$ (unique up to a change of phasis), and in this case $\ger{k}$ is Fréchet differentiable at $v$ and the problem
\begin{align*}
\mysup{u \in L^s(\Omega,\R) \\ \nor{u}{L^s}=1} \d_v \ger{k} (u) = \nor{\ro_{\p\ex{k}(v)} - \ro}{L^{\f{s}{s-1}}(\Omega)}
\end{align*}
is solved in the unique direction \eqref{maxu}, where $\ro_{\Gamma^*}=\ro_{\p\ex{k}(v)}$.

\subsection{Optimality in the dual problem}
Next, we analyze the optimality conditions. For our problem of searching a potential producing a prescribed density, the following result shows that we have to search among the maximizers of $\ger{k}$.

\begin{theorem}[Optimality in the dual problem]\label{kspropi}
Take $\Omega$ an open connected domain with Lipschitz boundary. Take $w \ge 0$, take a density $\ro \in L^1(\Omega)$, $\ro \ge 0$, $\int_\Omega \ro = N$, $\sqrt{\ro} \in H^1(\Omega)$, and consider a binding $v \in \spf{k}$.

	$i)$ The following assertions are equivalent

	\bhs $a)$ there is a $k^{\tx{th}}$ bound mixed state $\Gamma$ of $v$ such that $\ro_{\Gamma} = \ro$
	
	\bhs $b)$ $v$ is a local maximizer of $\ger{k}$

	\bhs $c)$ $v$ is a global maximizer of $\ger{k}$

	$ii)$ If $v$ maximizes $\ger{k}$, then it maximizes $\ger{\ell}$ for all $\ell \in \acs{m_k^v, \dots,k}$. Moreover, if $k \ge (M_k^v+m_k^v)/2$, then $\ro_{\p} = \ro$ for any normalized $\p \in \Ker \bpa{\hn(v)-\exc{k}}$. If $v$ is a local minimizer, then $k > (M_k^v+m_k^v)/2$.

	$iii)$ If $v$ maximizes $\ger{k}$ and $\dim \Ker \bpa{\hn(v)-\exc{k}(v)} \in \acs{1,2}$, then $v$ has a $k\expo{th}$ bound pure state $\p$ such that $\ro_{\p} = \ro$.

	$iv)$ For $d = 1$ and $w=0$, if $v$ maximizes $\ger{k}$, then there exists a pure state $\p \in \kerr$ such that $\ro_{\p} = \ro$.
\end{theorem}
A proof is provided in Section \ref{ssub:proofs_sollin_kspropi}. In $iii)$, we take $\exc{\ell}(v) := -\ii$ for $\ell \le -1$ by convention. When $k=0$ and $p>\max(2d/3,2)$, the maximizer is unique by the Hohenberg-Kohn theorem, and the equivalences do not need to assume $v \in \spf{0}$. Since pure states are mixed states, when we search pure states of $v$ representing $\ro$, we also need to maximize $\ger{k}$. Once the set $\cV\expo{max}_{\ro,k}$ of maximizers is found, one can finally compute
\begin{align}\label{fpure}
	\myinf{v \in \cV\expo{max}_{\ro,k}} \hspace{0.2cm} \mymin{\p \in \kerr \\ \int_{\Omega^N} \ab{\p}^2=1} \nor{\ro_{\p} - \ro}{L^{\f{p}{p-1}}(\Omega)},
\end{align}
which vanishes if and only if $\ro$ is pure-state representable.

Moreover, we conjecture that $\ger{k}$ has no local minimum and that for maximizing $v$'s, $k = m_k^v$, although we were not able to prove this result.

 In case where $w=0$ and $k=0$, the maximizing potential $v =: v\ind{ks}(\ro)$ is called the (mixed) Kohn-Sham potential \cite{KohSha65} for $\ro$. We use the term \apo{Kohn-Sham potential} for any $k \in \N$. When $d=3$, $k=0$ and $\ro$ is a ground state density of $\hn(u)$ for the Coulomb interaction $w=\ab{\cdot}^{-1}$, then $v\ind{ks}(\ro) -u-\ro * \ab{\cdot}^{-1}$ is called the exchange-correlation potential in the physics and quantum chemistry literature.

\subsection{Lower bound}%
\label{sub:lower_bound}

We remark that for $w \ge 0$, we have
\begin{align}\label{ineq:lower_bound_G}
- \int_{\Omega} v_+ \ro - L_{1,d} \int_{\Omega} \ab{v_-}^{1+\f d2} \le \ger{k}(v)
\end{align}
where $L_{1,d}$ is the Lieb-Thirring constant \cite{LieThi76}. This comes from the fact that $\exc{k}(v) \ge \exc{k}_{w=0}(v)$, then writing $\exc{k}_{w=0}(v)$ as a sum on negative eigenvalues and finally using the Lieb-Thirring inequality.

\subsection{Ill-posedness of the dual problem}
Let us search for a maximizing potential of $\ger{0}$. We consider a maximizing sequence $v_n$. If we are able to prove that $v_n$ converges weakly to some $v \in (L^p+L^{\ii})(\Omega)$, $v$ would be a maximizer by weak upper semi-continuity of $\ger{0}$. But we are not even able to prove that $v_n$ weakly converges locally.

Moreover we now show that $\ger{k}$ is ill-posed in the sense that it is not locally coercive in $L^p(\R^d)$ spaces, where we take $\Omega = \R^d$. Take $\ro \in L^1(\R^d,\R_+)$ having mass $\int_{\R^d} \ro = N$ be a target density which we want to represent by a potential. First of all, as a consequence of \cite[Theorem 3.8]{Lieb83b} and of $\ger{0} \le \ger{k}$, one needs to assume that $\sqrt{\ro} \in H^1(\R^d)$, otherwise $\ger{k}$ is not bounded from above. Now take $\alpha \ge 0$, $p \ge 1$, $\ro$ continuous at the origin, and a potential $v \in (L^1 \cap L^p)(\R^d,\R)$ with compact support. Consider the sequence $v_n(x) := n^{\alpha} v(nx)$, then $\nor{v_n}{L^p(\R^d)} = n^{\alpha - \f{d}{p}} \nor{v}{L^p(\R^d)}$ and $n^{d- \alpha}\int_{\R^d} v_n \ro \rightarrow \ro(0) \int_{\R^d} v$. We take $\alpha > d/p$ so that $\nor{v_n}{L^p(\R^d)} \rightarrow +\infty$, recall that we also need $v \in (L^q+L^{\infty})(\R^d,\R)$ where $q$ is as in \eqref{dims}. We then provide two kinds of counterexamples, the first one is when $v \ge 0$, and the second one when $v \le 0$.

When $v \ge 0$, $\exc{k}(v_n) = 0$ hence $\ger{k}(v_n) = - \int_{\R^d} v_n \ro$, we want $\alpha \ge d$ so that $\ger{k}(v_n)$ remains bounded, hence we take $\alpha = d$ and $p>1$. 

When $v \le 0$, \eqref{ineq:lower_bound_G} becomes $\ger{k}(v_n) \ge - L_{1,d} n^{\alpha\pa{1+\f d2} -d} \int_{\R^d} \ab{v}^{1 + \f d2}$ and we choose $\alpha = d/\pa{1+\f d2}$ and $p > 1 + \f d2$ so that $\ger{k}(v_n)$ remains bounded.

\section{Regularization}
We saw in Theorem \ref{kspropi} that to $v$-represent a density $\ro$ with pure or mixed states, we need to maximize $\ger{k}$, but we also saw that this problem is ill-posed in $L^p$ spaces. Hence we regularize it in this section, which will make it coercive.

\subsection{Pseudo-discrete regularizations of Levy-Lieb and Lieb functionals}


 We now relax the density constraint. Let us consider a subset $I \subset \N$ and a set $\weig = \pa{\wei_i}_{i \in I}$ of weight functions forming a partition of unity for $\Omega$, that is $\sum_{i \in I} \wei_i = \indic_{\Omega}$, where $\wei_i \in L^{\ii}(\Omega,\R_+)$. For $r \in \ell^1(I,\R_+) \cap \acs{ \sum_{i \in I} r_i = N}$, we introduce the regularized Levy-Lieb and Lieb functionals
 \begin{align*}
	 \fll{k}(r) & := \mysup{A \subset H^1\ind{a}(\Omega^N) \\ \dim_{\C} A = k}  \myinf{\p \in A^{\perp} \\ \int_\Omega \wei_i \ro_{\p} = r_i \hs \forall i \in I} \ps{ \p, H_N\pa{0} \p}, \\
	 \fl{k}(r)&  := \mysup{A \subset H^1\ind{a}(\Omega^N) \\ \dim_{\C} A = k}  \myinf{\Gamma \in \cans(A^{\perp},\Omega) \\  \int_\Omega \wei_i \ro_{\Gamma} = r_i \hs \forall i \in I} \tr H_N\pa{0} \Gamma,
 \end{align*}
and we define them to be $+\ii$ when for any $A\subset H^1\ind{a}(\Omega^N)$ such that $\dim A = k$, the minimizing sets are empty. We know that $\fl{0}$ is convex \cite{Lieb83b}. Consider now the assumption
\begin{align}\label{tigh}
\mylim{R \ra +\ii} \mysum{i \in I \\ \supp \wei_i \cap B_R^{\tx{c}} \neq \empt}{} r_i = 0.
 \end{align}

\begin{theorem}[Existence of minimizers in the ground states case]\label{exiss}
Take $\Omega$ an open connected domain with Lipschitz boundary.	Take $w \ge 0$ and $r \in \ell^1(I,\R_+)$ such that $\sum_{i \in I} r_i = N$, and $\weig = (\wei_i)_{i \in I}$ a partition of unity for $\Omega$. Under the tightness condition \eqref{tigh}, $\fll{0}(r)$ and $\fl{0}(r)$ have at least one minimizer when they are finite.
\end{theorem}
A proof is provided in Section \ref{ssub:proofs_exiss_proproi}. For a given $\ro \in L^1(\Omega,\R_+)$, we define
\begin{align*}
r_{\ro} := \pa{\int_\Omega \wei_i \ro}_{i \in I} \in \ell^1(I,\R_+).
\end{align*}
This sequence contains the partial information on the density $\ro$ which we are going to retain. Since the optimizing set in the definition of $F\ex{k}(\ro)$ is included in the one of $\fll{k}(r_{\ro})$, for any $\ro \ge 0$ with $\sqrt{\ro} \in H^1(\Omega)$ and $\weig$ as defined above we have
\begin{align*}
	\fll{k}(r_{\ro}) \le F\ex{k}(\ro) \bhs  \tx{ and } \bhs \fl{k}(r_{\ro}) \le F\ex{k}\ind{mix}(\ro).
\end{align*}
In particular, $\fll{0}(r_{\ro})$ and $\fl{0}(r_{\ro})$ are finite.

For $k=0$, our approximate Levy-Lieb and Lieb functionals converge to the exact ones when the integrated weights tend to carry all the information on the density. 

\begin{theorem}[Convergence to the exact model]\label{proproi} Take $\Omega \subset \R^d$ an open connected domain with Lipschitz boundary. Take $w \in (L^p+L^{\ii})(\R^d)$, $w \ge 0$, with $p$ as in \eqref{dims}. Consider a density $\ro \in L^1(\Omega,\R_+)$ such that $\sqrt{\ro} \in H_0^1(\Omega)$ and $\int_\Omega \ro = N$. We assume that $\weig_n=  (\wei_i^n)_{i \in I_n}$, where $\wei_i^n \in L^{\ii}(\Omega)$, is a sequence of weights forming a partition of unity for $\Omega$, and such that for any $f \in \cC^{\ii}\ind{c}(\Omega)$, we have
\begin{align}\label{hypo}
	\myinf{ g_n \in \vect ( \wei_i^n)_{i \in I_n}} \nor{f - g_n}{\pa{L^{p}+L^{\ii}}(\Omega)} \longrightarrow 0
\end{align}
when $n \ra +\ii$. We also assume that
\begin{align}\label{tig}
\mylim{R \ra +\ii} \; \mysup{n \in \N} \; \mysum{i \in I \\ \supp \wei_i^n \cap B_R^{\tx{c}} \neq \empt}{} \int_\Omega \ro \wei_i^n = 0.
 \end{align}
Then
\begin{align*}
	\mylim{n \ra +\ii} \flln{0}{n}\pa{r_{\ro}} = F\ex{0}(\ro), \bhs \mylim{n \ra +\ii} \fln{0}{n}\pa{r_{\ro}} = F\ind{mix}\ex{0}(\ro).
\end{align*}
	Let $\p_n$ be a sequence of approximate minimizers for $\flln{0}{n}(r_{\ro})$, that is, such that $\cE_0(\p_n) \le \flln{0}{n}(r_{\ro}) + \ep_n$ where $\ep_n \ra 0$ when $n \ra +\ii$ and $\int_\Omega \wei_i \ro_{\p_n} = \int_\Omega \wei_i \ro$ for any $i \in I$. Then $\p_n \ra \p\ind{exact}$ strongly in $H^1(\Omega^N)$ up to a subsequence, where $\p\ind{exact}$ is a minimizer for $F\ex{0}(\ro)$. If $\Gamma_n$ is a sequence of approximate minimizers for $\fln{0}{n}\pa{r_{\ro}}$, then $\Gamma_n \ra \Gamma\ind{exact}$ strongly in the kinetic energy space $\sch_{1,1}$ up to a subsequence, where $\Gamma\ind{exact}$ is a minimizer for $F\ind{mix}(\ro)$.
\end{theorem}
A proof is provided in Section \ref{ssub:proofs_exiss_proproi}. The space $\sch_{1,1}$ is the set of operators $A$ of $L\ind{a}^2(\Omega^N)$ endowed with the norm $\nor{A}{\sch_{1,1}} = \tr \ab{(-\Delta_D+1)^{\ud} A (-\Delta_D+1)^{\ud}}$.

Our assumption \eqref{tig} is used to control the decay at infinity. If all the $\wei_i^n$ have a compact support of diameter bounded by $\delta$ independent of $i$ and $n$, then 
 \begin{align*}
	 \mysum{i \in I \\ \supp \wei_i^n \cap B_R^{\tx{c}} \neq \empt}{} \int_\Omega \ro \wei_i^n \le \int_{\Omega \cap \{\ab{x} \ge R - \delta\}} \ro \us{R \ra +\ii}{\longrightarrow} 0
 \end{align*}
 and \eqref{tig} is satisfied. In \cite[(3.3.4)]{AlfCoyEhrLom19}, the authors use an inequality condition, simpler than \eqref{tig}.

If $\wei_i^n = \indic_{\Omega_i^n}$ is a sequence of partitions of $\Omega = \cup_{i \in \N} \Omega_i^n$ where $\Omega_i^n$ are convex, and $\sup_{i \in \N} \diam \Omega_i^n \ra 0$ when $n \ra +\ii$, then the assumption \eqref{hypo} is verified by Lemma~\ref{lemama} below. Assumption \eqref{tig} is verified as well. Again by Lemma~\ref{lemama}, if we further assume that $\ro$ is Lipschitz continuous, we have an explicit bound on the convergence of densities
\begin{align*}
	\nor{\ro-\ro_{\p_n}}{\bpa{L^1 \cap L^{q}}(\Omega)} \le c_d  \pa{ \nor{\sqrt{\ro}}{H^1(\Omega)}^2 + \sup_{n \in \N} \nor{\sqrt{\ro_{\p_n}}}{H^1(\Omega)}^2}\sup_{i \in \N} \diam \Omega_i^n,
\end{align*}
where $c_d$ only depends on $d$, and $q$ is as in \eqref{cdc}. Note that $\nor{\sqrt{\ro_{\p}}}{H^1(\Omega)}$ and $\nor{\sqrt{\ro_{\p_n}}}{H^1(\Omega)}$ are controlled by $F(\ro)$ due to the Hoffmann-Ostenhof inequality.

A typical choice for the $\wei_i^n$ is given by the partition of unity finite element method \cite{MelBab96,BabMel97}.

\subsection{Regularization of the dual problem}\label{lama}

Correspondingly to the previous part, we change the exact model by discretizing the space of potentials. We consider a sequence of weights $\weig = (\wei_i)_{i \in I}$ and take $r \in \ell^1(I,\R_+)$. The dual problem is the maximization of
\begin{align*}
	G\ex{k}_{r,\weig}(v):= \ger{k}\pa{\sum_{i \in I} v_i \wei_i} = \exc{k}\pa{\sum_{i \in I} v_i \wei_i} - \sum_{i \in I} v_i r_i,
\end{align*}
over the space $\ell^{\ii}(I,\R)$ of potential coefficients $v = (v_i)_{i \in I}$. We have
\begin{align*}
	\exc{0}\pa{\sum_{i \in I} v_i \wei_i} & = \myinf{ r \in \ell^1(I,\R_+) \\ \sum_{i \in I} r_i = N}  \pa{ \fl{0}(r) + \sum_{i \in I} v_i r_i}, \\
	\mysup{v \in \ell^{\ii}(I,\R)}  G\ex{k}_{r,\weig}(v) & = \fl{k}(r),
\end{align*}
as in the exact models, and by the same proofs. Again by the same proof as for the lower semi-continuity of the exact Lieb functional \cite[Theorem 3.6]{Lieb83b}, $G\ex{0}_{r,\weig}$ is weakly upper semi-continuous in the $\ell^{\ii}(I,\R)$ topology. Moreover, if $H_N\pa{\sum_{i \in I} v_i \wei_i}$ has a $k\expo{th}$ bound state $\p_v$, then
\begin{align*}
G\ex{k}_{r,\weig}(v) = \cE_0\pa{\p_v} + \sum_{i \in I} v_i\pa{- r_i + \int_\Omega \ro_{\p_v} \wei_i }.
\end{align*}

\subsubsection{Gauge invariance}

The gauge we are dealing with is the choice of a reference for energies, corresponding to the transformation $V \ra V + c$ for a constant $c \in \R$. The exact dual functional $V \mapsto \exc{k}(V) - \int_\Omega V \ro$ is gauge invariant, and since we want our approximate functional to be so as well, we are naturally led to take 
 \begin{align}\label{cdt}
 \sum_{i \in I} \wei_i= 1 \tx{ on } \Omega, \bhs\bhs \sum_{i \in I} r_i = N.
 \end{align}
 The last condition is of course fulfilled for $r = r_{\ro}$, which is the interesting situation.

\begin{remark}
	Let us explain why the previous conditions \eqref{cdt} are necessary to ensure gauge invariance. Let $v \in \ell^{\ii}(I,\R)$ be such that $H_N\bpa{\sum_{i \in I} v_i \wei_i}$ has a $k\expo{th}$ bound state, which we denote by $\p_v$. Take $c \in \R$, we have
\begin{align*}
	\exc{k}\pa{\sum_{i \in I} (v_i+c) \wei_i} & \le \cE_{\Sigma_{i \in I} (v_i+c) \wei_i}(\p_v) \\
						  &= \exc{k}\pa{\sum_{i \in I} v_i \wei_i} \hspace{-0.1cm} + c\hspace{-0.1cm}\int_\Omega \sum_{i \in I} \wei_i \ro_{\p_v},
\end{align*}
and hence
\begin{align*}
G\ex{k}_{r,\weig}(v+c) \le G\ex{k}_{r,\weig}(v) + c \pa{- \sum_{i \in I} r_i +  \int_{\Omega}\ro_{\p_v} \sum_{i \in I} \wei_i}.
\end{align*}
To have a gauge invariant theory, we want to have 
 \begin{align*}
\int_{\Omega} \pa{N^{-1} \sum_{i \in I} r_i - \sum_{i \in I} \wei_i} \ro_{\p_v} = 0,
 \end{align*}
otherwise $G\ex{k}_{r,\weig}(v+c) \ra - \ii$ for $c \ra +\ii$ or $c \ra -\ii$. This requirement should not depend on $v$, hence we need $\sum_{i \in I} r_i = N \sum_{i \in I} \wei_i$ a.e on $\Omega$. We are thus naturally led to assume \eqref{cdt}.
\end{remark}

\subsubsection{Uniqueness}

A Hohenberg-Kohn theorem adapted to our situation shows that the multivalued map $v \mapsto r_{\ro_{\p_{v}}}$, where $\p_{v}$ is a ground state of $H_N\bpa{\sum v_i \wei_i}$, is essentially injective. Hence if $G\ex{0}_{r,\weig}$ has a maximum, it is unique.

\begin{theorem}[Hohenberg-Kohn]\label{ada}
Let $\Omega \subset \R^d$ be an open and connected set with Lipschitz boundary, and consider homogeneous Dirichlet boundary conditions. Let $p > \max(2d/3,2)$, and take an interaction $w \in (L^p+L^{\ii})(\R^d,\R)$. Let $v,u \in \ell^{\ii}(I,\R)$ and $\weig = (\wei_i)_{i \in I}$ where $\wei_i \in L^{\ii}(\Omega,\R_+)$, be such that $H_N\bpa{\sum_{i \in I} v_i \wei_i}$ and $H_N\bpa{\sum_{i \in I}u_i \wei_i}$ have at least one ground state each, which we respectively denote by $\p_v$ and $\p_u$. If $\int_\Omega \wei_i \ro_{\p_v} = \int_\Omega \wei_i \ro_{\p_u}$ for any $i \in I$, then $v = u +c$ for some constant $c \in \R$.
\end{theorem}
The proof follows from the standard Hohenberg-Kohn theorem \cite{HohKoh64,Garrigue18} in the form of \cite[Theorem 2.1]{Garrigue19}.

\subsubsection{Coercivity}\label{boil}

The main goal of this section is to recover coercivity for the discretized dual problem, in order to make it well-posed.

If there is some $i \in I$ such that $r_i = 0$, denoting by $e_i$ the $i^{\tx{th}}$ degree of freedom of the potentials, when $c \ra +\ii$ we expect that $G\ex{k}_{r,\weig}(v + c e_i) \ra \exc{k}_D\bpa{\sum_{j \neq i} v_j \wei_j}$, where $\exc{k}_D\bpa{\sum_{j \neq i} v_j \wei_j}$ is finite and is the $k\expo{th}$ bound state energy of the system living in $\Omega \backslash \supp \wei_i$ with Dirichlet boundary conditions. This shows that 
\begin{align*}
r_i > 0 \hs\hs\hs\hs \forall i \in I
\end{align*}
is a necessary condition for $G\ex{k}_{r,\weig}$ to be coercive. 

We define 
\begin{align}\label{ccc}
	c_{\Omega} := -\f{\exc{k}(0)}{N} \le 0,
\end{align}
where $\exc{k}(0)$ is the $k\expo{th}$ energy level of $N$ interacting particles without external potential. It satisfies $\exc{k}(c_{\Omega} \indic_{\Omega}) = 0$ and it is non-positive because $w \ge 0$. It vanishes when $\Omega = \R^d$ for instance.  
 We can choose the gauge we want, so we will take potentials $v$ such that $\exc{k}\bpa{\sum_{i \in I} v_i \wei_i} = 0$ for convenience. 
  Our variational space of potentials can thus be
\begin{align*}
	\acs{v \in \ell_r^{1}(I,\R) \st \exc{k}\bpa{\smallsum_{i \in I} v_i \wei_i} = 0},
\end{align*}
where $\nor{v}{\ell^1_r} := \sum_{i \in I} \ab{v_i} r_i$. 

Now we can state our main result for the discretized model.
\begin{theorem}[Well-posedness of the dual problem]\label{cococo} Take $\Omega \subset \R^d$ an open connected domain with Lipschitz boundary. Take a non-negative interaction $w \in (L^{p}+L^{\ii})(\R^d,\R_+)$ where $p$ is as in \eqref{dims}.

\bul \tx{(Coercivity)} Let $\weig$ be a partition of unity of $\Omega$, with $\wei_i \in L^{\ii}(\Omega,\R_+)$, such that we have $R > 0$ for which
\begin{align*}
	\pa{\supp \wei_i} \backslash \cup_{j \in I, j \neq i} \supp \wei_j
\end{align*}
	contains a ball of radius $R$, uniformly in $i \in I$. Let $r \in \ell^1(I,\R_+)$ be such that $\sum_{i \in I} r_i = N$ and $r_i > 0$ for all $i \in I$. For any $v \in  \ell_r^{1}(I,\R)$ such that $\exc{k}\bpa{\sum_{i \in I} v_i \wei_i} = 0$, we have 
\begin{align}\label{cob}
	\boxed{G\ex{k}_{r,\weig}(v) \le -\min\pa{ 1, \f{\sum_{v_i \ge c_{\Omega}}r_i}{\sum_{v_i < c_{\Omega}}r_i}} \nor{v - c_{\Omega}}{\ell^1_{r}} + c_R,}
\end{align}
	where $c_R$ depends neither on $v$ nor on $r$, and $\come$ is defined in \eqref{ccc}. In particular when $I$ is finite, $G\ex{k}_{r,\weig}$ is coercive in $\ell^1_r(I,\R) = \R^{\ab{I}}$ hence it has at least one maximizer $v \in \ell^{1}_r(I,\R)$, unique if $k=0$ and $p>\max(2d/3)$.

	\bul \tx{(Existence of an optimizer)} Make the previous assumptions, and moreover assume that $I$ is finite and $\Omega$ bounded, $v$ being the maximizing potential. There is an $N$-particle $k\expo{th}$ bound mixed state $\Gamma_v \in \cans(\Omega)$ of $\hn\pa{\sum_{i \in I} v_i \wei_i}$ such that $\int_\Omega \ro_{\Gamma_v} \wei_i = r_i$ for all $i \in I$, and such that $\gew{k}(v) = \cE_0\pa{\Gamma_v} = \fl{k}(r)$. 
\end{theorem}
We provide a proof in Section \ref{ssub:proof_cococo}. In \eqref{cob}, we use the convention that $\min\pa{ 1, \f{\sum_{v_i \ge c_{\Omega}}r_i}{\sum_{v_i < c_{\Omega}}r_i}} = 1$ when $v \ge \come$. Here are some remarks. 

\textit{(i)} By Theorem \ref{kspropi}, if
 \begin{align*}
	 \dim \Ker \pa{H_N\pa{\sum_{i \in I} v_i \wei_i} -\exc{k}\pa{\sum_{i \in I} v_i \wei_i}} \in \acs{1,2},
 \end{align*}
then $\Gamma_v$ can be chosen to be pure, and $\cE_0\pa{\Gamma_v}= \fl{0}(r) = \fll{0}(r)$.

\textit{(ii)} The weight functions $\wei_i$ can have overlapping supports, but our assumption essentially says that the inside part is not too small. In the case $I$ infinite, it is not clear whether the bound \eqref{cob} implies that $G\ex{k}_{r,\weig}$ is coercive. However, when $I$ is finite, we have $\sum_{v_i < c_{\Omega}}r_i \le N$ and $\sum_{v_i \ge c_{\Omega}}r_i \ge \min r$ so \eqref{cob} yields
\begin{align*}
G\ex{k}_{r,\weig}(v) \le -\f{\min r}{N} \nor{v}{\ell^1_r} + c
\end{align*}
for any $v \in \ell^1_r(I,\R) = \ell^1(I,\R)$, where $c = \come \min r /N + c_R$ is independent of $v$, and $\min r > 0$ thus $G\ex{k}_{r,\weig}$ is coercive in the $\ell^1_{r}$ norm. 

\textit{(iii)} Our bound \eqref{cob} does not pass to the continuous model because then $R \ra 0$ and $c_R \ra +\ii$.


\textit{(iv)} The pair $(v,\Gamma_v)$ is a saddle point of the Lagrangian
 \begin{align}\label{lagr}
 \cL\pa{v,\Gamma} = \cE_0 \pa{\Gamma} + \sum_{i \in I} v_i \pa{ - r_i + \int_\Omega \ro_{\Gamma} \wei_i}
 \end{align}
 and $v$ is a Lagrange multiplier. 

\textit{(v)} We recall that the exact existence was proved for the quantum theory on a $\Z^d$ lattice \cite{ChaChaRus85}, and in the classical case at positive temperature \cite{ChaLie84}.

\textit{(vi)} Since the sum of coercive functionals is also coercive, the similarly regularized dual functional of ensemble DFT \cite{GroOliKoh88,CerSenRobFro22}, which is a weighted sum of $G\ex{k}_{r,\weig}$ over $k$, is coercive for any weights, and a $v$-representability result for ensemble DFT similar to Theorem \ref{canot} holds.


\subsection{Building Kohn-Sham potentials}

The problem of $v$-representability is, given a density $\ro \in L^1(\Omega,\R_+)$, $\int_\Omega \ro = N$, $\sqrt{\ro} \in H^1$, and $\ab{\acs{\ro = 0}\cap \Omega}=0$, to find a potential $v$ having a $k\expo{th}$ bound state $\p_v$ satisfying $\ro_{\p_v} = \ro$. We will call it the inverse potential. When $w = 0$ and $k=0$, it is called the Kohn-Sham potential \cite{KohSha65}.

\subsubsection{The mixed states case}

In the mixed states setting (at zero temperature) and using the Bishop-Phelps theorem, Lieb showed in \cite[Theorem 3.10, Theorem 3.11, Theorem 3.14]{Lieb83b}, that any such $\ro$ can be approached to any precision in $L^1 \cap L^{d/(d-2)}$ by a $v$-representable ground mixed state density. We can state a similar result for any $k$ using our variational approach.

\begin{corollary}[Constructive approximate $v$-representability in the mixed states setting]\label{canot} Let $\Omega \subset \R^d$ be a connected open set with Lipschitz boundary. Let $\ro \in L^1(\Omega,\R_+)$ be such that $\sqrt{\ro} \in H_0^1(\Omega)$ and $\ab{ \acs{ \ro = 0}\cap \Omega} = 0$, and $k \in \N$. There exists a sequence $v_n \in L^{\ii}(\Omega,\R)$ with compact support such that $H_N\pa{v_n}$ has a mixed $k\expo{th}$ bound state $\Gamma_{v_n}$ with $\cE_0\pa{\Gamma_{v_n}} \le F\ex{k}\ind{mix}(\ro)$ and $\ro_{\Gamma_{v_n}} \ra \ro$ strongly in $(L^1 \cap L^{q})(\Omega)$, where $q$ is as in \eqref{cdc}. 

	If moreover $k=0$, $\Gamma_{v_n} \ra \Gamma_{\ii}$ strongly in $\sch_{1,1}$ up to a subsequence, where $\Gamma_{\ii}$ is a minimizer of $F\ex{0}\ind{mix}(\ro)$. Furthermore, $\sqrt{\ro_{\Gamma_{v_n}}} \ra \sqrt{\ro}$ strongly in $H^1(\Omega)$, and $\cE_0\pa{\Gamma_{v_n}} \ra F\ex{0}\ind{mix}(\ro)$.
\end{corollary}
We provide a proof in Section \ref{ssub:proof_canot}. Although the existence part of Theorem \ref{cococo} holds only for bounded open sets $\Omega$, Corollary \ref{canot} holds even when $\Omega$ is unbounded. The proof uses Theorem \ref{cococo} on a sequence of growing bounded sets $\Omega_n$ with well-chosen weight functions $\wei_i^n$. We conjecture that if there exists a potential which exactly produces $\ro$, this sequence $v_n$ converges to this exact inverse potential, in a suitable sense.

However, any $\ro \ge 0$ such that $\sqrt{\ro} \in H^1(\Omega)$ and $\int_\Omega \ro = N$ is not necessarily exactly $v$-representable. For instance if $\ro$ decreases more than exponentially, then the Kohn-Sham sequence $v_n$ would not converge in an $L^s(\Omega)$ space where $s \in [1,+\ii]$, it would become very large as $\ab{x} \ra +\ii$. It will nevertheless probably converge locally.

A consequence of Corollary \ref{canot} and Theorem \ref{kspropi} $iv)$ is the density of non-interacting pure $v$-representable densities.
\begin{corollary}\label{corw}
Take $d=1$ and $w=0$. The set
\begin{align*}
\acs{\ro_{\p_v} \st v \in \spf{k}, \p_v \in \kerr, \int_{\Omega^N} \ab{\p_v}^2= 1}
\end{align*}
is dense in the set of densities $\acs{\ro \in L^1(\Omega,\R_+) \st \int_\Omega \ro = N}$, equipped with the $L^1$ distance.
\end{corollary}


\subsubsection{The pure states case}\label{purek}


The inequality $F\ex{k}\ind{mix}(\ro) < F\ex{k}(\ro)$ implies that $\ro$ is not $v$-representable with pure $k\expo{th}$ bound states. To continue, we make a conjecture.
\begin{conjecture}[Continuity of the Levy-Lieb and Lieb functionals]\label{conjcont}
	Take an open connected domain $\Omega \subset \R^d$ with Lipschitz boundary. Take densities $\ro,\ro_n \in L^1(\Omega,\R_+)$ such that $\sqrt{\ro}, \sqrt{\ro_n} \in H^1(\Omega)$. If $\sqrt{\ro_n} \ra \sqrt{\ro}$ in $H^1(\Omega)$, then $F\ex{0}(\ro_n) \ra F\ex{0}(\ro)$ and $F\ex{0}\ind{mix}(\ro_n) \ra F\ex{0}\ind{mix}(\ro)$.
\end{conjecture}

Conjecture \ref{conjcont} would imply that the set of pure-state $v$-representable ground densities is not dense in $L^1(\Omega,\R_+)$ when $d \ge 3$. Indeed, consider a density $\ro$ such that $F\ex{0}\ind{mix}(\ro) < F\ex{0}(\ro)$, the existence of such densities is presented in \cite[Theorem 3.4 (ii)]{Lieb83b} for $d=3$ but similar examples hold for any $d \ge 3$. Then by Conjecture \ref{conjcont} there exists $R > 0$ such that $F\ex{0}\ind{mix}(\chi) < F\ex{0}(\chi)$ for any positive $\sqrt{\chi} \in \bB_R(\sqrt{\ro})$, where we considered the ball $\bB_R(\sqrt{\ro}) \subset H^1(\Omega,\R)$. Hence $\bB_R(\sqrt{\ro})$ is an open (in the set of non-negative square functions) set of densities which are not pure-state $v$-representable.

However, with a different method which is not variational, it might still be possible to represent those densities, with excited states. As presented in \cite{FreLev82} for instance, the inverse potential can be seen as a Lagrange multiplier corresponding to the Euler-Lagrange equation of the Levy-Lieb functional. We give here a result for the discretized problem which only works for $N=1$.

\begin{theorem}[Pure excited $v$-representability, $N=1$]\label{juju}
	Take $N=1$, let $\Omega \subset \R^d$ be a connected bounded open domain with Lipschitz boundary, consider a finite partition of unity $(\wei_i)_{i \in I}$ for $\Omega$, and $r \in \ell^1(I,\R_+)$, $\sum_{i \in I} r_i = 1$ and such that $r_i > 0$ for any $i \in I$. There exist $v \in \ell^{\ii}(I,\R)$ and a pure one-particle ground or excited state $\p_r \in H^1\ind{a}(\Omega)$ of $-\Delta + \sum_{i \in I} v_i \wei_i$ such that for all $i \in I$, $\int_\Omega \ro_{\p_r} \wei_i = r_i$.
\end{theorem}

A proof is provided in Section \ref{ssub:proof_canot}. Applying the last result for an increasing sequence of $\weig$'s ($\weig_n \subset \weig_{n+1}$), we get the corresponding approximate representability, as we obtained Corollary \ref{canot}. For $N=1$ the limit potential must be Bohm's potential $\Delta \sqrt{\ro} / \sqrt{\ro}$ and the state must be the ground state. We conjecture that Theorem \ref{juju} holds for any $N$, a sufficient condition being that minimizers $\p$ of our approximate Levy-Lieb functionals are such that $\ab{\acs{\p=0}}=0$.


\begin{conjecture}\label{conjos}
	Any minimizer $\p$ of $\fll{0}(r)$ satisfies 
 \begin{align*}
	\ab{\acs{X \in \Omega^N \st \p(X) = 0}}=0.
 \end{align*}
\end{conjecture}
This conjecture is related to a unique continuation property. In the Hohenberg-Kohn theorem, one considers minimizers of the energy $\cE_v$, satisfying Schr\"odinger's equation, and this implies $\ab{\acs{\p=0}}=0$ by unique continuation \cite{Garrigue19}. Here, this is a converse property in the sense that we consider minimizers of $\fll{0}(r)$, and the property $\ab{\acs{\p=0}}=0$ of minimizers, that we want to show, would imply that they satisfy Schr\"odinger's equation (see the proof of Theorem \ref{juju}).







\section{The dual problem when $w=0$}\label{secopt} 
In this section, we provide properties on the dual problem in the Kohn-Sham non-interacting case. The decomposition of the eigenfunctions into Slater determinants will enable to get more information on the local problem and on the Euler-Langrange optimality conditions.

\subsection{Definitions}\label{defdu}
We first state some definitions. 

\bul We consider the exact continuous model \eqref{exc} with potentials in $L^p+L^{\ii}$ and $p$ as in \eqref{dims}. We saw in Section \ref{repsm} that the discrete version of $\ger{k}$ is coercive. We denote by
 \begin{align}\label{defd}
	 \cD_N\ex{k}(v) := \Ker_{\R} \bpa{\hn(v) - \exc{k}(v)}
 \end{align}
the real vector eigenspace of the $N$-body operator $\hn(v)$, associated to $\exc{k}(v)$. We recall that the eigenvalues are counted with their multiplicities.

\bul Let $\cR\ind{max}$ be the set of densities such that $\ger{0}$ has a maximizer. We can define the degeneracy of a density
\begin{align*}
	\deg : \begin{array}{rcl}
		\cR\ind{max} & \longrightarrow & \N \backslash \acs{0} \\
		\ro & \longmapsto & \dim_{\R} \Ker_{\R} \bpa{\hn(v_{\ro}) - \exc{0}(v_{\ro})}, \\
\end{array}
\end{align*}
where $v_{\ro}$ is the unique maximizer of $\ger{0}$. This map is expected to have a rich structure.

\bul In the case of $w=0$, that is when the model is an effective one-body one, we distinguish two types of degeneracies in the $N$-body problem. Let us denote by $(E_i)_{i}$ and $(\vp_i)_{i}$ the real eigenvalues and a corresponding orthonormal familly of eigenfunctions of $-\Delta + v$, where $i \mapsto E_i$ is non-decreasing. The eigenfunctions of the many-body problem $\sum_{j=1}^N (-\Delta_j + v(x_j))$ are the antisymmetrized tensor products $\wedge_{i \in I} \vp_i$, where $\ab{I} = N$, having energies $\sum_{i \in I} E_i$, hence 
 \begin{align*}
	\cD\ex{k}_N(v) = \vect_{\R} \acs{\wedge_{i \in I} \vp_i \st I \in \cI\ind{tot}}
 \end{align*}
 where $\cI\ind{tot}$ is a set of $N$-tuples, and $(\wedge_{i \in I} \vp_i)_{I \in \cI\ind{tot}}$ is an orthonormal basis of $\cD\ex{k}_N(v)$. We say that the $N$-body degeneracy has a \textit{coincidental} degeneracy when several sums of energy levels \apo{accidentally} superpose while at least one energy level is different, that is when there are $I,J \subset \cI\ind{tot}$ such that $\sum_{i \in I} E_i = \sum_{i \in J} E_i$ and there is $m \in \acs{1,\dots,N}$ such that the $m\expo{th}$ (order by energy) elements of $I$ and $J$, denoted by $I_m$ and $J_m$ verify $E_{I_m} \neq E_{J_m}$. We illustrate it on the left panel of Figure \ref{degfig}. When there exists $j \in \cI\ind{tot} \backslash I$ such that $E_j = E_i$, there is a one-body level which is partially occupied, the $N$-body degeneracy comes from a one-body degeneracy, and we say that the $N$-body degeneracy is \textit{essentially one-body}, as illustrated in the middle panel of Figure \ref{degfig}. We remark that coincidental and essentially one-body degeneracies can coexist, as examplified on the right panel of Figure \ref{degfig}.

\begin{figure}
\begin{tikzpicture}[scale=0.4]
\draw[level]   (-0.5,0)  -- (0.5,0);
\draw[level]   (-0.5,1.5)  -- (0.5,1.5);
\draw[level]   (-0.5,3)  -- (0.5,3);
\draw[level]   (-0.5,4.5)  -- (0.5,4.5);
	\draw (0,0.5) -- (0,-0.5)  (0,4)--(0,5);
\end{tikzpicture}
	\hspace{0.75cm} 
\begin{tikzpicture}[scale=0.4]
	\draw[color=white] (0,0.5) -- (0,-0.5)  (0,4)--(0,5);
\draw[level]   (-0.5,0)  -- (0.5,0);
\draw[level]   (-0.5,1.5)  -- (0.5,1.5);
\draw[level]   (-0.5,3)  -- (0.5,3);
\draw[level]   (-0.5,4.5)  -- (0.5,4.5);
	\draw (0,1) -- (0,2)  (0,2.5)--(0,3.5);
\end{tikzpicture}
	\hspace{0.3cm} 
\unskip\ \vrule\
	\hspace{0.3cm} 
\begin{tikzpicture}[scale=0.4]
\draw[level]   (-0.5,0)  -- (0.5,0);
\draw[level]   (-0.5,1.5)  -- (0.5,1.5);
\draw[level]   (-1.25,3)  -- (-0.25,3);
\draw[level]   (1.25,3)  -- (0.25,3);
\draw[level]   (-0.5,4.5)  -- (0.5,4.5);
	\draw (0,0.5) -- (0,-0.5)  (-0.75,2.5) -- (-0.75,3.5);
\end{tikzpicture}
	\hspace{0.75cm} 
\begin{tikzpicture}[scale=0.4]
\draw[level]   (-0.5,0)  -- (0.5,0);
\draw[level]   (-0.5,1.5)  -- (0.5,1.5);
\draw[level]   (-1.25,3)  -- (-0.25,3);
\draw[level]   (1.25,3)  -- (0.25,3);
\draw[level]   (-0.5,4.5)  -- (0.5,4.5);
	\draw (0,0.5) -- (0,-0.5)  (0.75,2.5) -- (0.75,3.5);
\end{tikzpicture}
	\hspace{0.3cm} 
\unskip\ \vrule\
	\hspace{0.3cm} 
\begin{tikzpicture}[scale=0.4]
\draw[level]   (-0.5,0)  -- (0.5,0);
\draw[level]   (-0.5,1.5)  -- (0.5,1.5);
\draw[level]   (-1.25,3)  -- (-0.25,3);
\draw[level]   (1.25,3)  -- (0.25,3);
\draw[level]   (-0.5,4.5)  -- (0.5,4.5);
	\draw (0,0.5) -- (0,-0.5)  (0,4) -- (0,5);
\end{tikzpicture}
	\hspace{0.75cm} 
\begin{tikzpicture}[scale=0.4]
\draw[level]   (-0.5,0)  -- (0.5,0);
\draw[level]   (-0.5,1.5)  -- (0.5,1.5);
\draw[level]   (-1.25,3)  -- (-0.25,3);
\draw[level]   (1.25,3)  -- (0.25,3);
\draw[level]   (-0.5,4.5)  -- (0.5,4.5);
	\draw (0,1) -- (0,2)  (-0.75,2.5) -- (-0.75,3.5);
\end{tikzpicture}
	\hspace{0.75cm} 
\begin{tikzpicture}[scale=0.4]
\draw[level]   (-0.5,0)  -- (0.5,0);
\draw[level]   (-0.5,1.5)  -- (0.5,1.5);
\draw[level]   (-1.25,3)  -- (-0.25,3);
\draw[level]   (1.25,3)  -- (0.25,3);
\draw[level]   (-0.5,4.5)  -- (0.5,4.5);
	\draw (0,1) -- (0,2)  (0.75,2.5) -- (0.75,3.5);
\end{tikzpicture}
	\caption{Three kinds of $N$-body degeneracies: coincidental, essentially one-body, and both.}\label{degfig}
\end{figure}
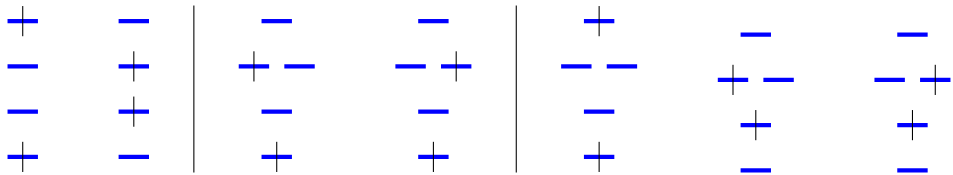

\bul Given a discrete set $S = \acs{s_i}_{i \in \N} \subset \R$ bounded below and such that $s_i \le s_j$ for $i < j$, we define, for each $j \in \N$,
 \begin{align}\label{defmu}
	 \mu_j(S) := s_j.
 \end{align}
 For a $n \times n$ diagonalizable matrix $M$ having real eigenvalues, we define $\mu_j(M) := \mu_j(D_M)$ where $D_M$ is the set of eigenvalues.

 \bul Consider that $(\p_I)_{I \in \cI\ind{tot}}$ is a basis of $\cD\ex{k}_N(v)$, where $\p_I = \wedge_{i \in I} \vp_i$, $\vp_i$ being orthonormal orbitals of $-\Delta + v$ as defined before. Let us define the \apo{inner} orbitals $I\ind{in} := \acs{i \in \N \st \forall I \in \cI\ind{tot}, i \in I}$, which are present in all the many-body functions of $\cD\ex{k}_N(v)$, they necessarily fill their energy levels. We define the \apo{inner} density 
 \begin{align*}
	 \ro\ind{in} := \mysum{i \in \N \\ i \in I \hs \forall I \in \cI\ind{tot}}{} \vp_i^2.
 \end{align*}
We now drop these \apo{inner} orbitals and only consider those which can change on $\cD\ex{k}_N(v)$. Now we define the set $\cI\ind{out} := \acs{ I \backslash I\ind{in} \st I \in \cI\ind{tot}}$ of $(N - \ab{I\ind{in}})$-body wavefunctions. The disjoint union of sets is denoted by $\cupdot$. For $k \in \N$, we define $I\ex{k} := \acs{i \in I \st \cE_v(\vp_i) = E_{k}}$ the subset of orbitals belonging to the $k\expo{th}$ one-body level, where $\cE_v(\vp) = \ps{\vp,H_1(v)\vp}$. For $I,J \in \cI\ind{out}$, we define the elements of the function-valued matrix of \apo{one-body correlations}
\begin{align}\label{matcor}
	\pa{\matix}_{IJ} = \delta_{IJ} \sum_{i \in I} \vp_i^2 +  \vp_{i}\vp_{j}\delta_{\substack{\exists \ell,i,j \in \N \\ I \cupdot J =I\ex{\ell} \cupdot J\ex{\ell} = \acs{i,j} \\   I\ex{t} = J\ex{t} \forall t \neq \ell}},
\end{align}
where $\delta_A = 1$ if and only if the condition $A$ is satisfied. For instance when the only degeneracy is essentially one-body, and comes from a one-body level $\vect (\vp_{\ell + i})_{1 \le i \le D}$ with degeneracy $D$, filled with one particle, then $\matix = \mat{\vp_{\ell+1} & \dots & \vp_{\ell+D}}^{\tx{T}} \otimes \mat{\vp_{\ell+1} & \dots & \vp_{\ell+D}}$. 
 In the same situation but when $D=3$ and when there are two particles in the one-body-level,
\begin{align*}
\matix = \begin{pmatrix} 
\phi_1^2 + \phi_2^2 &  \phi_2 \phi_3 & \phi_1\phi_3\\
 \phi_2 \phi_3 & \phi_1^2 + \phi_3^2 &  \phi_1\phi_2\\
 \phi_1\phi_3 &  \phi_1\phi_2 & \phi_2^2 + \phi_3^2 \\
\end{pmatrix},
 \end{align*}
 where $\phi_i := \vp_{\ell+i}$.

\subsection{Local problem}
By the next lemma, we can say that coincidental degeneracies \apo{do not correlate} the many-body eigenstates in our problem.

 \begin{lemma}[Local problem when $w=0$]\label{localp}
Take a density $\ro \in L^1(\Omega)$, $\ro \ge 0$, $\int_\Omega \ro = N$, $\sqrt{\ro} \in H^1(\Omega)$, consider a binding $v \in \spf{k}$, and take $w=0$. 

$i)$ We have 
\begin{align*}
	 {^+} \delta_v \ger{k} (u) = \mu_{k-m_k^v} \pa{\int_{\Omega} u \bpa{(\ro\ind{in} - \ro) \indic + \cM_{\vp}}}
 \end{align*}

	 $ii)$ The matrix $\matix = \cM_1 \oplus \dots \oplus \cM_T$ is block diagonal, where the blocks correspond to the different essentially one-body degeneracies. 

	 $iii)$ In the case of only coincidental degeneracies, with $\p_1, \dots, \p_{\dim \cD\ex{k}_N(v)}$ being an orthonormal basis of $\kerr$ composed of Slater determinants, then
 \begin{align*}
 {^+} \delta_v \ger{k} (u) = \mu_{k-m_k^v} \pa{ \int_\Omega u(\ro_{\p_i} - \ro)}_{1 \le i \le \dim \cD\ex{k}_N(v)}.
 \end{align*}
\end{lemma}
We provide a proof in Section \ref{ssub:proof_secopt}. We recall that $\mu$ is defined in \eqref{defmu}. The block-diagonalization shows that the degeneracies which complexify the problem are the essentially one-body degeneracies, not the coincidental ones. In $iii)$, the normalized direction $u^*$ maximizing ${^+} \delta_v \ger{k} (u)$ over the unit ball of $L^p(\Omega)$ is $\ab{\f{\ro_{\p_i}-\ro}{\nor{\ro_{\p_i}-\ro}{L^{\f{p}{p-1}}(\Omega)}}}^{\f{1}{p-1}} \sgn(\ro_{\p_i}-\ro)$, for the corresponding $i$.

\subsection{Optimality}

\begin{proposition}[Euler-Lagrange inequations when $w=0$]\label{ksprops}

Take a density $\ro \in L^1(\Omega)$, $\ro \ge 0$, $\int_\Omega \ro = N$, $\sqrt{\ro} \in H^1(\Omega)$, consider a binding $v \in \spf{k}$, take $w=0$, and assume that $v$ maximizes $\ger{k}$.

	$i)$ We have, a.e in $\Omega$,
\begin{align}\label{codl}
	\ro\ind{in} +  \mu_{k-m_k^v} \pa{  \matix } \le \ro \le \ro\ind{in} + \mu_{M_k^v-k} \pa{ \matix }.
 \end{align}

	$ii)$ If there is only one \apo{outer} particle, that is $N = 1 + \int_\Omega \ro\ind{in}$, then $k = m_k^v$ and 
\begin{align*}
	\ro\ind{in} \le \ro \le \ro\ind{in} + \sum_{i \in G} \vp_i^2,
 \end{align*}
	where $G$ is such that $(\vp_i)_{i \in G}$ is a basis of the one-body degenerate level $\Ker_{\R} (-\Delta + v - E_{\ell})$ producing the degeneracy of the $N$-body level.


	$iii)$ If $\kerr$ has only coincidental degeneracies, then it has a $k\expo{th}$ bound pure Slater state $\p = \wedge_{i \in I} \vp_i$ such that $\ro_{\p} = \ro$. 
\end{proposition}
We provide a proof in Section \ref{ssub:proof_secopt}. We remark that the larger $k-m_k^v \ge 0$, the more constraining are the Euler-Lagrange inequalities.

In particular, when $k=0$, there are only essentially one-body degeneracies and $i)$ becomes
\begin{align*}
\ro\ind{in} +  \min \sigma \pa{  \matix } \le \ro \le \ro\ind{in} +\max \sigma  \pa{ \matix }.
 \end{align*}

 Finally, we numerically find that when all the different partially filled one-body levels have dimensions 2, then there is a pure $k\expo{th}$ bound state representing $\ro$ at optimality. We believe that the sets of mixed states densities and pure densities are equal in such configuration.

\section{Numerical simulations}\label{simus}

In this section, we implement the dual problem and compute inverse potentials. The algorithm is presented for $w=0$, but its extension to any $w$ can be easily adapted from this presentation.

\subsection{Definition of the problem}

We do not use the framework of the discretized space of potentials introduced previously. It was developped to regularize the dual problem and show approximate $v$-representability of densities, but it is not useful for simulations since the expensive step is the computation of eigenstates, and it is faster to treat the problem with the full potential space directly. Instead we use a finite plane waves basis, which corresponds to the Fourier dual of the model studied in \cite{ChaChaRus85}, and hence periodic boundary conditions. Our spatial length will be denoted by $L > 0$, which will be equal to $1$ or $5$ in our applications. We implemented the algorithm in Julia \cite{Julia}, using the LOBPCG algorithm extract from DFTK \cite{DFTK20}. We consider
\begin{align*}
\hn^{w=0}(v) := \sum_{i=1}^N \pa{-\Delta_i + v(x_i)}.
\end{align*}
We use the same notations as in the continuous case, but they have to be taken in their discrete versions. The plane waves basis
\begin{align*}
e_k(x) = \f{e^{i 2\pi k \cdot x}}{L^d}
\end{align*}
where $x \in [0,L]^d$ and $\acs{-K,\dots,K}^d \subset \Z^d$ where $K \in \N$ is a cutoff, which can be taken to be $25$ for our figures because we took regular densities $\ro$. Given some density $\ro$ belonging to the set
\begin{align}\label{tden}
\acs{ \ro \in L^1(\Omega) \st \ro \ge 0, \mediumint \ro = N},
\end{align}
and called the target density, our goal is to find a potential $v$ such that $\hn^{w=0}(v)$ has at least one $k^{\tx{th}}$ bound state and such that there exists a mixed state $\Gamma_{v,k}$ with range on $\Ker_{\R} \bpa{\hn^{w=0}(v)-\exc{k}(v)}$ such that
\begin{align*}
\ro_{\Gamma_{v,k}} = \ro.
\end{align*}
At the discrete level, the existence is justified by \cite{ChaChaRus85} for $k=0$, and at the continuous level by our previous results for $k \ge 1$. We recall that for $k=0$ the searched potential is unique by the Hohenberg-Kohn theorem \cite{ChaChaRus85,HohKoh64}. We would also like to know whether the set of $v$-representable pure state densities
\begin{multline}\label{rep}
\bigg\{ \ro_{\p^{(k)}_v} \st v \in (L^p+L^{\ii})(\Omega),  \\
 \p^{(k)}_v \in \Ker \bpa{\hn^{w=0}(v)-\exc{k}(v)}, \int_{\Omega^N} \ab{\p^{(k)}_v}^2 = 1\bigg\}
\end{multline}
is dense in the set of densities \eqref{tden}. We previously saw in Corollary \ref{corw} that it is dense when $d=1$, but we conjectured that it is not so when $d=3$ (Conjecture \ref{conjcont}).  

More explicitely, since $w=0$, we are led to study the one-body operator $-\Delta + v$. Following the notations defined in Section \ref{defdu}, its eigenvectors are denoted by $\vp_i$ and the energies by $E_i = \int_\Omega \ab{\na \vp_i}^2 + \int_\Omega v \ab{\vp_i}^2$. The $N$-body $k\expo{th}$ bound states $\p_I = \wedge_{i \in I} \vp_i$ of $\hn^{w=0}(v)$ are labeled by $I \in \cI\ind{tot}$, the corresponding density is $\ro_{\p_I} = \sum_{i \in I} \ab{\vp_i}^2$ and the energy is $\exc{k}(v) = \sum_{i \in I} E_i$.


\subsection{Algorithm}

As we can see with Theorems \ref{kspropi} and \ref{cococo}, we need to maximize $\ger{k}$, which is a well-posed problem. Indeed, since in our simulations the system lives in a bounded set and since the space is discretized, $\ger{k}$ is coercive. Once we found its maximizers, knowing whether there exists a pure state density equal to $\ro$ boils down to computing \eqref{fpure}.

We apply a gradient ascent algorithm on $\ger{k}$ to maximize this function.

\subsubsection{Starting point}

When $N$ is small, that is $N \lesssim 3$, we start from Bohm's potential
\begin{align*}
v\ind{Bohm} := \f{\Delta \sqrt{\ro}}{\sqrt{\ro}},
\end{align*}
which produces exactly $\ro$ when $N=1$ and $k=0$. When $N$ is large, that is $N \gtrsim 4$, we start from Thomas-Fermi's potential
\begin{align}\label{eq:TF}
v\ind{TF} :=  - c\ind{TF} \ro^{\f 2d}, \qquad  c\ind{TF} := 4\pi^2 \pa{\f{d}{\ab{S^{d-1}}}}^{\f 2d},
\end{align}
which produces exactly $\ro$ when $N$ is large \cite{FouLewSol18}. The arbitrary threshold $4$ will be justified in Section \ref{sub:TF}.

\subsubsection{Ascent direction}

We start by computing the first $A$ eigenfunctions $(\vp_i)_{1 \le i \le A}$ of $-\Delta + v$, where $A \ge N+k$. The exact minimal needed value of $A$ can be larger than $N+k$ in case of degeneracies. Then we compute all the energy configurations $\sum_{i \in I} E_i$ for $I \subset \acs{1,\dots,A}$ such that $\ab{I} = N$, and we store those energies in the non-decreasing order. This gives us the $N$-body spectrum of $\hn(v)$, $\exc{k}(v)$ being the $(k+1)\expo{th}$ number of this list, and we deduce $\Ker_{\R} \bpa{\hn^{w=0}(v)-\exc{k}}$ by taking the configurations having energies close to $\exc{k}(v)$ as will be detailed later.

A direction of steepest ascent is given by \eqref{lao} and \eqref{maxu}, where we will take $p=2$ for simplicity, but it would be interesting to study the dependence of convergence with respect to this exponent. However, also for simplicity, we will not take a steepest ascent direction, but one solving
\begin{align}\label{paso}
	\mysup{u \in L^2(\Omega,\R) \\ \nor{u}{L^2}=1} {^+}\delta_v G\ex{m_k^v} (u) & = \mymin{\Gamma \in \canp\bpa{\cD\ex{k}_N(v)} \\  \Gamma \ge 0, \tr \Gamma =1} \nor{\ro_{\Gamma} - \ro}{L^2(\Omega)}.
\end{align}
The supremum is attained by
\begin{align*}
u^* =\f{\ro_{\Gamma^*}-\ro}{\nor{\ro_{\Gamma^*}-\ro}{L^2(\Omega)}}, 
\end{align*}
where $\Gamma^*$ is an optimizer of the right hand side of \eqref{paso}. This is justified because ${^+}\delta_v G\ex{m_k^v} \le {^+}\delta_v G\ex{k}$ so in this direction $u^*$, we still have 
 \begin{align*}
	 0 \le {^+}\delta_v G\ex{k} (u^*).
 \end{align*}
This scheme should lead to a maximum because local maximas of $G\ex{k}$ are global by Theorem \ref{kspropi}. Moreover, we experimentally remark that by using the direction given by \eqref{paso} or the steepest ascent direction given by \eqref{lao}, having $k \neq m_k^v$ is very rare, and for almost converged potentials, we always have $k = m_k^v$. In some situations, the min/max problem \eqref{lao} is not necessarily light to compute, it complexifies the implementation, only marginally accelerates the convergence, and is not convenient for implementing temperature, further justifying the use of the direction \eqref{paso}.

In case of degeneracies, which happens \apo{most of the time}, if we choose a direction $u = \pa{\ro_{\Gamma}-\ro}/\nor{\ro_{\Gamma}-\ro}{L^2(\Omega)}$ where $\Gamma$ is a randomly choosed mixed state of $\kerrw$, the algorithm starts to diverge. Hence optimizing over directions is necessary, except when $(d,k)=(1,0)$ because of the non-degeneragy theorem.

\subsubsection{Temperature}

In order to smooth out the behavior of the algorithm and improve its convergence, we introduce a \apo{temperature} effect. We define the set of $N$-body Slater functions of $v$ built on eigenfunctions
 \begin{align*}
	 \cF(v) := \acs{ \wedge_{i \in I} \vp^i \st \ab{I} = N, I \subset \N \backslash \acs{0}}
 \end{align*}
and consider the problem
 \begin{align}\label{passt}
 \cP(v) := \mymin{\Gamma \in \canp( \cF(v) ) \\ \Gamma \ge 0, \tr \Gamma = 1}  e^{\ud \bpa{\f{\cE_v\pa{\Gamma}-\exc{k}(v)}{T}}^2} \indic_{\ab{\cE_v\pa{\Gamma}-\exc{k}(v)} \le \cT}\int_{\Omega} \pa{\ro_{\Gamma} - \ro}^2,
 \end{align}
where $T$ and $\cT$ are fictitious temperatures. This problem \eqref{passt} is solved by an optimal damping algorithm (ODA) \cite{CanBri00,CanBri00a,Cances00b}. We define $\ro_n := \ro_{\Gamma^*}$ where $\Gamma^*$ is a minimizer of $\cP(v_n)$ and $n$ is the $n\expo{th}$ iteration step.
 In Appendix 3, we provide the computations needed in the implementation of this part of the algorithm.

First, the cut-off $\indic_{\ab{\cE_v\pa{\Gamma}-\exc{k}(v)} \le \cT}$ considerably lowers the dimension of the optimization set $\cF(v)$, dropping configurations having energies too far from the relevant one. Then, the smoothing factor $e^{- \ud \bpa{\f{\cE_v\pa{\Gamma}-\exc{k}(v)}{T}}^2}$ enables to take into account many-body states which do not exactly have energy $\exc{k}(v)$ but are close, addressing degeneracies in a continuous way. The absence of this last factor raises divergence issues.

Let $T_n,\cT_n$ be the temperatures at step $n$. We take 
 \begin{align*}
	 \cT_0 = \f{E_{N+k}^{v_0} - E_0^{v_0}}{B(N+k)}   
 \end{align*}
where $E_i^{v_0}$ is the $i\expo{th}$ eigenenergy of the one-body operator $-\Delta + v_0$ and $B$ is a parameter. We take $T_n = \cT_n/D$ for any $n \in \N$, with $D$ being a parameter, and progressively decrease them by choosing $\cT_{n} = \alpha^{\floor{n/M}} \cT_0$, where $\alpha < 1$, which we call the cooling factor, and $M \ge 1$. We remarked that cooling \apo{by steps} with the factor $\alpha^{\floor{n/M}}$, rather than with $\alpha^{n/M}$, improves the convergence. In practice, we find that $D=B=10$, $M=5$ and $\alpha = 3/4$ are good trade-offs.

\subsubsection{Line search optimization in the direction found}

The previous procedure provided us an ascent direction, we now want to optimize the step in this direction. We define
\begin{align*}
	v(\lambda) := v_n + \lambda \f{\ro_n-\ro}{\nor{\ro_n-\ro}{L^2(\Omega)}},
\end{align*}
and take parameters $\mu > 1$ and $\nu_0 > 0$. If $\cP(v(\nu_n)) \ge \cP(v(0))$, then we compute $\cP(v(\mu^j \nu_n))$ for increasing values of $j \in  \N \backslash \acs{0}$ until $\cP\pa{v(\mu^{j+1} \nu_n)} \le \cP(v(\mu^j \nu_n))$, the last value of $j$ being denoted by $j_n^*$. If $\cP(v(\nu_n)) \le \cP(v(0))$, then we compute $\cP(v(\mu^{-j} \nu_n))$ for $j \in \N \backslash \acs{0}$ until $\cP\pa{v(\mu^{-(j+1)}\nu_n)} \le \cP(v(\mu^{-j}\nu_n))$, this defines $j_n^*$ in this other case. Finally, the new potential will be
 \begin{align*}
	 v_{n+1} := v\bpa{\mu^{j_n^*} \nu_n}.
 \end{align*}
 We \apo{learn} the step size in the sense that $\nu_{n+1} = \mu^{j^*_n} \nu_n$ if $\nor{\ro_{n+1}-\ro}{L^2(\Omega)}/N \le 10^{-3}$, and $\nu_{n+1} = \nu_0$ otherwise.

\subsubsection{Convergence criterion}

We consider that the algorithm converged when $\sqrt{\cP(v_n)}/N \le \ep$ and $\alpha^{\floor{n/M}} \le \delta$, where we take $\ep = 10^{-5}$ and $\delta = 10^{-2}$. This ensures that we found $v$ and a mixed state $\Gamma$ supported on the $\ell\expo{th}$ bound states $\p^{\ell}_i$ such that 
 \begin{align*}
	 \f{\nor{\ro_{\Gamma} - \ro}{L^2(\Omega)}}{N} \le \ep, \bhs B(N+k) \f{\ab{\cE_v(\p_i^{\ell}) - \exc{k}(v)}}{E_{N+k}^{v_0} - E_0^{v_0}} \le \delta.
 \end{align*}

\subsubsection{Pure states representability}
Once we found an approximate maximizing potential $v_n$ of $\ger{k}$, if we want to know whether it produces a pure state density $\ro$, we compute
 \begin{align}\label{passtp}
	 \mymin{\p \in \vect_{\C} \cF(v_n) \\ \int_{\Omega^N} \ab{\p}^2 = 1} \indic_{\ab{\cE_{v_n}\pa{\p}-\exc{k}(v_n)} \le \cT_n} \int_{\Omega} \bpa{\ro_{\p} - \ro}^2.
 \end{align}
 The above problem is computed using a particle swarm optimization algorithm implemented in the library \href{https://manoptjl.org/stable/}{Manopt.jl}. The base manifold is a complex $(\dim \cD\ex{k}_N(v) -1)$-dimensional Grassman manifold representing the optimizing set of pure states.

\subsubsection{Remarks}\tx{ }

\bul The global constant in potentials has no importance during the scheme, we only fix it in the graphs for readability purposes. 

\bul In the case where the $k\expo{th}$ level of $H^{w=0}_N(v_n)$ is non-degenerate, one can use a Newton or quasi-Newton algorithm such as BFGS to accelerate the convergence. With our notation and considering that the spectrum is purely discrete, for two directions $u, h \in (L^p+L^{\ii})(\Omega)$, 
 \begin{align*}
	 \d^2_v \ger{k}(u,h) & = \ps{u, \d_v \ro\ex{k}(v) h} \\
	 & = 2 \mysum{J \subset \N, \ab{J} = N \\ I_k \cupdot J = \acs{i,j}}{} \bigg( \exc{k}(v)-\sum_{\ell \in J} E_{\ell}\bigg) ^{-1} \ps{\vp_i,u \vp_j} \ps{\vp_i,h \vp_j},
 \end{align*}
where $v \mapsto \ro\ex{k}(v)$ is the map from potentials to eigenstate densities, and where $I_k$ is the configuration corresponding to $\Ker_{\R} \bpa{H^{w=0}_N(v)-\exc{k}(v)} = \R \bpa{ \wedge_{i \in I_k} \vp_i}$. In a finite basis $(u_i)_i$, the Hessian $\na_v \ger{k} = \bpa{\d^2_v \ger{k}(u_i,u_j)}_{ij}$ is non-degenerate when $k=0$. In case of degeneracies, see \cite{ShaFan95} for a full treatment, see also \cite{PolRoc96}.

 \bul The library DFTK is configured for periodic boundary conditions, but we only consider densities which are very close to zero close to the boundaries, this implies potentials which are very large at the boundaries, and we recover a situation equivalent to a setting with Dirichlet boundary conditions.

\bul As expected, for $N=1$ and $k=0$, the potential $v_0 = \Delta \sqrt{\ro} / \sqrt{\ro}$ has a density very close to the target density $\ro$. For any $d, N, k$, the algorithm converges significantly faster when we start from $v_0 = \Delta \sqrt{\ro} / \sqrt{\ro}$ compared to $v_0 = 0$.


\subsection{Convergence results}

The convergence is theoretically justified by Theorem \ref{cococo}, and confirmed in our simulations. Up to slight adaptations of the parameters $\alpha$, $D$, $B$, and $M$, the algorithm always converges both at the levels of densities and of potentials, as expected, for any $d \in \acs{1,2,3}$ and any $\ro, N, k$. Moreover, the larger $N$, the faster the convergence. We obtain arbitrary precision on \eqref{passtp}, and observe numerically that it decreases as $1/n$. We give a first qualitative illustration on Figure~\ref{f3}, for $k=0$, related to the LDA which locally approximates densities by uniform electron gas partitions \cite{GiuVig05,LewLieSei18,LewLieSei19}.
\begin{figure}
\begin{center}
\includegraphics[width=0.49\textwidth]{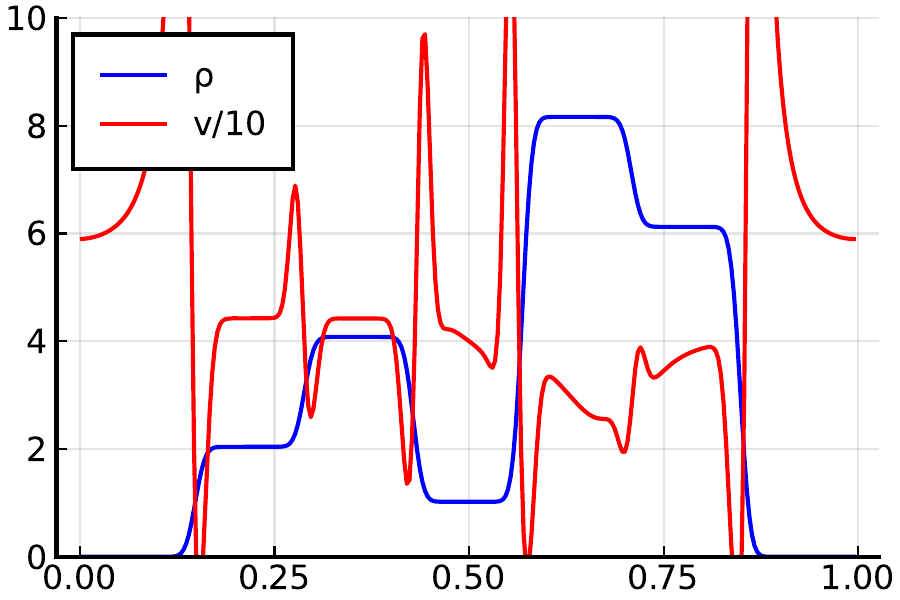} 
\includegraphics[width=0.49\textwidth]{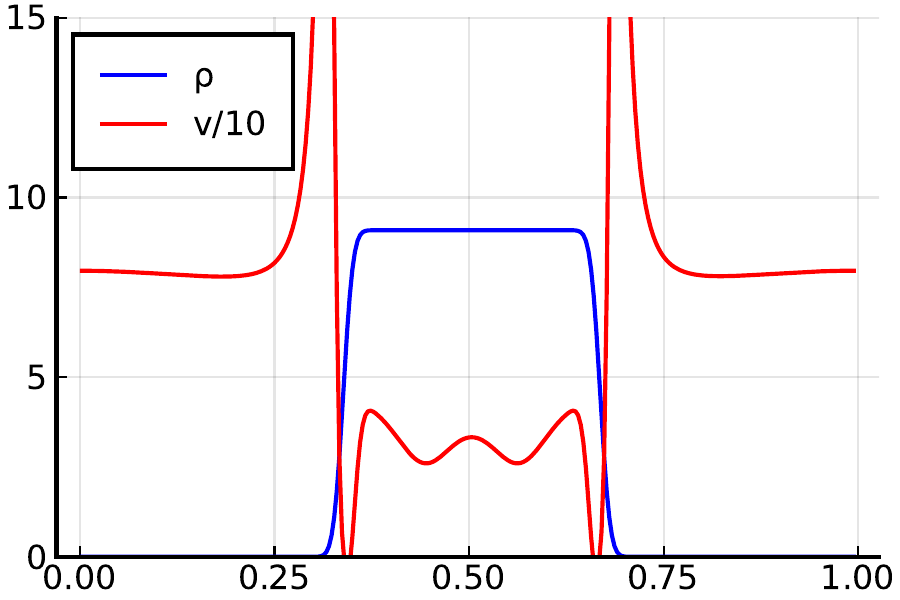} \\
\includegraphics[width=0.49\textwidth,trim={2.5cm 0cm 2.5cm 0},clip]{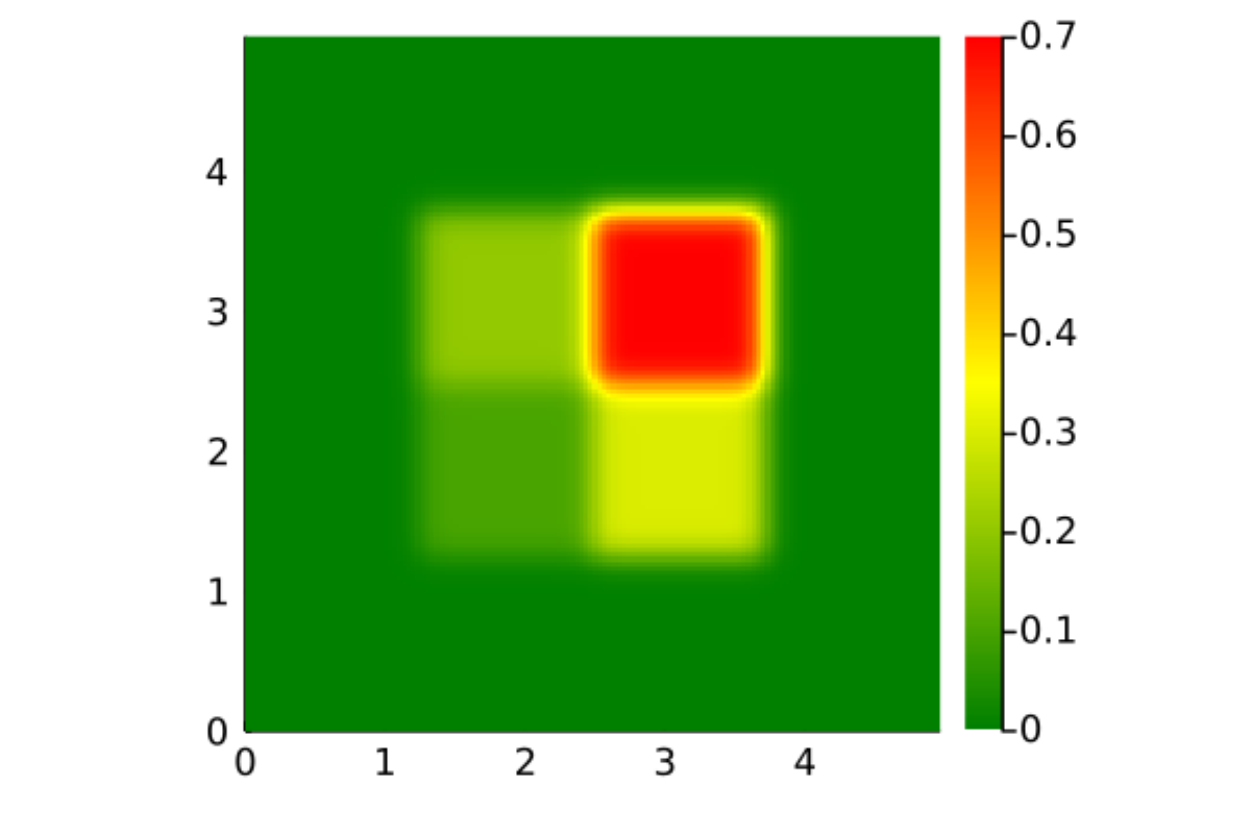}
\includegraphics[width=0.49\textwidth,trim={2.5cm 0cm 2.5cm 0},clip]{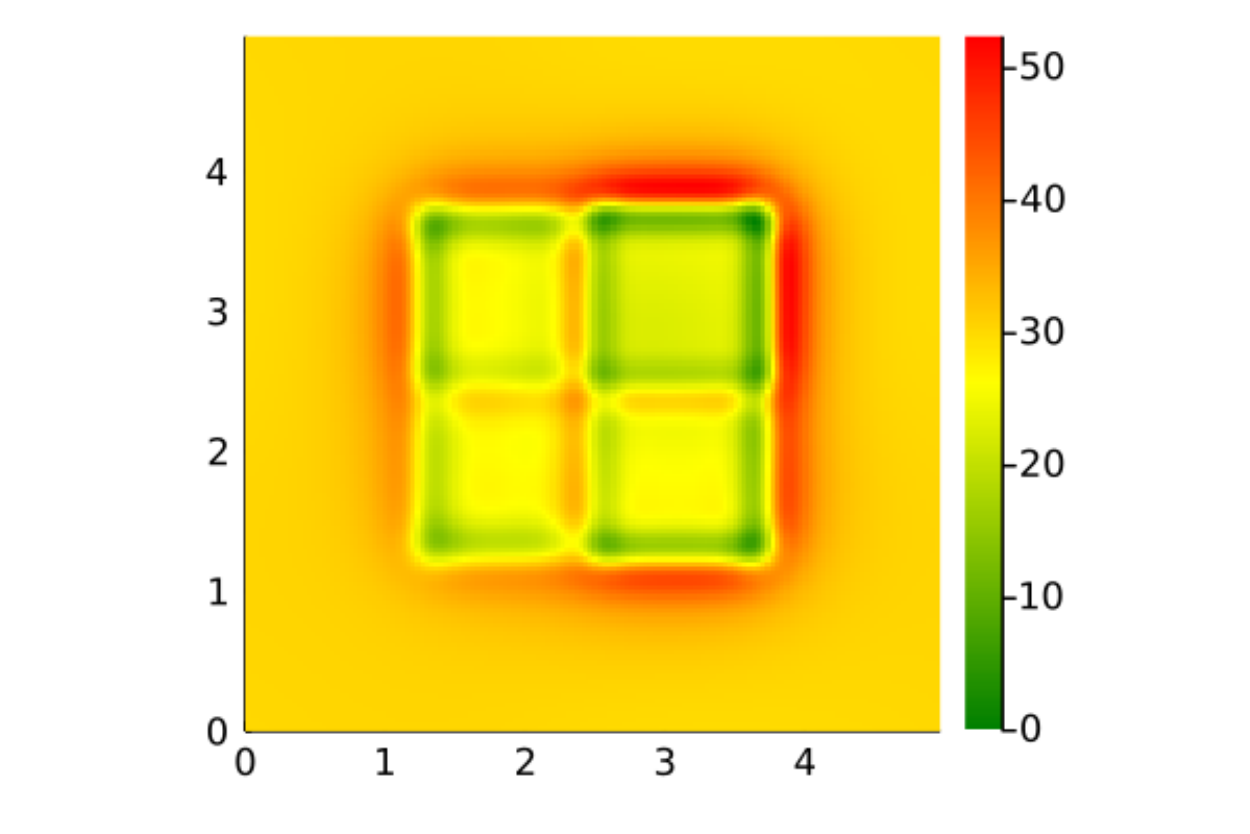}
\end{center}
	\caption{Target densities and their Kohn-Sham potentials for $k=0$. On the first line, $d=1$ and $N=3$, densities (blue) and inverse potentials (red) are in different units. On the second line, $d=2$, $N=2$, the target density is on the left and its inverse potential on the right.}\label{f3}
\end{figure}

As an illustration for $d=3$, we give on Figure~\ref{f4} a representation of $\ro$, $v_n$ and $\log_{10} \ab{\ro_n-\ro}$ where $N=4$, $k=1$ and where $\ro$ is a sum of three Gaussians.
\begin{figure}
\begin{center}
\includegraphics[width=0.49\textwidth]{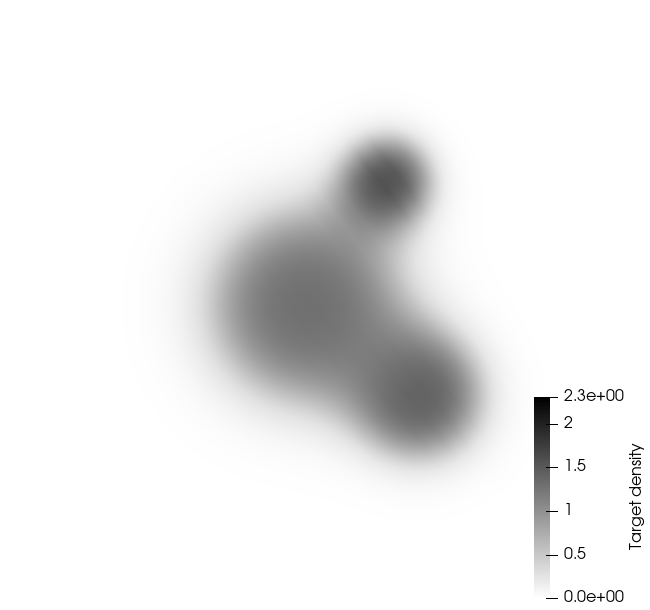}
\includegraphics[width=0.49\textwidth]{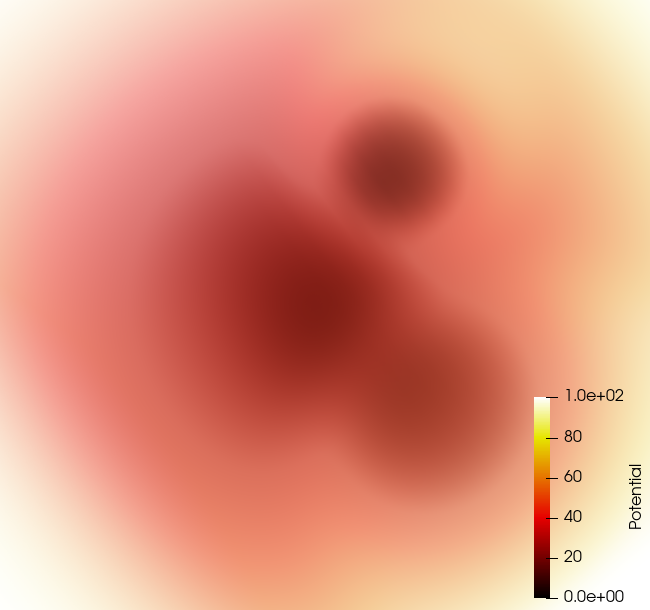} 
\end{center}
\caption{We take $d=3$, $N=4$, $k=1$, $\ro$ is a sum of three Gaussians, and the convergence criterion was reached for $n=50$. We plot $\ro$ on the upper panel and $v_n$ on the left lower panel.}\label{f4}
\end{figure}

\subsubsection{Uniqueness}
\begin{figure}
\begin{center}
\includegraphics[width=0.32\textwidth]{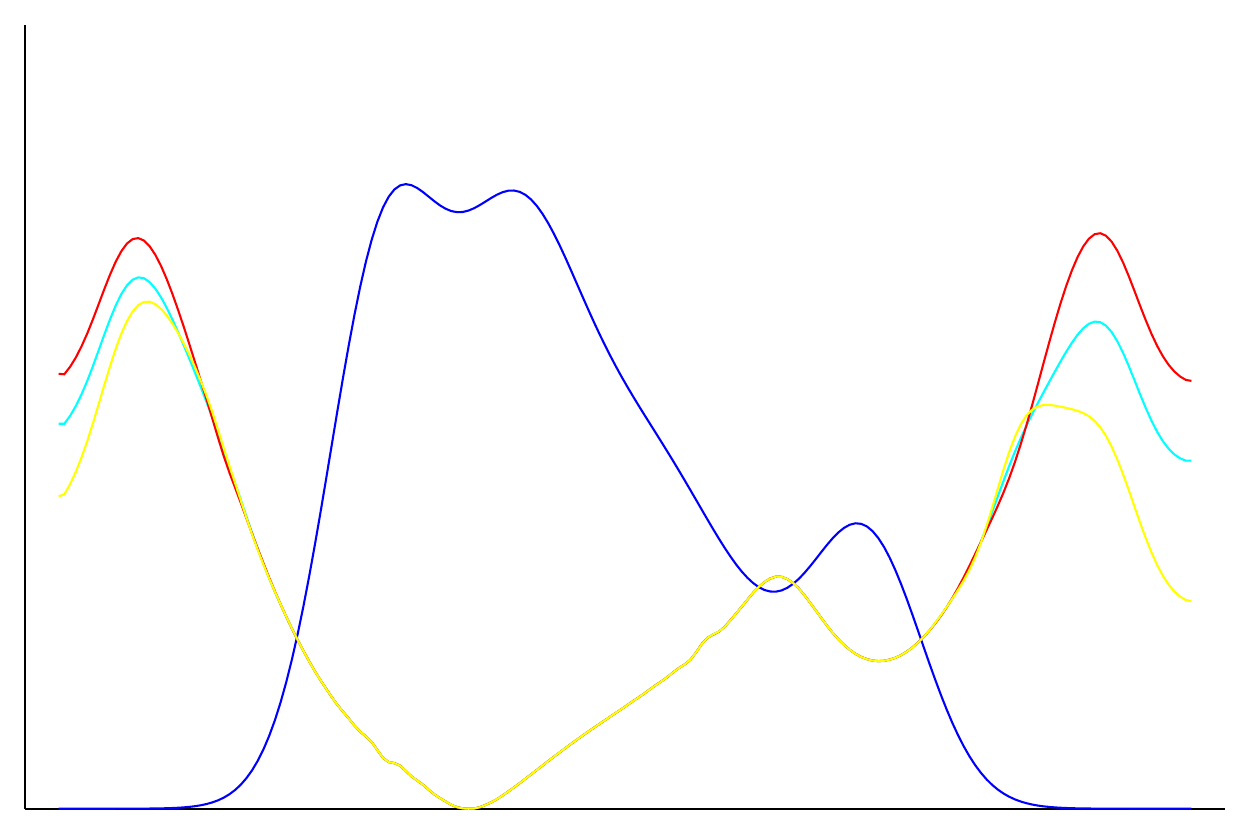}
\includegraphics[width=0.32\textwidth]{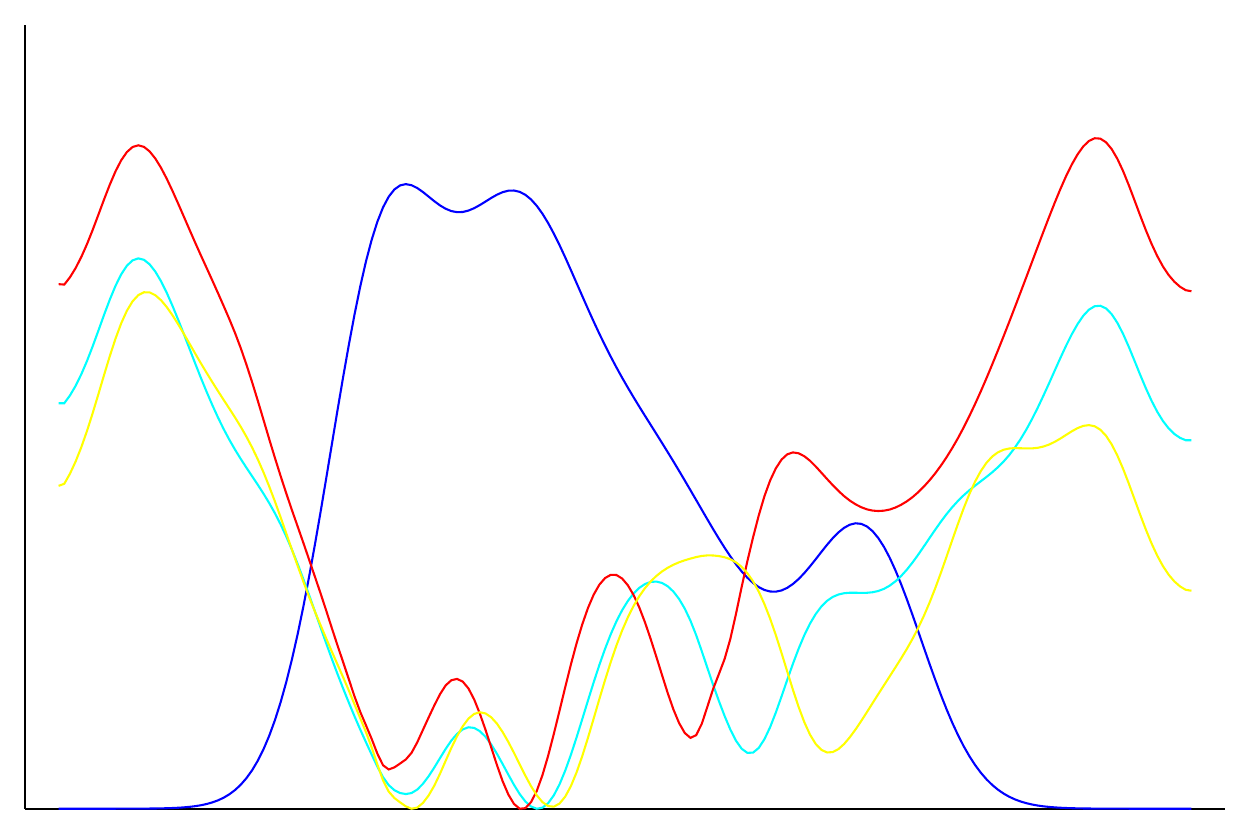}
\includegraphics[width=0.32\textwidth]{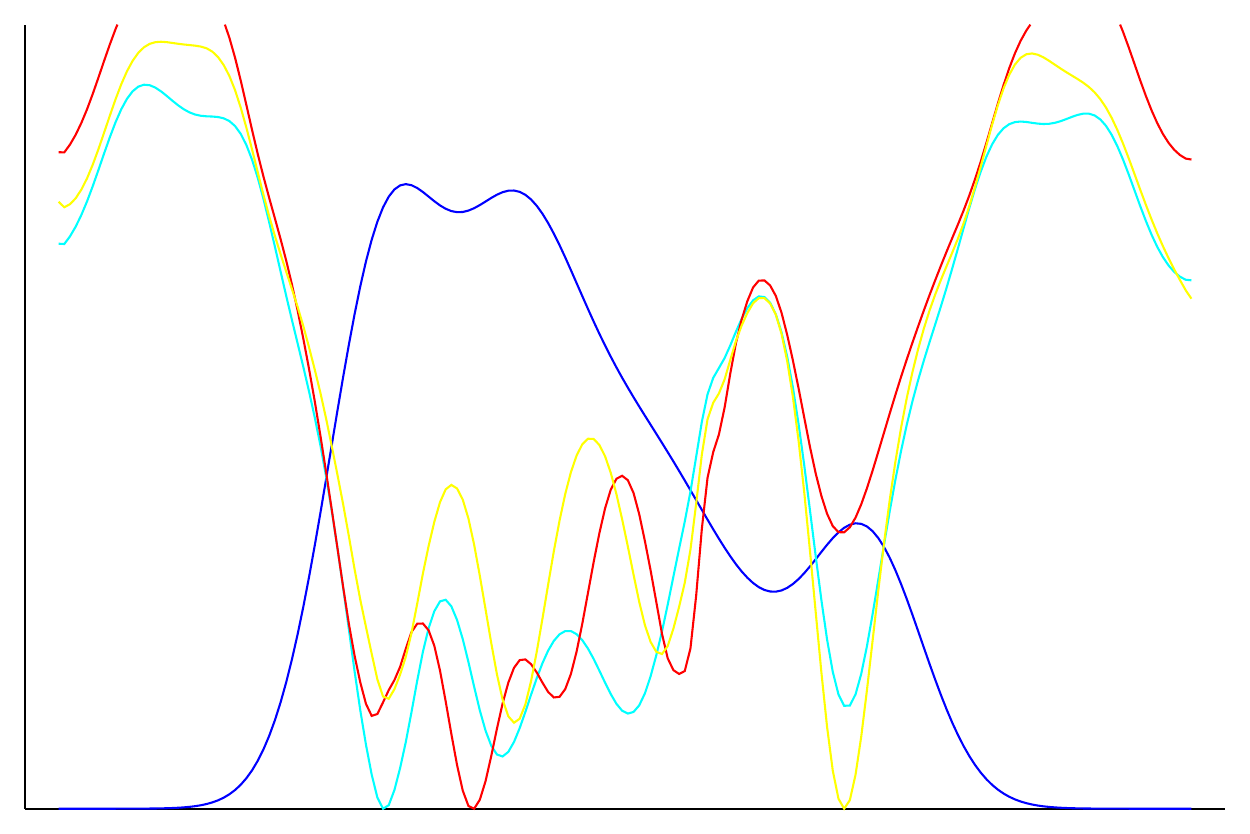}
\end{center}
\caption{Simulations for $d=1$, $N=3$, and $k = 0$ on the left, $k=1$ in the middle and $k=5$ on the right. Densities are in blue and inverse potentials are in other colors, their differences come from different initial potentials $v_0$. Potentials are in the same units, but not densities.} \label{f1dd}
\end{figure}

On Figure \ref{f1dd}, we represent inverse potentials in the case $d=1$, $N=3$, for the same density, varying $k$ and varying $v_0 = \Delta \sqrt{\ro} / \sqrt{\ro} + u_i$, where $(u_i)_{1 \le i \le 3}$ is the same sequence for $k \in \acs{0,1,5}$, but the $u_i$'s are pairwise different. This confirms previous numerical studies \cite{GauBur04} showing that uniqueness of potentials does not hold for $k\ge 1$, in which the authors used the kernel of $\d_v \bpa{u \mapsto \ro_{\p\ex{k}(u)}}$ to find such examples, where $\p\ex{k}(u)$ is the non-degenerate $k\expo{th}$ eigenstate of a potential $u$. There must be an infinite number of inverse potentials, differing by \apo{oscillations}.

\subsubsection{Reconstruction of potentials}

We choose an initial potential $v$, compute $\ro_{\Gamma_v}$ for some $\Gamma_v \in \cS\bpa{\cD_N\ex{k}(v)}$ with $\Gamma_v \ge 0$, $\tr \Gamma_v = 1$, and launch the algorithm on $\ro_{\Gamma_v}$. For $k=0$, by the Hohenberg-Kohn theorem, $v_n$ should converge to $v$, this is what we call the reconstruction of potentials.

We show in Figure~\ref{f1} an example of what one can obtain for $d=1$, in the ground state setting $k=0$ and for the third excited state $k=3$, with the same target potential. On the left on Figure~\ref{f1}, we see that the potential is well reconstructed.
\begin{figure}
\begin{center}
	\hspace{0.1cm}
\includegraphics[width=0.46\textwidth,trim={0.6cm 0.5cm 0.2cm 0},clip]{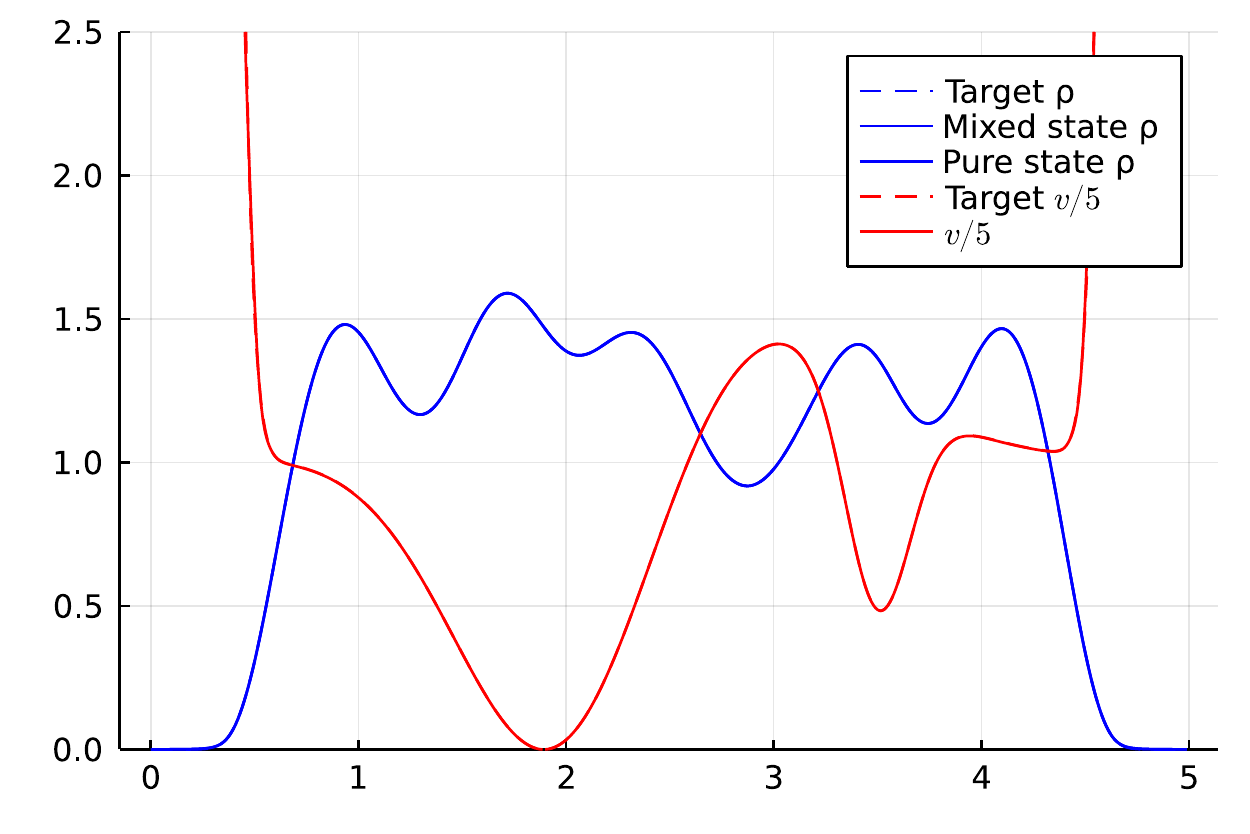}
	\hspace{0.07cm}
\includegraphics[width=0.46\textwidth,trim={0.6cm 0.5cm 0.2cm 0},clip]{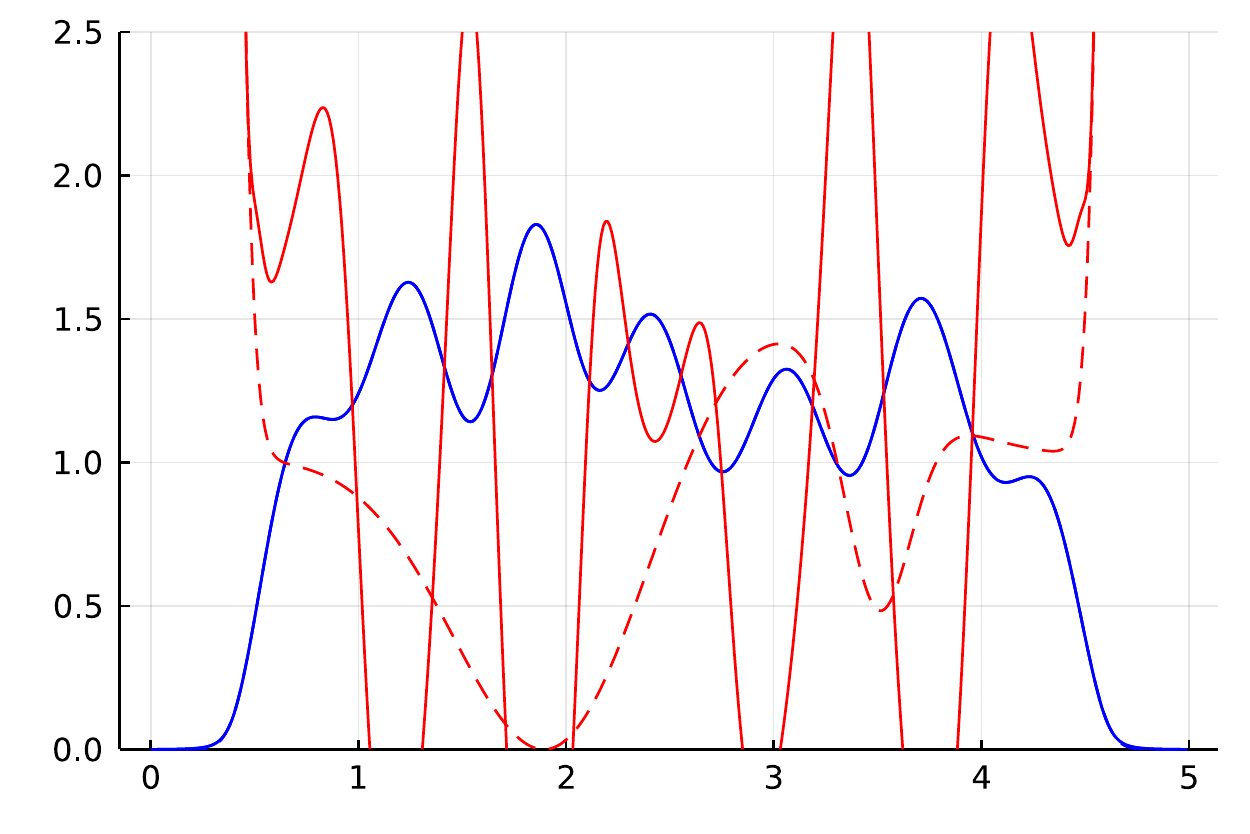} \\
\includegraphics[width=0.47\textwidth,trim={0.6cm 0.5cm 0.2cm 0},clip]{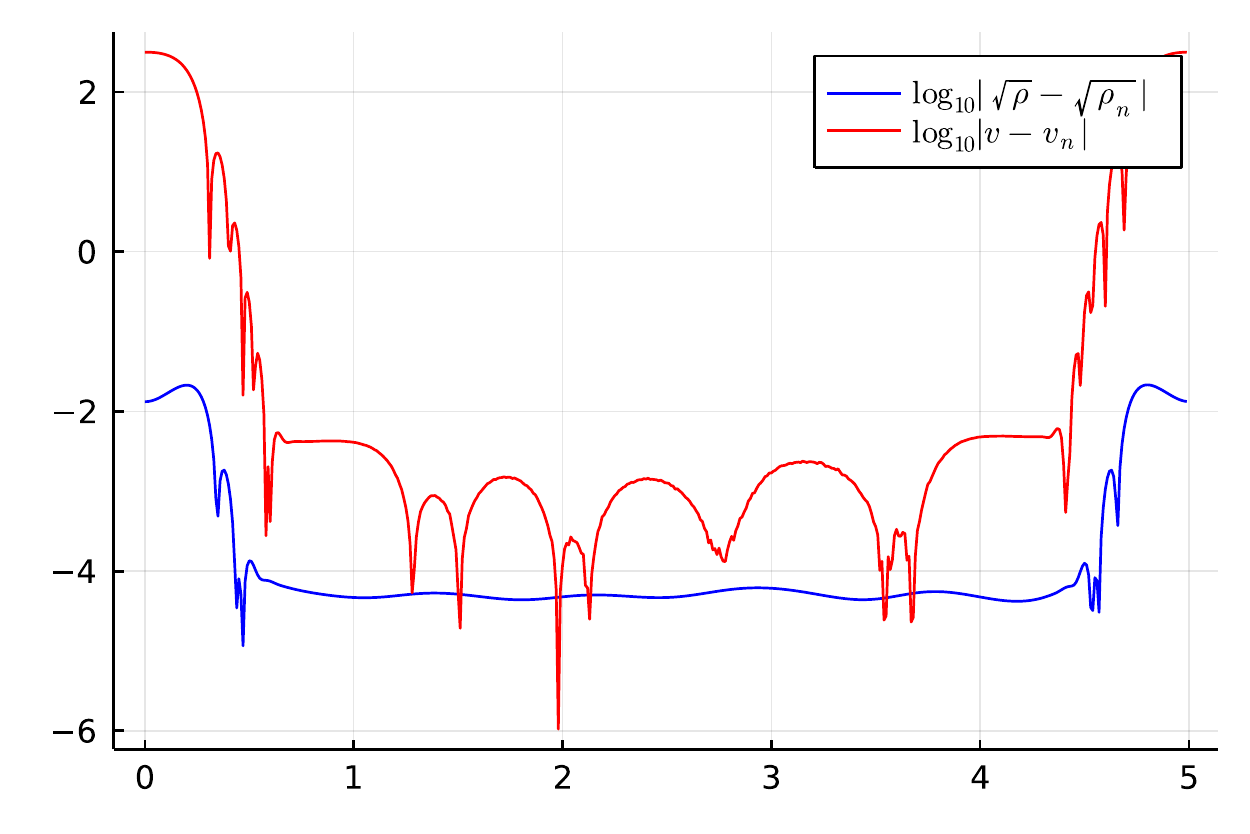}
\includegraphics[width=0.47\textwidth,trim={0.6cm 0.5cm 0.2cm 0},clip]{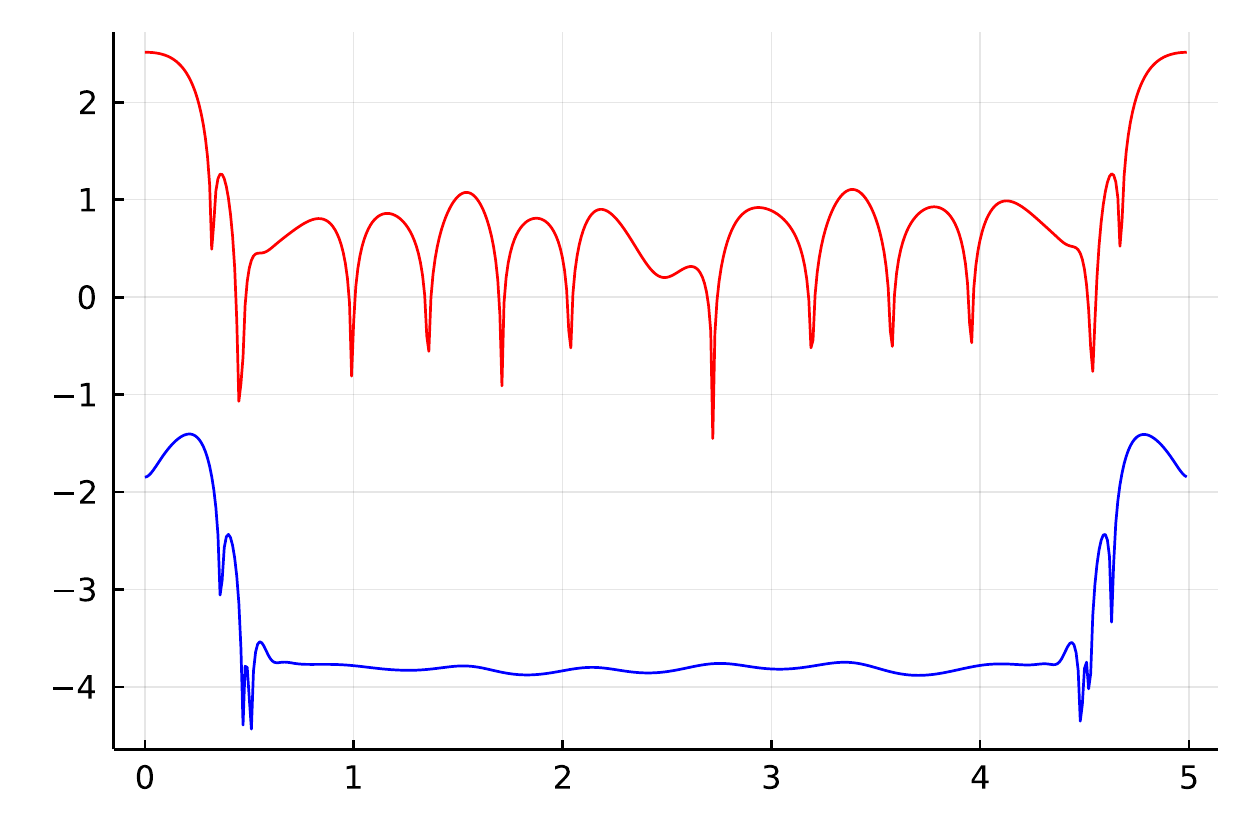} 
\end{center}
\caption{Plot for $d=1$, $N=5$, in the ground state setting $k=0$ on the left and in the third excited state $k=3$ on the right. We start from the same target potential $v$ and launch the algorithm on the same $\ro = \ro_{\p}$ for some normalized $\p \in \kerw$, with $n \simeq 1200$. On the first line, densities and potentials are in different units, and on the second line, we plot $\log_{10} \ab{\sqrt{\ro_n}-\sqrt{\ro}}$ (blue) and $\log_{10} \ab{v_n - v}$ (red) in the same units. We also remark that pure and mixed states densities coincide, consistently with the theory.} \label{f1}
\end{figure}

In the case $d=2$ and on Figure \ref{f2d} we show an example of reconstruction of a potential, where $N=5$ and $k=0$.
\begin{figure}
\begin{center}
\includegraphics[width=6cm,trim={2.5cm 0cm 2.5cm 0},clip]{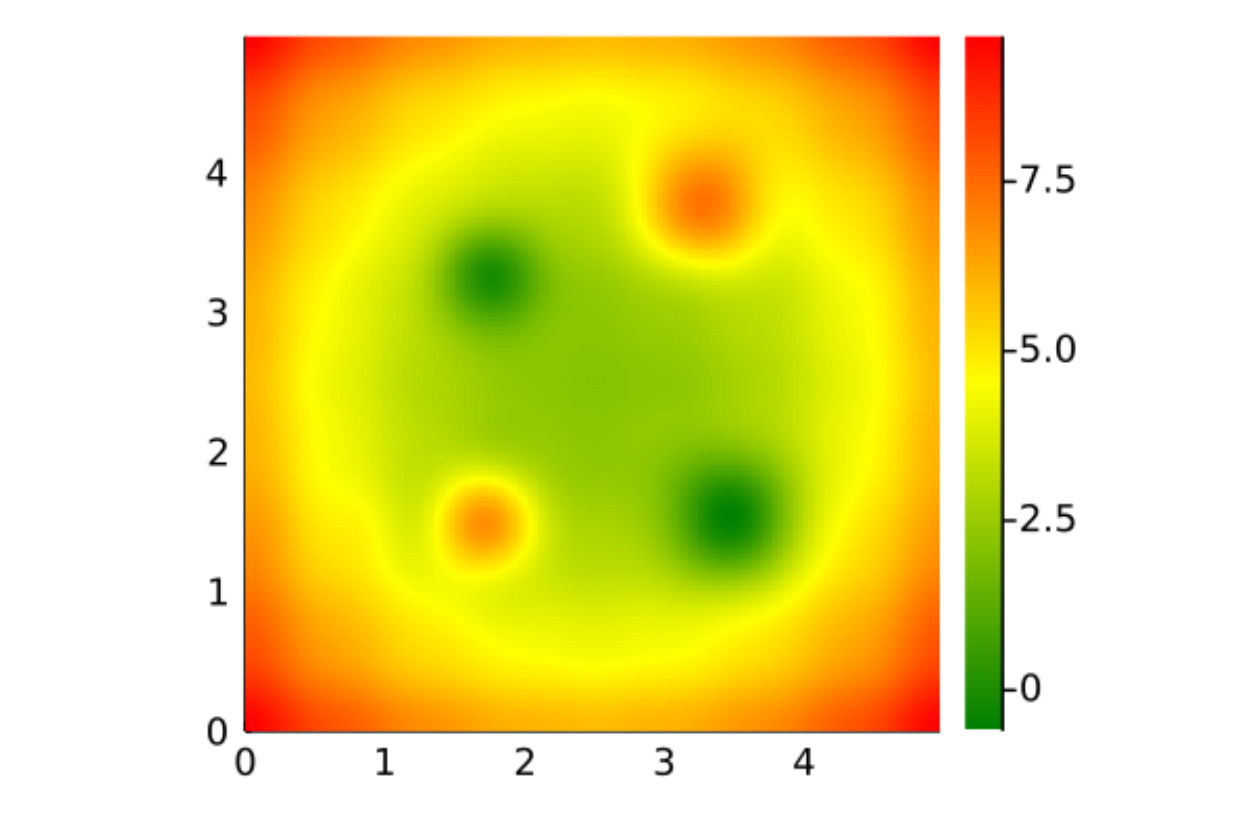}
\includegraphics[width=6cm,trim={2.5cm 0cm 2.5cm 0},clip]{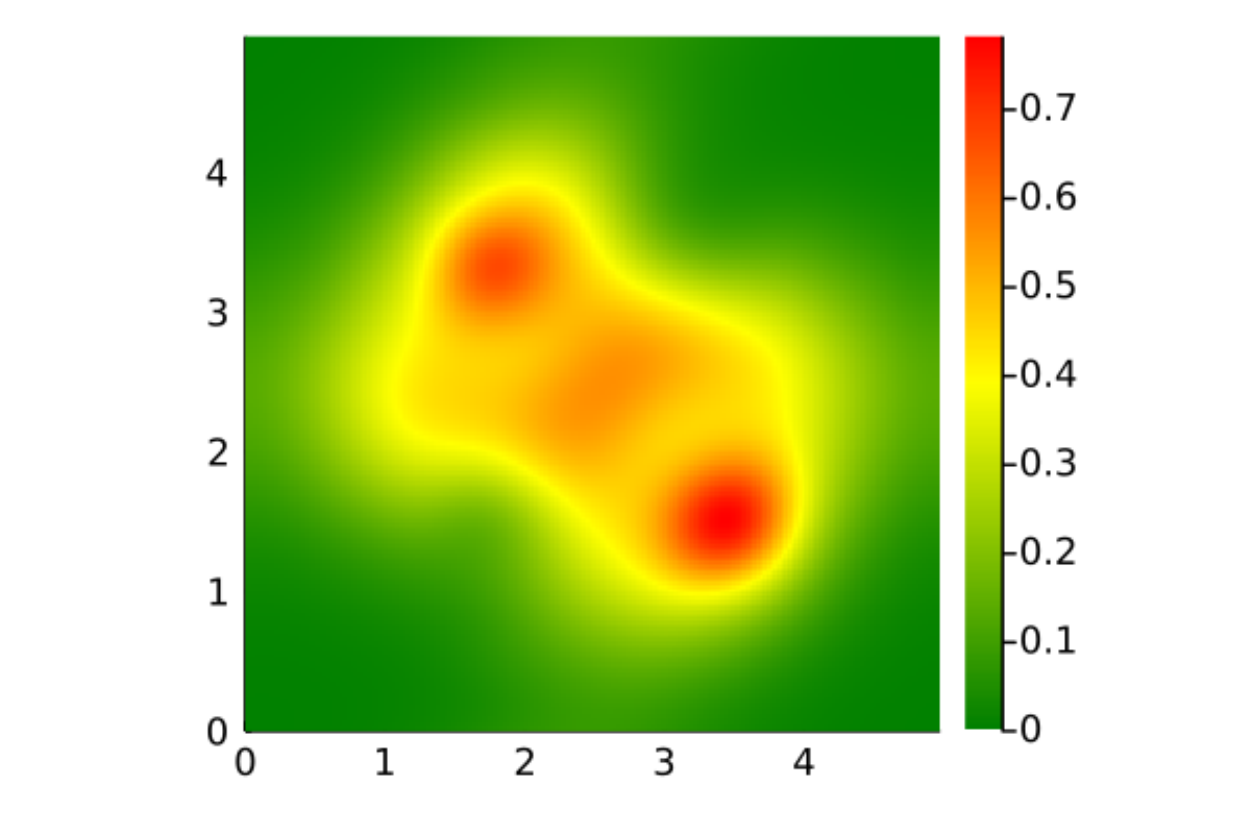} \\
\includegraphics[width=6cm,trim={2.5cm 0cm 2.5cm 0},clip]{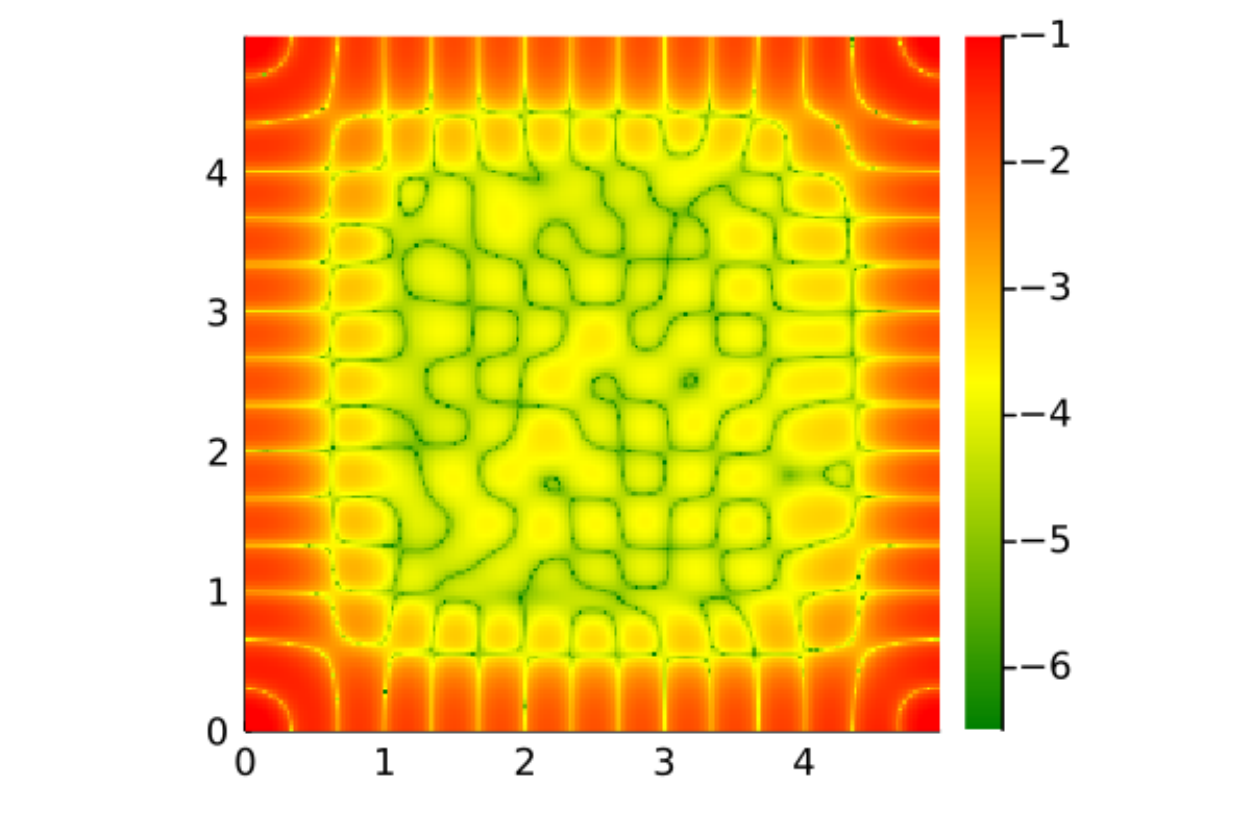}
\includegraphics[width=6cm,trim={2.5cm 0cm 2.5cm 0},clip]{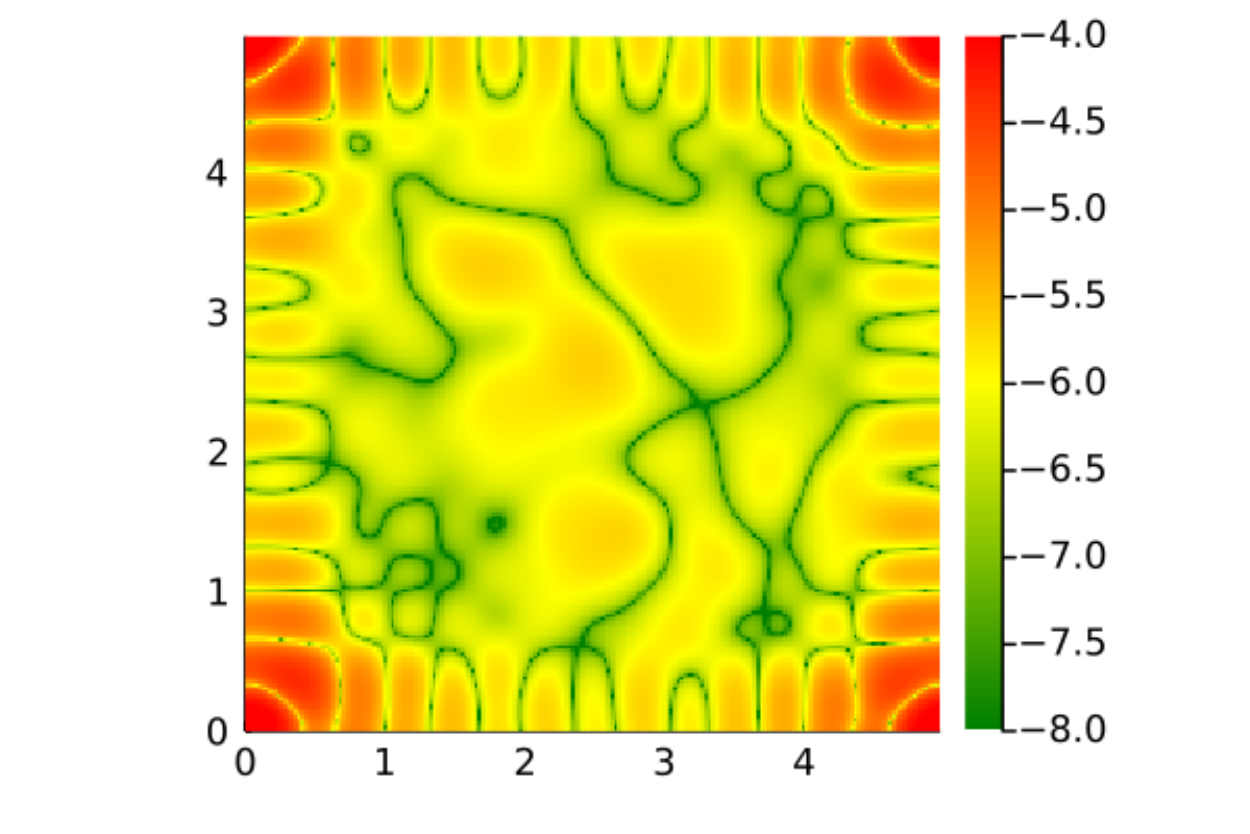} 
\end{center}
	\caption{We initially choose $d=2$, $N=5$, $k=0$, and the confining and non-degenerate potential $v$ plotted on the the top left panel. We compute its density $\ro(v)$, displayed on the top right, and launch the algorithm on it. We display $\log_{10} \ab{v_n-v}$ on the bottom left and $\log_{10} \ab{\sqrt{\ro_n}-\sqrt{\ro(v)}}$ on the bottom right, in their natural units. The length of the squares is 4 in in space units.} \label{f2d}
\end{figure}

On the graphs, the converged density $\ro_n$ is indistinguishable from the target $\ro$ for $d \in \acs{1,2,3}$. Also, $v_n$ is indistinguishable from $v$ for $k=0$ in the regions where $\ro(x)/(\max \ro) \ge 10^{-3}$, and $v_n$ is never close to $v$ when $k \ge 1$. Hence we do not plot $\ro_n$ and $v_n$ anymore. We can obtain arbitrary precision on $\ab{v_n-v}$ in the regions where $\ro$ is not \apo{small}, and as expected, we see that the precision on potentials is much lower than the precision on densities, which is linked to the local weak-strong continuity of the map $v \mapsto \ro$ proved in \cite{Garrigue21}.

Moreover, when $k=0$, convergence is muh easier, and faster, because $\ger{0}$ is concave.

\subsection{Pure states, degeneracies and Levy-Lieb}

We saw that our algorithm converges to some potential having a $k\expo{th}$ bound \textit{mixed} state representing $\ro$. Now we address the problem of finding a \textit{pure} $k\expo{th}$ bound state representing $\ro$. The situation is very different as we choose $d \in \acs{1,2}$ or $d=3$.

\subsubsection{$d=1$}
For one-dimensional systems, Theorem \ref{kspropi} $iv)$ justifies the existence of pure states representing $\ro$. 

\subsubsection{$d=2$}

For two-dimensional systems, there is no non-degenerate theorem, and for instance the harmonic oscillator has arbitrary large degeneracy as we increase the energy level we consider. However as for $d=1$, we numerically remark that \eqref{passt} equals \eqref{passtp} for any potential, and we always obtain pure-state representability, for any $k$.

We do not have a theoretical proof of this fact, contrarily to $d=1$. But roughly speaking, we think that this is due to the fact that the only relevant degeneracies for our problem come from the spherical Laplacian. For $d=2$, this operator is defined on the one-dimensional circle $S^1$ and has only two-fold degeneracies, but the set of mixed states on a real vector space of dimension 2 is equal to the set of pure states on this same vector space, as showed in Theorem \ref{kspropi} $iii)$. Then, for instance degeneracies arising in the harmonic oscillator at the one-body level should be unessential degeneracies for our problem. 

\subsubsection{$d=3$}

In the proof of \cite[Theorem 3.4]{Lieb83b}, Lieb identified a class of densities $\ro$ such that $F\ex{0}\ind{mix}(\ro) < F\ex{0}(\ro)$, using radial symmetry and degeneracies of spherical orbitals in dimension $3$. Since we believe that those functionals are continuous (see Conjecture \ref{conjcont}), this indicates that some densities are not (approximately) pure-state representable. Numerically, we verify it on Lieb's example, for $\ro(x) = \pa{2\pi \sigma}^{-3/2} e^{- \ab{x}^2/(2 \sigma^2)}$, $\sigma = 1/10$, $k=0$, $N=2$,
 \begin{align*}
	 \mymin{\p \in \vect_{\C} \cF(v_n) \\ \int_{\Omega^N} \ab{\p}=1} \indic_{\ab{\cE_{v_n}\pa{\p}-\exc{k}(v_n)} \le \cT_n} \nor{\ro_{\p} - \ro}{L^2(\Omega)} \gtrsim 10^{-1.5},
 \end{align*}
for any $n \in \N$, while the distance with mixed states \eqref{passt} is arbitrarily small as $n \ra +\ii$, confirming that the algorithm converges. Since the potential to which we converge is unique, we conclude that there is no pure ground state representing this $\ro$. 

Of course, we can also find infinitely many densities which are $v$-representable with pure states, but a general necessary and sufficient condition for a density to be $v$-representable with pure states seems out of reach. See \cite{EngEng83} for a discussion on this problem and for necessary conditions.

For $k \ge 1$ and $d=3$, we did not solve the complete problem of pure-state representability \eqref{fpure}, because we would have to test all maximizers, this set of maximizer must be very large and finding it should not be easy. We however think that the conclusion would be the same as for $k=0$.

\subsection{Convergence to the Thomas-Fermi potential}\label{sub:TF}\tx{ }

Let us take $k \in \N$, some density $\ro \ge 0$ such that $\sqrt{\ro} \in H^1(\Omega)$ and $\int_{\Omega} \ro = 1$, and let us denote by $V\ex{k}_N$ the inverse potential of $N \ro$, which is an $N$-particle density. When $N \rightarrow +\infty$, we expect it to converge to the Thomas-Fermi potential, that is
\begin{align}\label{eq:conv_TF}
\f{V\ex{k}_N}{N^{\f 2d}} \rightarrow -c\ind{TF}  \ro^{\f 2d}
\end{align}
where $c\ind{TF}$ is defined in \eqref{eq:TF}. A version of this last statement via the direct problem can be found in \cite[Theorem 1.2]{FouLewSol18}. On Figure \ref{fig:TF} we present the rescaled inverse potentials $V_N / (10 N^{\f 2d})$, for $k=0$, and we remark that they converge very quickly to the Thomas-Fermi potential. Indeed, they are already very close for $N=15$, and become indistinguishable for $N \ge 20$. We numerically confirm that the convergence \eqref{eq:conv_TF} holds for any $k$.
\begin{figure}
\begin{center}
\includegraphics[width=0.99\textwidth]{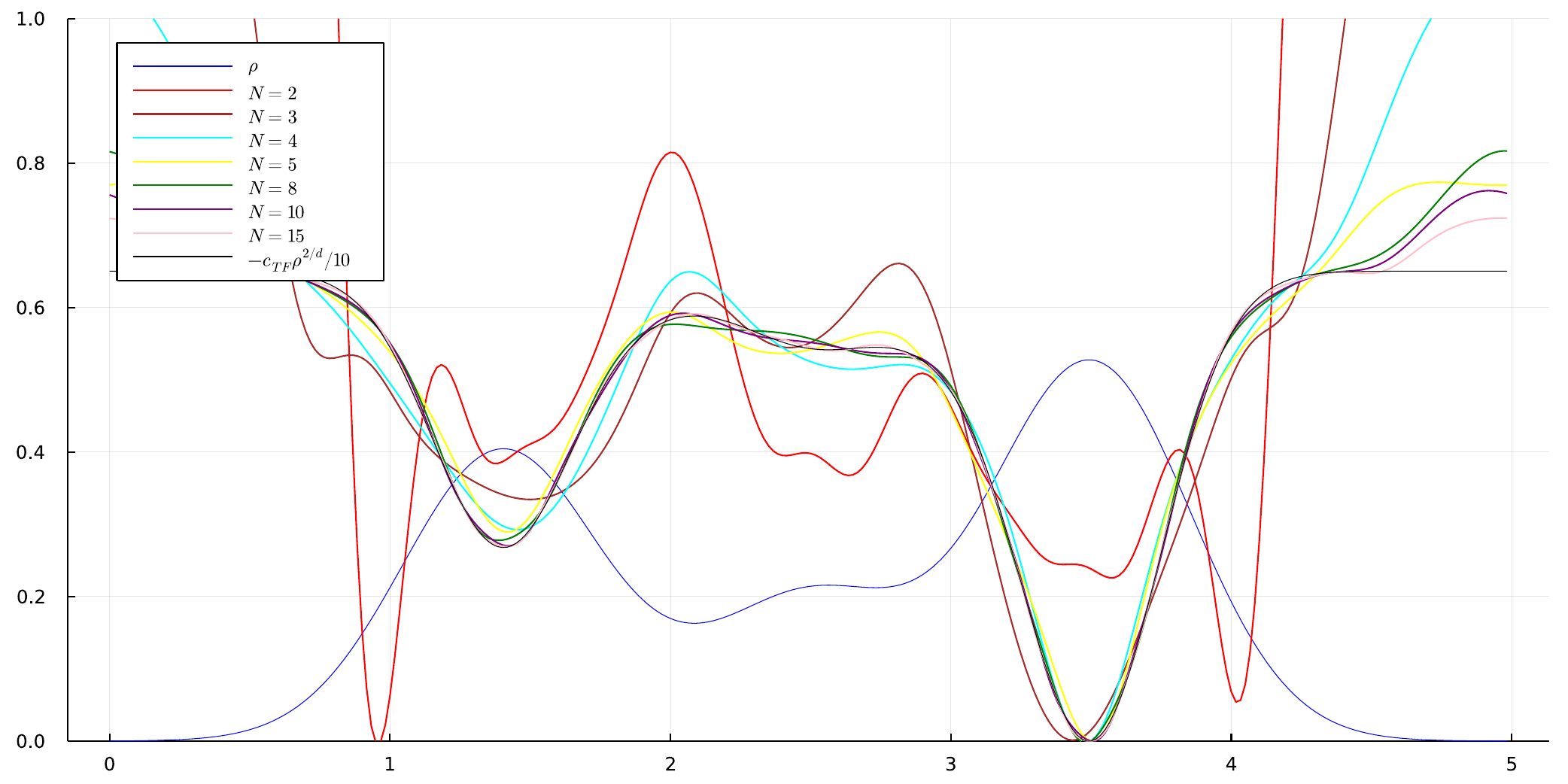}
\end{center}
\caption{For $k=0$, density $\ro$ in blue, inverse potentials $V_N / N^{\f 2d}$ for several values of $N$, and the Thomas-Fermi potential for one particule $-c\ind{TF} \ro^{\f 2d}$ in black, to which the inverse potentials converge when $N$ becomes large.}\label{fig:TF}
\end{figure}

\subsection{Kato cusp}\label{sub:kato_cusp}\tx{ }

An important application in quantum chemistry is for molecules, which one-body eigenstate densities have cusps \cite{Kato57}, that is singularities produced by atomic potentials. In Figure \ref{fig:cusp}, we give an application with a cusp on density of the form $\alpha \exp\pa{-\beta \ab{x - \f{L}{2}}}$, and we plot the inverse potential and the error on densities. It is well-known that for such singularities, the inverse potential becomes singular.
\begin{figure}
\begin{center}
\includegraphics[width=0.49\textwidth,trim={1.5cm 0cm 1cm 0},clip]{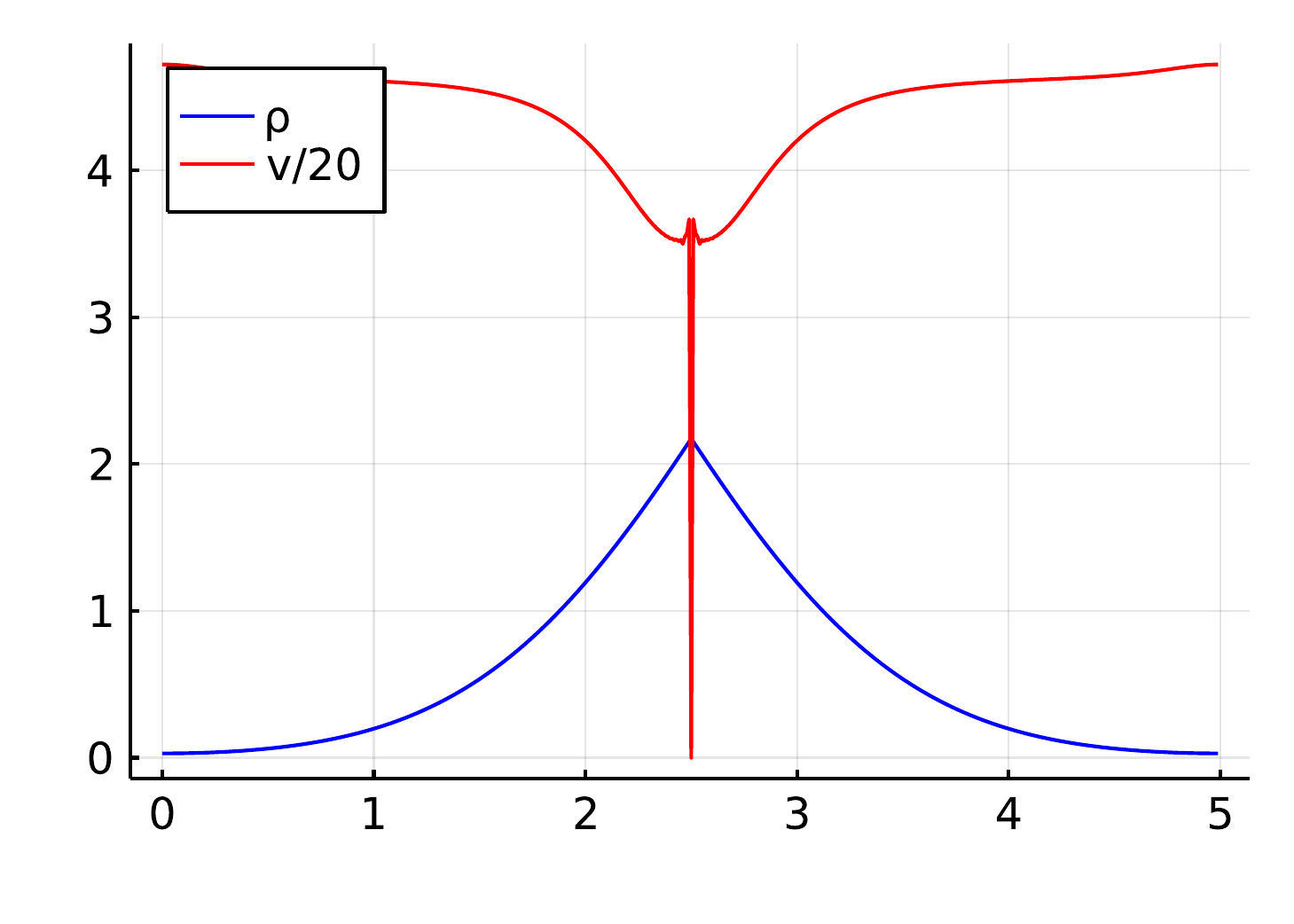}
\includegraphics[width=0.49\textwidth,trim={2.5cm 0cm 1cm 0},clip]{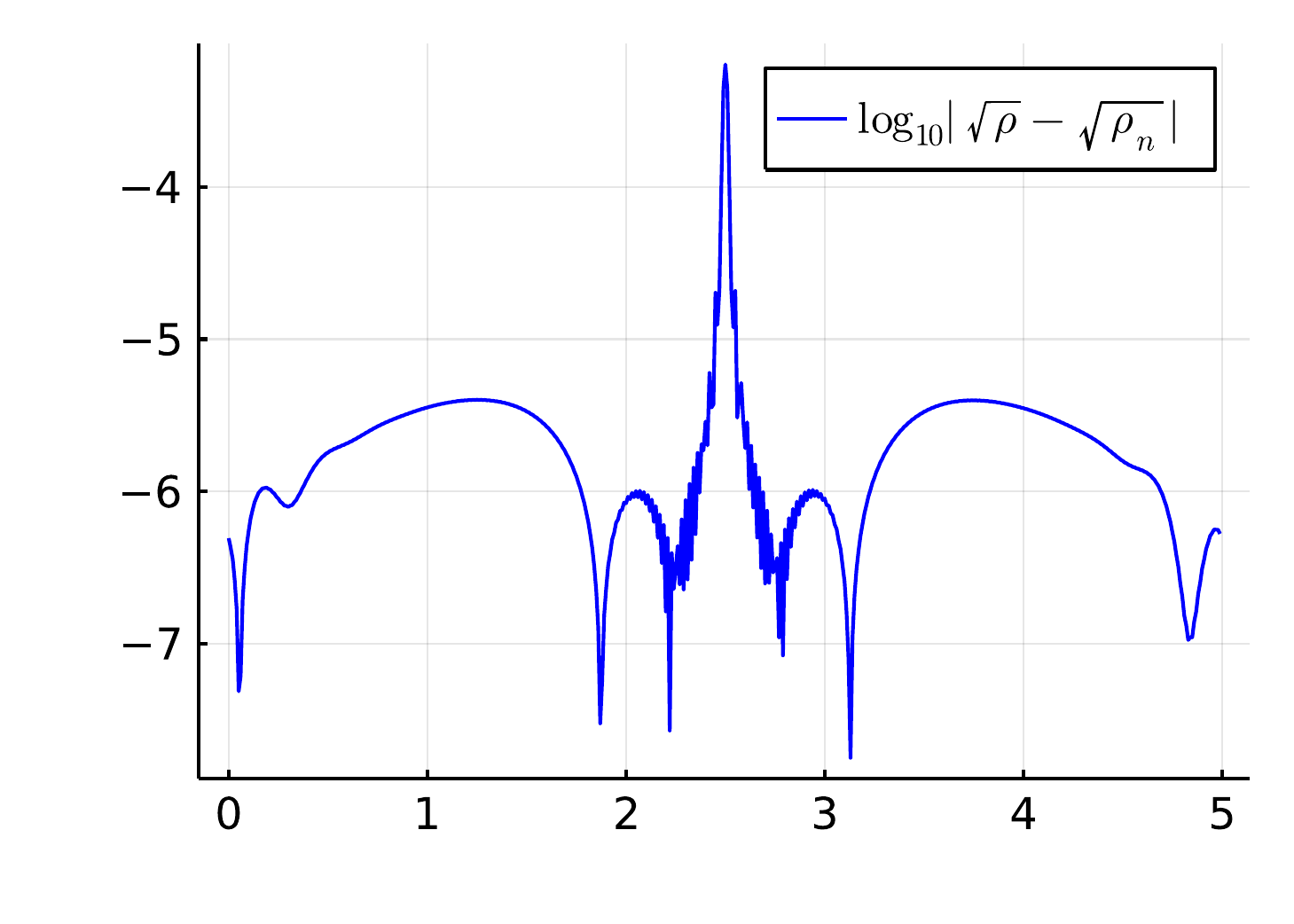}
\end{center}
\caption{For $k=0$, $N=3$, density $\ro$ in blue of the form $\alpha \exp\pa{-\beta \ab{x - \f{L}{2}}}$ where $\alpha,\beta > 0$, having a cusp at $x = L/2$, and its inverse potential in red. On the right, we represent the error $\log_{10} \ab{\sqrt{\ro_n}-\sqrt{\ro}}$.}\label{fig:cusp}
\end{figure}

\subsection{Conclusions}\tx{ }

\bul To obtain a mixed Kohn-Sham potential, one can maximize $\ger{k}$, in all situations. The algorithm presented here is simple and always converges.

\bul For $k \ge 1$, starting from a different potential or changing some parameters of the algorithm leads to very different inverse potentials, emphasizing that many potentials lead to the same density.

\bul For $d\in \acs{1,2}$, inverse potentials with pure states exist, hence this suggests that $F\ex{k}\ind{mix}(\ro) = F\ex{k}(\ro)$ and that the set of $v$-representable densities with pure states \eqref{rep} is dense in the space of densities

Based on the nature of the degeneracy (essentially one-body, coming from the spherical Laplace operator, ...), and on its number, for $d=3$ and $k=0$, some densities do not have $k\expo{th}$ pure states representing them and this indicates that \eqref{rep} is not dense in the set of densities. However, Section \ref{purek}, suggests that given $\ro$, there exists $k$ such that $\ro$ is approximately a $k\expo{th}$ bound pure state density.

\bul Degeneracies have to be taken into account in the algorithm, not only when the inverse potential is non-degenerate, but all along the procedure, because intermediate degenerate potentials can block the algorithm, corresponding to eigenvalues crossings. The only case where we do not need to take them into account is when $(d,k)=(1,0)$, because then the $N$-body ground level is non-degenerate by the non-degeneracy theorem (Proposition \ref{ndt}). Using our algorithm, we launched simulations on hundreds of densities, choosed to be sums of gaussians with random parameters, and were surprised by the fact that for $d \in \acs{1,2,3}$, \apo{most of the time}, the inverse potential is degenerate, in the sense that it gives rise to degenerate eigenvalues at the $k\expo{th}$ level. Perturbation of potentials generically lifts degeneracies, but at the level of this inverse problem, perturbation of a density \textit{does not} lift degeneracies. Hence, the fact that \apo{generically} inverse potentials are degenerate can seem counterintuitive. In the SCF procedure, perturbation does not lift degeneracies either, as shown in \cite{CanMou14}. Finally, we remark that we always have $k= m_k^v$ for $v_n$ close to optimality.

\section{Proofs}
First of all, we recall the Sobolev injections expressed for our densities. If $\Omega \subset \R^d$ has uniformly Lipschitz boundary, then for any $\ro \ge 0$, $\sqrt \ro \in H^1\pa{\Omega}$,
\begin{align}\label{ineq:sob_inj}
	\nor{\ro}{L^\infty(\Omega)} \le c_\Omega \nor{\sqrt{\ro}}{H^1(\Omega)}^2 &\text{ if } d=1, \nonumber  \\
	\nor{\ro}{L^s(\Omega)} \le c_{\Omega,s}\nor{\sqrt{\ro}}{H^1(\Omega)}^2 &\text{ for any $s \in [1,+\infty)$ if } d \in \{1,2\},  \\
	\nor{\ro}{L^{\f{d}{d-2}}(\Omega)} \le c_{\Omega,d}\nor{\sqrt{\ro}}{H^1(\Omega)}^2 &\text{ if } d \ge 3, \nonumber
\end{align}
where $c_{\Omega}$, $c_{\Omega,s}$ and $c_{\Omega,d}$ do not depend on $\ro$. See \cite[Corollary 11.9, Exercise 11.26, Exercise 11.37, Theorem 12.15]{Leoni17} for instance. If $d=1$, we also have that 
\begin{align*}
\ab{\sqrt{\ro(x)}-\sqrt{\ro(y)}} \le c_{\Omega} \sqrt{\ab{x-y}},
\end{align*}
for any $x,y \in \Omega$, and if $\Omega$ is unbounded, $\ro(x) \rightarrow 0$ when $\ab{x} \rightarrow +\infty$. To get the second inequality for $d =1$, we used that $\nor{\sqrt{\ro}}{H^{\f 12}(\Omega)} \le \nor{\sqrt{\ro}}{H^1(\Omega)}$.

\subsection{Proofs of Proposition \ref{sollin} and Theorem \ref{kspropi}}\label{ssub:proofs_sollin_kspropi}
We start by presenting a remark. Let $Q$ be a real vector space of real wavefunctions, then for $\Gamma \in \cans(Q,\Omega)$, we have 
 \begin{align*}
	\ro_{\Gamma} = \ro_{\re \Gamma}.
 \end{align*}
Indeed, self-adjointness $\Gamma = \Gamma^*$ implies $\Gamma(x_1,\dots,x_N;x_1,\dots,x_N) \in \R$. In case $\dim Q$ is finite, take an orthonormal frame $\pa{\p_i}_{i \in I}$ of real wavefunctions of $Q$, write the hermitian matrix $\tx{mat}_Q \Gamma =: \pa{\Gamma_{ij}}_{1 \le i,j \le n}$, then $\Gamma_{ii} \in \R$ and we have
\begin{align*}
\ro_{\Gamma} = \sum_{i=1}^n \Gamma_{ii} \ro_{\p_i} + 2 N \sum_{1 \le i < j \le n} \pa{\re \Gamma_{ij}} \int_{\Omega^{N-1}} (\p_i\p_j)(x,x_2,\dots) \d x_2 \cdots \d x_N.
\end{align*}

\begin{proof}[Proof of Proposition \ref{sollin}]
First, for $d \ge 3$, $s \ge d/2$ hence $s/(s-1) \le d/(d-2)$ and since $\ro \in L^{\f{d}{d-2}}(\Omega)$ by \eqref{ineq:sob_inj} and $\ro \in L^1(\Omega)$, then $\ro \in L^{\f{s}{s-1}}(\Omega)$. For $d \in \{1,2\}$, $s /(s-1) \ge 1$ and we also have that $\ro \in L^{\f{s}{s-1}}(\Omega)$ by \eqref{ineq:sob_inj}. As explained in \cite[Theorem 1.6]{Garrigue21}, we have
\begin{align}\label{expg}
{^+}\delta_v \ger{k}(u) & = \mymin{\gr \subset \cD\ex{k}_N(v)\\ \dim Q = k-m_k^v+1}\mymax{ \p \in \vect_{\C} Q \\ \int_{\Omega^N} \ab{\p}^2 =1 } \int_\Omega  \pa{\ro_{\p} -\ro}u  \\
& = \mymax{\gr \subset \cD\ex{k}_N(v) \\ \dim \gr = M_k^v-k}\mymin{ \p \in \vect_{\C} Q \\ \int_{\Omega^N} \ab{\p}^2 =1 } \int_\Omega \pa{\ro_{\p} -\ro}u, \nonumber
\end{align}
where $\cD\ex{k}_N(v)$ is defined in \eqref{defd} and $\vect_{\C} Q$ is the complex vector space built on the real vector space $Q$. We will need the following classical lemma.
\begin{lemma}\label{enm}
For any real linear subspaces $A$ and $B$ of $H^1\ind{a}(\Omega^{N},\R)$, where $B$ is finite dimensional, and any potential $v \in (\lpi)(\Omega,\R)$, we have 
\begin{align*}
\myinf{\Gamma \in \cS\pa{A} \\ \Gamma \ge 0, \tr \Gamma = 1}\hspace{-0.2cm} \cE_v\pa{\Gamma} = \myinf{\p \in \vect_{\C} A \\ \int_{\Omega^N} \ab{\p}^2 =1}  \cE_v\pa{\p}, \hs\hs\hs\hs\hs\hs  \mymax{\Gamma \in \cS\pa{B} \\ \Gamma \ge 0, \tr \Gamma = 1} \hspace{-0.2cm}\cE_v\pa{\Gamma} = \mymax{\p \in \vect_{\C} B \\ \int_{\Omega^N} \ab{\p}^2 =1} \hspace{-0.2cm} \cE_v\pa{\p}.
\end{align*}
\end{lemma}
We use the notation $\cS\pa{A}$ in case of infinite-dimensional vectors spaces $A$, by the natural extension of the definition in finite dimension. The statement of Lemma \ref{enm} can be seen by taking a real mixed state $\Gamma \in \cans(A,\Omega)$, decomposing it into $\Gamma = \sum_{i=1}^{+\ii} \lambda_i \proj{\p_i}$ where $\p_i \in A$ is an orthonormal basis of real wavefunctions of $A$, $\lambda_i \in \R_+$ and $\sum_{i=1}^{+\ii} \lambda_i = 1$. We evaluate $\cE_v(\Gamma) = \sum_{i=1}^{+\ii} \lambda_i \cE_v(\p_i) \ge \inf_{\p \in A, \int_{\Omega^N}\ab{\p}^2 = 1} \cE_v(\p)$. This is similar for the max.

Now we use the theorem of existence of saddle points \cite[Theorem 49.A p459]{Zeidler3} for the Lagrangian $L(\Gamma,u) := \int_\Omega u \pa{\ro_{\Gamma}-\ro}$, affine in its variables, to commute
\begin{align*}
\mysup{u \in L^s(\Omega,\R) \\ \nor{u}{L^s} \le 1}\mymin{\Gamma \in \cS\pa{A} \\ \Gamma \ge 0, \tr \Gamma = 1} \int_\Omega u \pa{\ro_{\Gamma}-\ro} = \mymin{\Gamma \in \cS\pa{A} \\ \Gamma \ge 0, \tr \Gamma = 1} \mysup{u \in L^s(\Omega,\R) \\ \nor{u}{L^s} \le 1}\int_\Omega u \pa{\ro_{\Gamma}-\ro},
\end{align*}
where $A$ is any finite-dimensional real vector space. But since the two suprema are positive, this also holds when we optimize over 
\begin{align*}
\acs{u \in L^s(\Omega,\R), \nor{u}{L^s(\Omega)} = 1}.
\end{align*}
We define $s' := s/(s-1)$. Next, for $f \in L^{s'}(\Omega)$,
\begin{align*}
\mysup{g \in L^s\pa{\Omega,\R} \\ \nor{g}{L^s}=1} \int_\Omega fg = \nor{f}{L^{s'}(\Omega)}.
\end{align*}
Indeed, considering the Lagrangian $L(g,\lambda) = \int_\Omega fg - \lambda\pa{\int_\Omega g^s - 1}$, and searching for its extremizer, we obtain $f = \lambda s \ab{g}^{s-1} \sgn g$. The condition $\nor{g}{L^s(\Omega)}=1$ yields $\lambda s = \nor{f}{L^{s'}(\Omega)}$, and the extremizer is 
\begin{align*}
g = \ab{\f{f}{\nor{f}{L^{s'}(\Omega)}}}^{{s'}-1} \sgn{f}.
\end{align*}
Hence for any $0 \le \Gamma \in \cS\pa{A}$ such that $\tr \Gamma = 1$,
\begin{align*}
\mysup{u \in L^s(\Omega,\R) \\ \nor{u}{L^s}=1}\int_\Omega u \pa{\ro_{\Gamma}-\ro} = \nor{\ro_{\Gamma}-\ro}{L^{s'}(\Omega)}
\end{align*}
is attained by \eqref{maxu}. 
 
Using all the previous steps, we deduce that
\begin{align*}
&\mysup{u \in L^s(\Omega,\R) \\ \nor{u}{L^s}=1} {^+}\delta_v G\ex{k} (u) \\
& \bhs = \mysup{u \in L^s(\Omega,\R) \\ \nor{u}{L^s}=1} \mymax{\gr \subset \cD\ex{k}(v) \\ \dim \gr = k-m_k^v} \mymin{\p \in \vect_{\C} \pa{\gr^{\perp}\cap \cD\ex{k}(v)}\\ \int_{\Omega^N} \ab{\p}^2 = 1} \int_\Omega u \pa{\ro_{\p}-\ro} \\
& \bhs= \mysup{u \in L^s(\Omega,\R) \\ \nor{u}{L^s}=1} \mymax{\gr \subset \cD\ex{k}(v) \\ \dim \gr = k-m_k^v} \mymin{\Gamma \in \cS\pa{\gr^{\perp}\cap \cD\ex{k}(v)} \\ \Gamma \ge 0, \tr \Gamma = 1} \int_\Omega u \pa{\ro_{\Gamma}-\ro} \\
& \bhs=  \mymax{\gr \subset \cD\ex{k}(v) \\ \dim \gr = k-m_k^v} \mysup{u \in L^s(\Omega,\R) \\ \nor{u}{L^s}=1}\mymin{\Gamma \in \cS\pa{\gr^{\perp}\cap \cD\ex{k}(v)} \\ \Gamma \ge 0, \tr \Gamma = 1} \int_\Omega u \pa{\ro_{\Gamma}-\ro} \\
&\bhs =  \mymax{\gr \subset \cD\ex{k}(v) \\ \dim \gr = k-m_k^v} \mymin{\Gamma \in \cS\pa{\gr^{\perp}\cap \cD\ex{k}(v)} \\ \Gamma \ge 0, \tr \Gamma = 1} \mysup{u \in L^s(\Omega,\R) \\ \nor{u}{L^s}=1}\int_\Omega u \pa{\ro_{\Gamma}-\ro} \\
&\bhs =  \mymax{\gr \subset \cD\ex{k}(v) \\ \dim \gr = k-m_k^v} \mymin{\Gamma \in \cS\pa{\gr^{\perp}\cap \cD\ex{k}(v)} \\ \Gamma \ge 0, \tr \Gamma = 1} \nor{\ro_{\Gamma}-\ro}{L^{s'}(\Omega)} \\
&\bhs =  \mymax{\gr \subset \cD\ex{k}(v) \\ \dim \gr = M_k^v-k+1} \mymin{\Gamma \in \cS\pa{\gr} \\ \Gamma \ge 0, \tr \Gamma = 1} \nor{\ro_{\Gamma}-\ro}{L^{s'}(\Omega)},
\end{align*}
where $Q$ are real vector spaces.
\end{proof}

\begin{proof}[Proof of Theorem \ref{kspropi}]\tx{ }

\bul $i)$ $a) \implies c)$ The state $\Gamma$ is a $k^{\tx{th}}$ mixed bound state of $v$ and has density $\ro_{\Gamma} = \ro$, hence in the energy we can restrict the optimization search to states having density $\ro$,
\begin{align*}
	\cE_v\pa{\Gamma} &= \exc{k}(v) = \mysup{A \subset H^1\ind{a}(\Omega^{N})\\ \dim A = k+1} \; \myinf{\Lambda \in \cans\pa{A^{\perp},\Omega}} \cE_v\pa{\Lambda} \\
&  = \mysup{A \subset H^1\ind{a}(\Omega^{N})\\ \dim A = k+1} \; \myinf{\Lambda \in \cans\pa{A^{\perp},\Omega} \\ \ro_{\Lambda} = \ro}  \cE_v\pa{\Lambda}  = F\ex{k}\ind{mix}(\ro) + \int_\Omega v \ro  \\
&= \sup_{u \in (\lpi)(\Omega,\R)} \pa{\ger{k}(u) + \int_\Omega v \ro}.
\end{align*}
On the other hand, $\ger{k}(v) = \cE_v\pa{\Gamma} - \int_\Omega v \ro$, thus $v$ maximizes $\ger{k}$.

$b) \implies a)$ Since $v$ is a local maximizer and $\spf{k}$ is open, then for $u$ close enough to $v$, we have $\ger{k}(u) \le \ger{k}(v)$ and $u \in \spf{k}$, thus 
\begin{align*}
\mysup{u \in (\lpi)(\Omega,\R) \\ \nor{u}{\lpi} = 1} {^+}\delta_v \ger{k}(u) \le 0
\end{align*}
Now we know that if $u \in L^p(\Omega,\R)$, then $\nor{u}{\lpi} = \nor{u}{L^p}$, hence 
\begin{align*}
\acs{u \in L^p(\Omega,\R), \nor{u}{L^p}=1} \subset \acs{u \in (\lpi)(\Omega,\R), \nor{u}{\lpi} = 1},
\end{align*}
so
\begin{align*}
\mysup{u \in L^p(\Omega,\R) \\ \nor{u}{L^p} = 1} {^+}\delta_v \ger{k}(u) \le \mysup{u \in (\lpi)(\Omega,\R) \\ \nor{u}{\lpi} = 1} {^+}\delta_v \ger{k}(u) \le 0
\end{align*}
 and we conclude by \eqref{lao} that for any $Q \subset \Ker_{\R} \bpa{\hn(v)-\exc{k}(v)}$, with $\dim Q = M_k^v-k+1$, there exists a $0 \le \Gamma = \Gamma^{\tx{T}} \in \cS(Q)$ such that $\tr \Gamma = 1$ and $\ro_{\Gamma} = \ro$.

Remark when $k=0$. Let $v$ be an optimizer of $\ger{0}$. We know that there exists a mixed state $\Gamma_{\ro}$ such that $\ro_{\Gamma_{\ro}} = \ro$ and $\tr \hn(0) \Gamma = \ger{0}(v)$ by \cite[Corollary 4.5]{Lieb83b}. This implies that $\tr \hn(v) \Gamma_{\ro} = \exc{0}(v)$. By diagonalizing $\Gamma_{\ro}$ similarly as in the proof of Theorem \ref{cococo}, we can show that it is a ground mixed state for $v$.

\bul $ii)$ We also know that ${^+}\delta_v \ger{m} \le {^+}\delta_v \ger{n}$ for any $m,n \in \acs{m_k^v, \dots, M_k^v}$ such that $m \le n$. Hence for $\ell \in \acs{m_k^v,\dots, k}$, 
\begin{align*}
\mysup{u \in (\lpi)(\Omega,\R) \\ \nor{u}{\lpi} = 1} {^+}\delta_v \ger{\ell}(u) \le \mysup{u \in (\lpi)(\Omega,\R) \\ \nor{u}{\lpi} = 1}{^+}\delta_v \ger{k}(u) =0
\end{align*}
and $v$ is a local maximizer of $\ger{\ell}$.

We can see from \eqref{expg} that 
\begin{align*}
{^+}\delta_v \ger{k}(-u) = -{^+}\delta_v \ger{M_k^v+m_k^v-k}(u)
\end{align*}
for any direction $u \in (L^p+L^{\ii})(\Omega,\R)$.  Thus
 \begin{align*}
-{^+}\delta_v \ger{k}(-u) = {^+}\delta_v \ger{M_k^v+m_k^v-k}(u) \le {^+}\delta_v \ger{k}(u) \le 0,
 \end{align*}
where we used that $M_k^v+m_k^v-k \le k$ and that $v$ is a local maximizer. We deduce that for any direction $u$, ${^+}\delta_v \ger{k}(u) \ge 0$, so ${^+}\delta_v \ger{k} = \d_v \ger{k} = 0$. This yields $\bpa{\d_v \exc{k}} u = \int_\Omega u \ro$ and we can conclude.

	If $v$ is a local minimizer and $k \ge (m_k^v+M_k^v)/2$, by a similar reasoning we have $\d_v \ger{k} = 0$, and by \eqref{lao} there is a $k\expo{th}$ bound mixed state of $v$ such that $\ro_{\Gamma} = \ro$. As we saw in $i)$, this implies that $v$ is a global maximizer, but it cannot be a local minimizer at the same time.

\bul $iii)$ Take real orthonormal vectors $\vp,\phi$ we write $\Gamma =  \mat{a & c \\ c & b} \in \R^{2 \times 2}$ in this basis. The condition $\Gamma \ge 0$ is equivalent to $a \ge 0, b \ge 0, ab \ge c^2$, and $\tr \Gamma = a + b = 1$. Then 
 \begin{align*}
 \ro_{\Gamma} & = a \ro_{\vp} + b \ro_{\phi} + 2 N c \int_{\Omega^{N-1}} \vp \phi.
 \end{align*}
 In the case of pure states, we have $\psi = \alpha \vp + \beta \phi$, with $\ab{\alpha}^2 + \ab{\beta}^2 = 1$ so we can take the parametrization $\alpha = \pa{\cos t} e^{i \eta }$, $\beta = \pa{\sin t}e^{i \pa{\eta + \theta}}$, and
\begin{align*}
\ro_{\psi} & = \pa{\cos t}^2 \ro_{\vp} + \pa{\sin t}^2  \ro_{\phi} + 2 \pa{ \cos t} \pa{\sin t} \pa{\cos \theta}  N \int_{\Omega^{N-1}} \vp \phi.
\end{align*}
Since
 \begin{multline*}
 \acs{ \pa{ (\cos t)^2, (\sin t)^2, \cos t \sin t \cos \theta} \st t, \theta \in \seg{0,2\pi}} \\
 = \acs{\pa{a,b,c} \st a,b,c \in \R, c^2 \le ab, a+ b = 1},
 \end{multline*}
the two spanned spaces of density are equal
\begin{multline*}
\acs{\ro_{\psi} \st \psi \in \vect_{\C} \pa{\vp,\phi}, \mediumint_{\Omega^N} \ab{\psi}^2 = 1} \\
= \acs{\ro_{\Gamma} \st \Gamma \in \cS \pa{\vect_{\R} (\vp,\phi)}, \Gamma \ge 0, \tr \Gamma = 1}.
\end{multline*}
Since we know that $v$ has a $k\expo{th}$ bound mixed state representing $\ro$, it also has a pure one.

\bul $iv)$ We will use a well-known result specific to the dimension one.
\begin{proposition}[Non-degeneracy theorem]\label{ndt}
For $d=1$ and any potential $v \in (L^1 + L^{\ii})(\R,\R)$, every eigenstate of $-\Delta +v$ is non-degenerate.
\end{proposition}
We recall its proof for the convenience of the reader.
\begin{proof}
Let $\psi,\vp \in H^1(\R,\R)$ be normalized to 1 and such that $(-\d^2/\d x^2 + v-E) \phi = 0$ for $\phi \in \acs{\psi,\vp}$, with $E \in \R$. Multiplying the first equation by $\vp$, the second by $\psi$ and substracting, we get $0 = \psi \vp'' - \vp \psi'' = \pa{ \psi\vp' - \vp \psi'}'$. Hence $\psi\vp' - \vp \psi' = c$ for some constant $c \in \R$, but since $\psi,\vp \in L^2$, then $c = 0$. We have thus $\psi'/\psi = \vp'/\vp$ on $\acs{ x \in \R \st \psi \neq 0, \vp \neq 0}$, which has full measure by unique continuation \cite{Garrigue19}. Finally $\psi = a \vp$ with $a = \pm 1$.
\end{proof}
This result shows that the eigenspaces of $\hn(v)$ with $w=0$ can only have coincidental degeneracies, and we apply Proposition \ref{ksprops} $iii)$.
\end{proof}

\subsection{Proofs of Theorems \ref{exiss} and \ref{proproi}}\label{ssub:proofs_exiss_proproi} 
We start by proving the existence of minimizers.
\begin{proof}[Proof of Theorem \ref{exiss}]
We use a tightness argument similar to the one in \cite{Lieb83b}. Let us denote by $\p_n$ a minimizing sequence for $\fll{0}(r)$. Since $w \ge 0$, then 
\begin{align*}
\int_{\Omega} \ab{\na \sqrt{\ro_{\p_n}}}^2 \le  \int_{\Omega^N} \ab{\na \p_n }^2 \le \cE_0(\p_n) \ra \fll{0}(r),
\end{align*}
by the Hoffman-Ostenhof inequality. Then, $\p_n$ is bounded in $H^1(\Omega^N)$ and by the theorem of Banach-Alaoglu, there exists $\p_{\ii} \in H^1\ind{a}(\Omega)$ such that $\p_n \wra \p_{\ii}$ in $H^1(\Omega^N)$, and $\sqrt{\ro_{\p_n}} \wra \sqrt{\ro_{\ii}}$ in $H^1(\Omega)$ by \cite[Theorem 1.3]{Lieb83b}. At this step, $\ro_{\ii}$ and $\p_{\ii}$ are not related. By summing all the constraints on the density and using that $\indic_{\Omega}\sum_{i \in I} \wei_i = \indic_{\Omega}$, we have $\int_\Omega \ro_{\p_n} = N$. We estimate
\begin{align*}
\int_{B_R^{\tx{c}} \cap \Omega} \ro_{\p_n} \le \int_\Omega \ro_{\p_n} \mysum{\supp \wei_i \cap B_R^{\tx{c}} \neq \empt}{} \wei_i =  \mysum{\supp \wei_i \cap B_R^{\tx{c}} \neq \empt}{} r_i,
\end{align*}
and using the assumption \eqref{tigh} yields $\sup_n \int_{B_R^{\tx{c}}\cap \Omega} \ro_{\p_n} \ra 0$ when $R \ra +\ii$. This implies that $\sqrt{\ro_{\p_n}}$ converges strongly in $L^2(\Omega)$, up to extraction of a subsequence. The tightness of $\ro_{\p_n}$ also implies that $\p_n \ra \p_{\ii}$ strongly in $L^2(\Omega^N)$, that the limit of $\ro_{\p_n}$ is $\ro_{\p_{\ii}}$, and eventually that $\int_\Omega \ro_{\p_{\ii}} \wei_i = r_i$. By lower semi-continuity of the energy functional since $w \ge 0$, $\cE_0(\p_{\ii}) \le \fll{0}(r)$, hence $\p_{\ii}$ is a minimizer. By equivalence of the quadratic form $\cE_0$ with the one of $H^1(\Omega^N)$, $\p_n \ra \p_{\ii}$ strongly in $H^1(\Omega^N)$.

In the mixed state case, let us denote by $\Gamma_n$ a minimizing sequence. We use the compactness of the Fock space of particle number less than $N$, $\cans\pa{\cF^{\le N}}$ for the geometric convergence \cite[Lemma 2.2]{Lewin11}, we thus have $\Gamma_n \wra_g \Gamma_{\ii}$ for some $\Gamma_{\ii} \in \cans\pa{\cF^{\le N}}$. As before, the tightness of $\ro_{\Gamma_n}$ implies $\Gamma_n \ra \Gamma_{\ii}$ strongly in trace-class by \cite[Lemma 2.3]{Lewin11}, hence $\Gamma_{\ii}$ is an $N$-particle density matrix.  
 \end{proof}

One does not need the weight functions $\wei_i$ to have the diameters of their supports converging to zero to get that our regularized functionals converge to the exact one. Nevertheless, when this is the case, we can deduce bounds on the rate of convergence of the densities of minimizers to the target density. More precisely, the following result quantifies the distance between two densities satisfying $\int_\Omega \ro \wei_i = \int_\Omega \chi \wei_i$ for all $i \in I$. We consider exponents
\begin{align}\label{cdc}
q \in [1,2] \tx{ if } d =1,\hs \hs q \in [1,2) \tx{ if } d=2  \hs\hs\tx{ and } q = d/(d-1) \tx{ if } d \ge 3.
\end{align}

\begin{lemma}[Bounds on approximate densities]\label{lemama} Let $\Omega \subset \R^d$ be a connected open set with Lipschitz boundary, let $q$ be as in \eqref{cdc}. Take $\wei_i = \beta_i^{k_i}$ where $k_i \in \R_+$ and $\beta_i$ are non-negative concave functions on $\supp \wei_i$, with $\sup_{i \in \N} \diam \supp \wei_i < +\ii$, $\indic_{\Omega} \sum_{i=1}^{+\ii} \wei_i = \indic_{\Omega}$. Let $\ro,\chi \in L^1(\Omega,\R_+)$ such that $\sqrt{\ro}, \sqrt{ \chi} \in H^1_0(\Omega)$ and such that $\ro$ and $\chi$ are Lipschitz continuous. If $\int_\Omega  \ro \wei_i = \int_\Omega \chi \wei_i$ for any $i \in I$, then 
\begin{align*}
\nor{\ro-\chi}{\pa{L^1 \cap L^{q}}(\Omega)} \le c_{d,\Omega}  \pa{ \nor{\sqrt{\ro}}{H^1(\Omega)}^2 + \nor{\sqrt{\chi}}{H^1(\Omega)}^2}\sup_{i \in \N} \diam \supp \wei_{i}, 
\end{align*}
where $c_d$ only depends on $d$ and $\Omega$.
\end{lemma}
\begin{proof}[Proof of Lemma \ref{lemama}]
Take $s \ge 1$. We use the weighted Poincaré-Wirtinger inequality from \cite[Theorem 1.1]{ChuWhe06} with $c_{f,i} := \pa{\int_\Omega \wei_i}^{-1} \int_\Omega f \wei_i$. We obtain
\begin{align*}
\int_\Omega \ab{f - c_{f,i}}^s \wei_i \le c_s \pa{\diam \supp \wei_i}^s \int_\Omega \ab{\na f}^s \wei_i,
\end{align*}
for $f \in \acs{\ro, \chi}$, and it is to apply this inequality that we need the assumption on Lipschitz continuity. Thus, since $c_{\ro,i} = c_{\chi,i}$ by assumption,
\begin{align*}
\int_\Omega \ab{\ro - \chi}^s \wei_i & = \int_\Omega \ab{\ro -c_{\ro,i} - \pa{\chi - c_{\ro,i}}}^s \wei_i  \\
& \le c_s \pa{ \int_\Omega \ab{\ro - c_{\ro,i}}^s \wei_i + \int_\Omega \ab{\chi - c_{\ro,i}}^s \wei_i} \\
& \le c_s \pa{ \sup_{i \in \N} \diam \supp \wei_{i} }^s \int_\Omega \bpa{\ab{\na \ro}^s+\ab{\na \chi}^s} \wei_i.
\end{align*}
Summing over $i$ and raising to the power $1/s$ yields
\begin{align}\label{ic}
\nor{\ro - \chi}{L^s(\Omega)} \le c_s \pa{ \sup_{i \in \N} \diam \supp \wei_{i} } \bpa{ \nor{\na \ro}{L^s(\Omega)} + \nor{\na \chi}{L^s(\Omega)}}.
\end{align}
We now decompose $\na f = 2 \sqrt{f} \na \sqrt{f}$. For $d \ge 2$ we take $s = q$, and by the Sobolev injections \eqref{ineq:sob_inj}, we have
\begin{align*}
\nor{\na f}{L^q(\Omega)} \le 2 \nor{\na \sqrt{f}}{L^2(\Omega)}  \nor{\sqrt{f}}{L^{\f{2q}{2-q}}(\Omega)} \le c \nor{\sqrt{f}}{H^1(\Omega)}^2.
\end{align*}
For $d=1$ we take $s=2$ and use
\begin{align*}
\nor{\na f}{L^2(\Omega)} \le 2 \nor{f}{L^{\infty}(\Omega)}^{\f 12} \nor{\na \sqrt{f}}{L^2(\Omega)} \le c \nor{\sqrt{f}}{H^1(\Omega)}^2.
\end{align*}
Applying \eqref{ic} concludes the proof.
\end{proof}

We prove now the convergence of our regularized functionals to the exact ones.

\begin{proof}[Proof of Theorem \ref{proproi}]
Let us denote by $\p_n$ a sequence of approximate minimizers for $\flln{0}{n}(r_{\ro})$. Since $w \ge 0$, then 
\begin{align*}
\int_\Omega \ab{\na\sqrt{\ro_{\p_n}}}^2 \le \int_{\Omega^N} \ab{\na \p_n }^2 \le \cE_0(\p_n) \le \flln{0}{n}(r_{\ro}) + \ep_n \le F\ex{0}(\ro) + \ep_n,
\end{align*}
where $\ep_n \ra 0$. Hence $\pa{\p_n}_{n \in \N}$ is bounded in $H^1(\Omega^N)$ and there exists $\p_{\ii} \in \wedge^N H^1(\Omega)$ such that $\p_n \wra \p_{\ii}$ weakly in $H^1(\Omega^N)$. By summing all the constraints on the density and using that $\indic_{\Omega}\sum_{i \in I} \wei_i = \indic_{\Omega}$, we have $\int_\Omega \ro_{\p_n} = N$. We estimate

\begin{align*}
\int_{B_r^{\tx{c}} \cap \Omega} \ro_{\p_n} \le \int_\Omega \ro_{\p_n} \mysum{\supp \wei_i^n \cap B_r^{\tx{c}} \neq \empt}{} \wei_i^n = \int_\Omega \ro \mysum{\supp \wei_i^n \cap B_r^{\tx{c}} \neq \empt}{}\wei_i^n,
\end{align*}
and using the assumption \eqref{tig} yields $\sup_n \int_{B_r^{\tx{c}}\cap \Omega}\ro_{\p_n} \ra 0$ when $r \ra +\ii$. This implies that $\sqrt{\ro_{\p_n}}$ converges strongly in $L^2(\Omega)$ and weakly in $H^1(\Omega)$, up to extraction of a subsequence. The tightness of $\ro_{\p_n}$ also implies that $\p_n \ra \p_{\ii}$ strongly in $L^2(\Omega^N)$ and that the limit of $\ro_{\p_n}$ is $\ro_{\p_{\ii}}$. 

Let $f \in \cC^{\ii}\ind{c}(\Omega)$, by assumption \eqref{hypo}, there exists a sequence of functions $f_n \in \vect\acs{\wei_i^n, i \in I_n}$ such that $\nor{f-f_n}{(L^{p}+L^{\ii})(\Omega)} \ra 0$ when $n \ra +\ii$. We also have $\int_\Omega f_n \ro_{\p_n} = \int_\Omega f_n \ro$ because $f_n \in \vect\acs{\wei_i^n, i \in I_n}$. By using
 \begin{multline*}
	\ab{ \int_\Omega f (\ro_{\p_n} - \ro) }   =  \ab{ \int_\Omega (f-f_n)(\ro_{\p_n} - \ro) }\\
	  \le c_{d,\Omega} \pa{ \nor{\sqrt{\ro}}{H^1(\Omega)}^2 + \mysup{n \in \N} \nor{\sqrt{\ro_{\p_n}}}{H^1(\Omega)}^2} \nor{f - f_n}{(L^{p} + L^{\ii})(\Omega)},
 \end{multline*}
	we deduce that $\int_\Omega f \ro_{\p_n} \ra \int_\Omega f \ro$. This is a convergence of $\ro_{\p_n}$ to $\ro$ in the sense of distributions and by uniqueness of the limit, we have then $\ro = \ro_{\p_{\ii}}$.
	

	We deduce that $\p_{\ii}$ belongs to the minimizing set of $F\ex{0}(\ro)$, consequently $F\ex{0}(\ro)\le\cE_0(\p_{\ii})$. By also using lower semi-continuity of the energy functional since $w \ge 0$, we have 
 \begin{align*}
F\ex{0}(\ro)\le\cE_0(\p_{\ii}) \le \liminf \cE_0(\p_n) \le F\ex{0}(\ro).
 \end{align*}
	We have thus equality and we conclude that $\p_{\ii}$ is a minimizer of $F(\ro)$.  Let us consider the quadratic form $q(\p) := \ps{\p,\hn(0)\p}$. The convergence on the Levy-Lieb functionals $\flln{0}{n}(r_{\ro}) \ra F\ex{0}(\ro)$ gives $q(\p_n) \ra q(\p_{\ii})$, and since $w \ge 0$ is $(-\Delta)$-bounded as a quadratic form, the associated norm of $q$ is equivalent to the $H^1(\Omega)$ norm, and hence $\p_n \ra \p_{\ii}$ in $H^1(\Omega)$.

	In the mixed states case, we follow a similar adaptation as for proving Theorem \ref{exiss}. As in the pure states case, the norm of $\cE_0$ is equivalent to the norm of $\sch_{1,1}$, hence $\Gamma_n \ra \Gamma_{\ii}$ strongly in $\sch_{1,1}$.
 \end{proof}

\subsection{The dual problem: proof of Theorem \ref{cococo}}\label{ssub:proof_cococo} 

In this section, we prove Theorem \ref{cococo} on the coercivity of the dual functional $\gew{k}$. In the proofs we will use the notation
\begin{align*}
\edi{k}(v) := \exc{k}\pa{ \sum_{i \in I}  v_i \wei_i}, \bhs V(v) := \sum_{i \in I} v_i \wei_i.
\end{align*}

We recall that $c_{\Omega} := - \exc{k}(0)/N$ is the constant such that the energy $\exc{k}(c_{\Omega}) = 0$ vanishes. We present a fact about the sign of the potential. 

\begin{lemma}\label{lelele}
Let $v \in \ell^{\ii}(I,\R)$ be such that $\edi{0}(v) = 0$ and $v \neq \come$. If $\hn\bpa{V(v)}$ has a ground state, then there exists $i \in I$ such that $v_i < c_{\Omega}$. If $\Omega$ is bounded, there exist $i,j \in I$ such that $v_i < \come < v_j$.
\end{lemma}

\begin{proof}[Proof of Lemma \ref{lelele}]\tx{ }
	
\bul Let $\p$ be a ground state of $\hn(V(v))$. We have
\begin{align*}
0 & = \edi{0}(v) = \cE_{V(v)}(\p) = \cE_{c_{\Omega}}(\p) + \int_\Omega \pa{V(v) - c_{\Omega}} \ro_{\p} \\
& \ge \exc{0}\bpa{c_{\Omega}} + \int_\Omega \pa{V(v) - c_{\Omega}} \ro_{\p} = \sum_{i \in I} \pa{v_i - \come} \int_\Omega \ro_{\p} \wei_i.
\end{align*}
We used that $\sum_{i \in I} \ab{v_i} \alpha_i \le \sup_{i \in I} \ab{v_i}$ and dominated convergence to commute sum and integral. Thus
\begin{align}\label{comee}
\sum_{v_i > \come} \ab{v_i - \come} \int_\Omega \ro_{\p_v}\wei_i \le \sum_{v_i  < \come} \ab{v_i - \come} \int_\Omega \ro_{\p_v} \wei_i.
\end{align}
If $v \ge \come$, then the right hand side of \eqref{comee} vanishes and $v = \come$ because $\ab{\acs{\ro_{\p_v}  =0}}=0$ and therefore $\int_\Omega \ro_{\p_v} \wei_i > 0$ by unique continuation \cite[Remark 1.6]{Garrigue19}. Thus there is $i \in I$ such that $v_i < \come$.
	
\bul If $\Omega$ is bounded, then $\exc{0}(0)$ has a minimizer $\p$ and
\begin{align*}
0 = \edi{0}(v) \le \cE_{V(v)}(\p) 
 = \sum_{i \in I} \pa{v_i - \come} \int_\Omega \ro_{\p} \wei_i,
\end{align*}
then we obtain the inequality opposite to \eqref{comee},
\begin{align}\label{comeee}
\sum_{v_i  < \come} \ab{v_i - \come} \int_\Omega \ro_{\p}\wei_i \le \sum_{v_i > \come} \ab{v_i - \come} \int_\Omega \ro_{\p} \wei_i,
\end{align}
for this particular state. If $v \le \come$, then the right hand side of \eqref{comeee} vanishes, but this is not possible since the left hand side has to be strictly positive, hence there is $\ell \in I$ such that $v_{\ell} > \come$. 

The map $u \mapsto \edi{0}(u)$ is strictly increasing by \cite[Corollary 1.5]{Garrigue21}. By taking a potential $v \in \ell^{\ii}(\Omega,\R)$ such that $\edi{0}(v) = 0$ and $v_{\ell}  > \come$ for some $\ell \in I$, if we suppose that $v \ge \come$, then $0 = \edi{0}(\come) < \edi{0}(v)$, which is a contradiction. We conclude that there is also $j \in I$ such that $v_j < \come$.
\end{proof}

We are now ready to prove the coercivity inequality of the regularized dual functional.

\begin{proof}[Proof of Theorem \ref{cococo}]\tx{ }

\bul We first prove \eqref{cob}. We assumed that there are points $y_i \in \R^d$ such that for any $i \in I$,
\begin{align*}
B_R(y_i) \subset \pa{\supp \wei_i} \backslash \cup_{j \in I, j \neq i} \supp \wei_j.
\end{align*}
We write $X = (x_1,\dots,x_N)$ and $Y_i := (y_i,\dots,y_i)$. Take normalized $\Phi_0,\dots,\Phi_{k} \in \wedge^N H^1_0(B_R)$ with disjoint supports. Take some non-empty $Q \subset I$ and for $j \in \acs{0,\dots,k}$, form
\begin{align*}
\p_{j,Q}(X) := \inv{\sqrt{\sum_{i \in Q} r_i}} \sum_{i \in Q} \sqrt{r_i} \hspace{0.2cm} \Phi_j(X-Y_i).
\end{align*}
This satisfies $\int_{\Omega^{N}} \ab{\p_{j,Q}}^2 = 1$, $T(\p_{j,Q}) = T(\Phi_j)$, $W(\p_{j,Q}) = W(\Phi_j)$ 
and
\begin{align*}
\ro_{\p_{j,Q}}(x) = \pa{\sum_{i \in Q} r_i }^{-1} \sum_{i \in Q} r_i \ro_{\Phi_j}(x-y_i).
\end{align*}
We use the expression
\begin{align*}
	\exc{k}(V) = \myinf{\dim A = k+1} \mymax{\p \in A \\ \int_{\Omega^N} \ab{\p}^2 = 1} \cE_V(\p)
 \end{align*}
and choose the frame $A := \bpa{\p_{0,Q},\dots,\p_{k,Q}}$ so that
 \begin{align*}
	 \gew{k}(v) \le - \sum_{i \in I} v_i r_i + \mymax{\lambda_j \in \C \\ \sum_{j=0}^k \ab{\lambda_j}^2 =1} \cE_{V(v)}\pa{\sum_{j=0}^k \lambda_j \p_{j,Q}}.
 \end{align*}
	For any $i \in I$, the only non-vanishing element of $\weig$ in $B_R(y_i)$ is $\wei_i$, so $\wei_i = 1$ on $B_R(y_i)$ and 
 \begin{align*}
	 \int_{\Omega} \wei_i \ro_{\p_{j,Q}} & = \f{N r_i \delta_{i \in Q}}{\sum_{\ell \in Q} r_{\ell}}, \\
	 \int_{\Omega} V(v) \ro_{\sum_{j=0}^k \lambda_j \p_{j,Q}} & = \sum_{j =0}^{k} \ab{\lambda_j}^2 \int_{\Omega} V(v) \ro_{\p_{j,Q}} = \f{N \sum_{i \in Q} v_i r_i}{\sum_{\ell \in Q} r_{\ell}}.
 \end{align*}
We see that the external potential energy of the trial state does not depend on the $\lambda_j$'s. Defining
 \begin{align*}
	  c_R := \mymax{\lambda_j \in \C \\ \sum_{j=0}^k \ab{\lambda_j}^2 =1} \cE_{0}\pa{\sum_{j=0}^k \lambda_j \p_{j,Q}},
 \end{align*}
 we deduce that
	\begin{align}\label{ipo}
	\gew{k}\bpa{v} & \le c_R + \f{N}{\sum_{i \in Q} r_i}\sum_{i \in Q} v_i r_i - \sum_{i \in I} v_i r_i \\
& = c_R + \f{\sum_{i \in I \backslash Q} r_i }{\sum_{i \in Q} r_i}\sum_{i \in Q} v_i r_i - \sum_{i \in I \backslash Q} v_i r_i. \nonumber
\end{align}
Since $G$ is gauge invariant, for any $\mu \in \R$ and any non-empty $Q \subset I$, we have
\begin{align}\label{lalalala}
	\gew{k}(v) & = \gew{k}(v-\mu) \le c_R + \f{\sum_{I \backslash Q} r_i }{\sum_{i \in Q} r_i}\sum_{i \in Q} (v_i - \mu) r_i - \sum_{i \in I \backslash Q} (v_i-\mu) r_i.
 \end{align}
	We define the two sets $I^{\pm}_v := \acs{ i \in I \st \pm v_i > \pm c_{\Omega}}$. In the case $I^-_v \neq \empt$, we take $Q = I^-_v$ and $\mu = c_{\Omega}$ yielding
	\begin{align}\label{indua}
		\gew{k}(v)  - c_R  & \le  \f{\sum_{v_i \ge c_{\Omega}} r_i }{\sum_{v_i < c_{\Omega}} r_i}\sum_{v_i < c_{\Omega}} (v_i - c_{\Omega}) r_i - \sum_{v_i \ge c_{\Omega}} (v_i-c_{\Omega}) r_i \nonumber \\
		& \le \min \pa{ 1, \f{\sum_{v_i \ge c_{\Omega}}r_i}{\sum_{v_i < c_{\Omega}} r_i}} \hspace{-0.1cm} \pa{\sum_{v_i < c_{\Omega}} (v_i - c_{\Omega}) r_i - \sum_{v_i \ge c_{\Omega}} (v_i-c_{\Omega}) r_i }\nonumber  \\
	  & \le   -\f{\sum_{v_i \ge c_{\Omega}}r_i}{N} \nor{v - c_{\Omega}}{\ell^1_{r}}.
 \end{align}
In the case $v = c_{\Omega}$, we have $\gew{k}(v) = - c_{\Omega} N = \exc{k}(0) \le c_R$ by using the same trial state as before, hence the bound also holds. 

	If $v \ge \come$, then $\inf v = \come$, because otherwise $\inf v > \come$ and
 \begin{align*}
	 0 = E\ex{k}\ind{dis}(v) \ge E\ex{k}\ind{dis}(\inf v) > E\ex{k}\ind{dis}\pa{\come} = 0.
 \end{align*}
	We take a sequence $(v_{\vp(n)})_{n \in \N}$ where $\vp(n) \in I$ is such that $v_{\vp(n)} \ra \come$ when $n \ra +\ii$. If $I$ is finite, then there is $\ell \in I$ such that $v_{\ell} = 0$ and we take $\vp(n) = \ell$ for any $n \in \N$. We choose $Q = \acs{\vp(n)}$ with only one element, and \eqref{ipo} yields
 \begin{align*}
	 \gew{k}(v) & \le c_R + Nv_{\vp(n)}  - \sum_{i \in I} v_i r_i.
 \end{align*}
 By gauge invariance, we also have
 \begin{align*}
	 \gew{k}(v) =  \gew{k}(v- \come ) \le - \nor{v- \come}{\ell^1_r} + N \pa{v_{\vp(n)}-\come} + c_R,
 \end{align*}
where we used $v \ge \come$. Taking the limit $n \ra +\ii$ yields
 \begin{align*}
	 \gew{k}(v)  \le - \nor{v- \come}{\ell^1_r} + c_R.
 \end{align*}


In case $k=0$ and $I$ is infinite, we can still prove that there is a maximizer. Let $v^n$ be a maximizing sequence. By coercivity, $\sum_{i \in I} \ab{v_i^n}r_i$ is bounded hence $\ab{v_i^n} \le c/r_i$ uniformly in $i,n$. 
	There exists $(v^{\ii}_i)_{i \in I} \in \ell^{\ii}(I,\R)$ such that $v_i^n \ra v_i^{\ii}$ for all $i \in I$, up to a subsequence and finally $v^{\ii} \in \ell^1_{r}(I,\R)$ by Fatou's lemma. We conclude by using weak upper semi-continuity of $\gew{0}$. For $k \ge 1$, $\gew{k}$ is not upper semi-continuous or concave but when $I$ is finite, since it is coercive and lives in a finite-dimensional space, it has a maximum.

\bul Now assume that $\Omega$ is bounded, so that every potential is binding, moreover $I$ is necessarily finite. We first give the beginning of a proof which would not use Theorem \ref{kspropi} to see why it only works for $k=0$.

We would first need that $\fl{k}(r)$ has an optimizer $(\Gamma_r,A_r)$, this was proved by Lieb \cite{Lieb83b} for $k=0$ and we are not able to prove it for $k \ge 1$. Since $\int_\Omega \ro_{\Gamma_r}\wei_i = r_i$, then
\begin{align*}
	\cE_{V(v^{\ii})}(\Gamma_r) &  = \tr \hn(0) \Gamma_r + \sum_{i \in I} v_i^{\ii} r_i = \gew{k}(v^{\ii}) + \sum_{i \in I} v_i^{\ii} r_i = \edi{k}(v^{\ii}).
\end{align*}
We diagonalize $\Gamma_r =: \sum_{\ell \in \N} \lambda_{\ell} \proj{\vp_{\ell}}$, where $\sum_{\ell \in \N} \lambda_{\ell} = 1$ and $\vp_{\ell} \in A_r^{\perp}$. By linearity of the energy functional, we have 
\begin{align*}
	\sum_{\ell \in \N} \lambda_{\ell} \cE_{V(v^{\ii})}(\vp_{\ell}) = \edi{k}(v^{\ii}).
\end{align*}
We need again $k=0$, because then $\cE_{V(v^{\ii})}(\vp_{\ell}) \ge \edi{0}(v^{\ii})$ and thus we have $\cE_{V(v^{\ii})}(\vp_{\ell}) = \edi{0}(v^{\ii})$ for any $\ell \in \N$, and finally
\begin{align*}
	\vp_{\ell} \in \Ker \pa{ \hn\pa{V(v^{\ii})} - \edi{0}(v^{\ii})}.
\end{align*}

Nevertheless, by an adaptation of Theorem \ref{kspropi} to the discretized case, we know that since $v$ maximizes $\gew{k}$, then it has a $k\expo{th}$ bound mixed state $\Gamma$ with density satisfying $\int_\Omega \ro_{\Gamma} \wei_i = r_i$.
\end{proof}

\begin{remark}
	When $\Omega = \R^d$, the situation is different. Assume also $\edi{k}(v) = 0$. We apply \eqref{lalalala} with $Q = \acs{ v_i < \ep}$ for some $\ep > 0$, which is not empty since $\inf v = 0$. This yields
 \begin{align*}
	 \gew{k}(v) - c_R  & \le  - \min\pa{1,\f{\sum_{v_i \ge \ep} r_i}{\sum_{v_i < \ep} r_i}} \nor{v-\ep}{\ell^1_{r,\weig}}
 \end{align*}
and we conclude by letting $\ep \ra 0$,
 \begin{align*}
	 \gew{k}(v)- c_R  & \le -\min\pa{1,\f{\sum_{v_i > 0}r_i}{\sum_{v_i \le 0}r_i}} \nor{v }{\ell^1_{r}} \le -\min\pa{1,\f{\sum_{v_i > 0}r_i}{N}} \nor{v }{\ell^1_{r}}.
 \end{align*}
 The problem is that we are not able to find a strictly positive lower bound for $\sum_{v_i > 0}r_i$, which would provide coercivity.
\end{remark}

\begin{remark}
A natural norm on potentials is the gauge invariant quotient norm 
\begin{align*}
\nor{v}{\ell^{1}_{r} / \sim} = \inf_{\mu \in \R} \nor{v-\mu}{\ell^1_{r}},
\end{align*}
where $v \sim u$ when $v-u$ is constant. If in \eqref{lalalala} we take $Q =  \acs{i \in I \st v_i = \inf v}$, $\mu = \inf v$ we obtain
\begin{align*}
	\gew{k}(v) - c_R  \le  - \sum_{v_i > \inf v} (v_i - \inf v) r_i =  - \nor{v - \inf v}{\ell^1_{r}}\le - \nor{v}{\ell^1_{r}/\sim}.
\end{align*}
	Hence $\gew{k}$ is coercive in the $\ell^1_{r}/\sim$ norm, but this is not a convenient norm because by definition we do not control the constant.
\end{remark}

\subsection{Building Kohn-Sham potentials: proofs of Corollary~\ref{canot} and Theorem \ref{juju}}\label{ssub:proof_canot} 

In this section, we show how Theorem~\ref{cococo} yields approximate $v$-representability when $\Omega$ is unbounded.

\subsubsection{The mixed states case}

\begin{proof}[Proof of Corollary \ref{canot}]
	Consider copies of the cube $C^n := [-1/n,1/n)^d$, centered on the grid points of $\pa{\Z/n}^d$. Take $\Omega_n$ to be the union of all those cubes $(C^n_i)_{i\in I_n}$ which are included in $\Omega \cap B_n$, where $B_r$ is the ball of radius $n$. They form an increasing sequence $\Omega_n \subset \Omega_{n+1}$. We choose $\wei_i^n := \indic_{C^n_i}$. We apply Theorem~\ref{cococo} to $r\expa{n}  := c_n r_{\ro\indic_{\Omega_n}}$ with $c_n := N/ \int_{\Omega_n} \ro$. The condition $\ab{ \acs{ \ro = 0}} = 0$ ensures that $r_i\expa{n} > 0$ for any $i \in I_n$, $n \in \N$. We denote by $v^n \in \ell^{\ii}(I,\R)$ the maximizer of $G\ex{k}_{r\expa{n},\weig_n}$, we have $\int_\Omega \ro_{\Gamma_{n}} \wei_i^n = c_n \int_\Omega \ro \wei_i^n$ and we apply Theorem~\ref{proproi} for each $n$.

	Each $\Gamma_{n}$ lives in $\cans(\Omega_n)$ and each $\hn(\sum v^n_i \wei_i^n)$ is an operator of $L^2\ind{a}(\Omega_n^N)$, but by \apo{digging pits} close to where the density is localized, we can create a potential $V_n = -\lambda_n \indic_{\Omega_n} + \sum_i v^n_i \wei_i^n$ with $\lambda_n$ large enough so that we keep the same properties for systems living in an unbounded domain $\Omega$.

For $k=0$, we take $\Gamma_{n}$ a minimizer of $\flln{0}{n}(r\expa{n})$ and apply Theorem~\ref{proproi}.

\end{proof}

\subsubsection{The pure states case}


\begin{proof}[Proof of Theorem \ref{juju}]
 We assume $N \ge 1$ and only restrict to $N=1$ at the end of the argument. We define the map
\begin{align*}
\rodii : \begin{array}{rcl}
	H^1\ind{a}(\Omega^N) \cap \acs{ \int_{\Omega^N} \ab{\p}^2 = 1} & \longrightarrow & \pa{\R_+}^{\ab{I}} \\
	\p & \longmapsto & \pa{\int_\Omega \ro_{\p} \wei_i }_{i \in I}. \\
\end{array}
\end{align*}
It is $\cC^{\ii}$ and since $I$ is finite, $\ran \d_{\p} \rodii$ is closed for any $\p \in L^2\ind{a}(\Omega^N)$ because the image lives in a finite-dimensional space. Now $\d_{\p} G$ has closed range for any $\p \in \wedge^N H^1(\Omega)$ because its target space is finite-dimensional. We compute, for any $\vp \in H^1\ind{a}(\Omega^N)$ 
	and any $v \in \R^{\ab{I}}$,
\begin{align*}
	& \sum_{i \in I} v_i \pa{ \pa{\d_{\p} \rodii} \vp}_i = 2 N  \sum_{i \in I} v_i \int_{\Omega^N}  (\p {\vp})(x_1,\dots,x_N) \wei_i(x_1) \d x_1 \cdots \d x_N \\
	& \bhs = 2  \int_{\Omega^N} \sum_{j=1 }^N \pa{\sum_{i \in I} v_i \wei_i}(x_j) (\p {\vp})(x_1,\dots,x_N) \d x_1 \cdots \d x_N \\
	& \bhs = 2   \ps{ \vp, \vv{V(v)} \p}
\end{align*}
	Furthermore, $h : \p \mapsto \ps{\p, \hn(0)\p}$ is $\cC^{\ii}$ with differential 
 \begin{align*}
\pa{\d_{\p} h} \vp = 2  \ps{\vp, \hn(0)\p}.
 \end{align*}
	Let $\p \in H\ind{a}^1(\Omega^N)$ with unit norm be a minimizer of $\fll{0}(r)$, then $\int_\Omega \ro_{\p} \wei_i = r_i$ for any $i \in I$. We apply \cite[Prop 43.19 p291]{Zeidler3}, ensuring the existence of Lagrange multipliers $(v_i)_{i \in I} \in \pa{ \ell^1(I,\R) }^* = \ell^{\ii}(I,\R)$ such that $\hn(V(v)) \p = 0$ weakly, or such that $\vv{V(v)} \p = 0$ weakly and $v \neq 0$. 

	We prove by contradiction that the second case is impossible, so let us assume that $\vv{V(v)} \p = 0$. This is where we need $N=1$, which implies $V(v) \p = 0$ and $0= v_i \int_\Omega \ro_{\p} \wei_i =v_i r_i$, but since $r_i > 0$, we conclude that $v_i = 0$, which is a contradiction. If Conjecture \ref{conjos} holds, then $\vv{V(v)} \p = 0$ implies $\vv{V(v)} = 0$ a.e and $v = 0$, and this avoids the second case for all $N$.

We hence know that 
\begin{align}\label{truc}
\hn(V(v)) \p = 0
 \end{align}
	with $V(v) \in L^{p}(\Omega)$, $p$ as in \eqref{dims}, so $\p$ is in the domain of $\hn(V(v))$ and consequently is an eigenvalue. Since $\Omega$ is bounded, $\p$ is in the discrete spectrum.
\end{proof}

We remark that Conjecture \ref{conjos} also implies that for any $r \in \ell^1(I,\R_+)$, $\rodii^{-1}(r)$ is a manifold, because $\d_{\p} \rodii$ is then surjective and one can apply the preimage theorem \cite[Theorem 73.C p556]{Zeidler4}. Since $\d_{\p} \rodii$ is closed, this mathematical framework can be applied to define the so-called adiabatic connection of DFT, because one can apply the implicit functions theorem and let $w$ decrease while the inverse potential $v$ increases, keeping the density fixed.

\subsection{Proofs of Section \ref{secopt}: Lemma \ref{localp} and Proposition \ref{ksprops}}\label{ssub:proof_secopt}

 \begin{proof}[Proof of Lemma \ref{localp}]\tx{ }

	 $i)$ We define $D_N := \dim \Ker_{\R} \pa{\hn(v)-\exc{k}(v)}$. We choose a basis $(\p_I)_{I \in \cI\ind{tot}}$ of $\Ker_{\R} \pa{\hn(v)-\exc{k}(v)}$ composed of antisymmetric products of one-body real orbitals. We compute
 \begin{align*}
	 & \int_{\Omega^{N-1}} \pa{\p_I\p_J}(x,x_2,\dots,x_N) \d x_2 \cdots \d x_N  \\
	 & \hs\hs\hs\hs\hs\bhs  = \delta_{I, J} N^{-1} \pa{ \sum_{i \in I} \vp_i(x)^2}  + \delta_{I \cupdot J = \acs{i,j}} N^{-1}  \vp_i(x)\vp_j(x). 
 \end{align*}

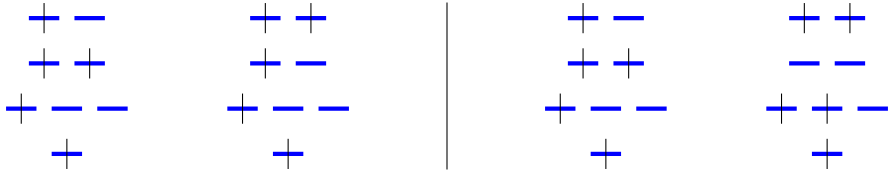
\begin{figure}
\begin{tikzpicture}[scale=0.4]
\draw[level]   (-0.5,0)  -- (0.5,0);
\draw[level]   (-0.5,1.5)  -- (0.5,1.5);
\draw[level]   (-2,1.5)  -- (-1,1.5);
\draw[level]   (1,1.5)  -- (2,1.5);
\draw[level]   (-1.25,3)  -- (-0.25,3);
\draw[level]   (1.25,3)  -- (0.25,3);
\draw[level]   (-1.25,4.5)  -- (-0.25,4.5);
\draw[level]   (1.25,4.5)  -- (0.25,4.5);
	\draw (0,0.5) -- (0,-0.5) (-1.5,1) -- (-1.5,2) (-0.75,2.5) -- (-0.75,3.5)(0.75,2.5) -- (0.75,3.5) (-0.75,4)--(-0.75,5);
\end{tikzpicture}
	\hspace{1cm} 
\begin{tikzpicture}[scale=0.4]
\draw[level]   (-0.5,0)  -- (0.5,0);
\draw[level]   (-0.5,1.5)  -- (0.5,1.5);
\draw[level]   (-2,1.5)  -- (-1,1.5);
\draw[level]   (1,1.5)  -- (2,1.5);
\draw[level]   (-1.25,3)  -- (-0.25,3);
\draw[level]   (1.25,3)  -- (0.25,3);
\draw[level]   (-1.25,4.5)  -- (-0.25,4.5);
\draw[level]   (1.25,4.5)  -- (0.25,4.5);
	\draw (0,0.5) -- (0,-0.5) (-1.5,1) -- (-1.5,2) (-0.75,2.5) -- (-0.75,3.5) (0.75,4) -- (0.75,5) (-0.75,4)--(-0.75,5);
\end{tikzpicture}
	\hspace{1cm} 
\unskip\ \vrule\
	\hspace{1cm} 
\begin{tikzpicture}[scale=0.4]
\draw[level]   (-0.5,0)  -- (0.5,0);
\draw[level]   (-0.5,1.5)  -- (0.5,1.5);
\draw[level]   (-2,1.5)  -- (-1,1.5);
\draw[level]   (1,1.5)  -- (2,1.5);
\draw[level]   (-1.25,3)  -- (-0.25,3);
\draw[level]   (1.25,3)  -- (0.25,3);
\draw[level]   (-1.25,4.5)  -- (-0.25,4.5);
\draw[level]   (1.25,4.5)  -- (0.25,4.5);
	\draw (0,0.5) -- (0,-0.5) (-1.5,1) -- (-1.5,2) (-0.75,2.5) -- (-0.75,3.5)(0.75,2.5) -- (0.75,3.5) (-0.75,4)--(-0.75,5);
\end{tikzpicture}
	\hspace{1cm} 
\begin{tikzpicture}[scale=0.4]
\draw[level]   (-0.5,0)  -- (0.5,0);
\draw[level]   (-0.5,1.5)  -- (0.5,1.5);
\draw[level]   (-2,1.5)  -- (-1,1.5);
\draw[level]   (1,1.5)  -- (2,1.5);
\draw[level]   (-1.25,3)  -- (-0.25,3);
\draw[level]   (1.25,3)  -- (0.25,3);
\draw[level]   (-1.25,4.5)  -- (-0.25,4.5);
\draw[level]   (1.25,4.5)  -- (0.25,4.5);
	\draw (0,0.5) -- (0,-0.5) (-1.5,1) -- (-1.5,2) (0,1) -- (0,2)(0.75,4) -- (0.75,5) (-0.75,4)--(-0.75,5);
\end{tikzpicture}
	\caption{Vanishing crossing terms between two configurations.}\label{figuu}
\end{figure}

The second term does not vanish if there is exactly one of the degenerate one-body energy levels of $I$ and $J$ differing by exactly one particle, the other particles of the outer levels should have the same distribution. Indeed, as illustrated in Figure \ref{figuu}, if the difference belongs to two different levels (left), the energies are different so those terms actually do not appear, and if the energies are equal but differences belong to two different levels, then there are more than one difference. We deduce that
 \begin{align}\label{croc}
 & \int_{\Omega^{N-1}} \pa{\p_I\p_J}(x,x_2,\dots,x_N) \d x_2 \cdots \d x_N   \\
	 & \bhs = \delta_{I, J} N^{-1} \pa{ \sum_{i \in I} \vp_i(x)^2}  +  N^{-1} \delta_{\substack{\exists \ell,i,j \\ I \cupdot J = I\expa{\ell} \cupdot J\expa{\ell} = \acs{i,j}\\I\expa{t} = J\expa{t} \forall t \neq \ell}}\vp_i(x)\vp_j(x).\nonumber 
 \end{align}
 Finally we can rewrite
 \begin{align*}
 {^+}\delta_v \exc{k} (u)= \mu_{k-m_k^v}\pa{P \vv{u} P},
 \end{align*}
where $P$ is the projector onto $\Ker_{\R} \pa{\hn(v)-\exc{k}(v)}$. Using \eqref{croc}, the matrix elements of $P \vv{u} P$ are
 \begin{align*}
 \ps{\p_I, \vv{u} \p_{J}}  = \int_\Omega \pa{\delta_{IJ} \ro\ind{in}  + \cM_{\vp,IJ}}u.
 \end{align*}
Defining the elementwise action of the integral on matrices $\pa{\int_\Omega u \matix}_{IJ} := \int_\Omega u \matixx{IJ}$, we have
 \begin{align*}
	 ^+\delta_v G_{\ro} (u) =  \int_\Omega \pa{\ro\ind{in} - \ro} u + \mu_{k-m_k^v}\pa{ \int_\Omega u \matix }.
 \end{align*}

 $ii)$ As we see in \eqref{croc} and Figure \ref{figuu}, $\int_{\Omega^{N-1}} \p_I \p_J = 0$ when $I$ and $J$ have more than one particle difference.

 $iii)$ In this case, $\matix$ is diagonal by $ii)$.
 \end{proof}

\begin{proof}[Proof of Proposition \ref{ksprops}]\tx{ }

$i)$ The optimality condition $\sup_{u \in (\lpi)(\Omega,\R)} {^+}\delta_v \ger{k} (u) \le 0$ is equivalent to
\begin{align}\label{cela}
	\mu_{k-m_k^v}\pa{ \int_{\Omega} u \matix } \le \int_{\Omega} \pa{\ro-\ro\ind{in}}u  \le \mu_{M_k^v-k}\pa{ \int_{\Omega} u \matix },
 \end{align}
	for all $u \in (L^p+L^{\ii})(\Omega,\R)$. These are the Euler-Lagrange inequalities, replacing equalities because of the degeneracies. Applying this condition \eqref{cela} to the sequences $u_n := n^d \indic_{B_{1/n}(x)}$ and $-u_n$, and taking $n \ra +\ii$, we obtain \eqref{codl}. However, \eqref{codl} does not imply \eqref{cela}. More generally, for a function-valued real symmetric matrix $S$, and for $k,K\in \N$, $\mu_k \pa{S(x)} \le 0 \le \mu_K \pa{S(x)}$ a.e. is local and it does not imply $\mu_k \pa{\int_\Omega u S} \le 0 \le \mu_K \pa{\int_\Omega u S}$ for all $u$, which is global. A counterexample is the $D \times D$ diagonal matrix $S_{ii} = -1$ if $x \in [(i-1)/D[$ and $S_{ii} = 1$ otherwise, indeed with $u = 1$ we obtain $\min \sigma \pa{\int_\Omega u S} = (D-2)/D > 0$.

	$ii)$ We define $N\ind{in} := \int_\Omega \ro\ind{in}$ the number of \apo{inner} particles, and $N\ind{out} := N - N\ind{in}$ the number of outer particles. If $N\ind{out} =1$, we are in the situation where the only degeneracy comes from a one-body degeneracy at the one-body eigenspace $\vect (\phi_i)_{1 \le i \le D_N}$. At almost every $x \in \R^d$, $\matix = \mat{\phi_1 \dots \phi_{D_N}}^{\tx{T}} \mat{\phi_1 \dots \phi_{D_N}}$, all the eigenvalue are $0$ except one, which is $\sum_{i=1}^{D_N} \phi_i^2$. Indeed, the columns are proportional hence the rank is one, and the last eigenvalue is equal to the trace. Applying \eqref{codl} and integrating yields
 \begin{align*}
 N\ind{out} \le \int_\Omega \mu_{M_k^v-k}\pa{\matix} = 
\left\{
\begin{array}{ll}
	\int_\Omega \sum_{i=1}^{D_N} \phi_i^2 = D_N & \tx{if } k = m_k^v  \\
	0 & \tx{otherwise.}
\end{array}
\right.
 \end{align*}
 We deduce that $m_k^v = k$.

	$iii)$ We apply Lemma \ref{localp} $iii)$.
\end{proof}

\section*{Appendix 1: two remarks \\on density functionals for excited states}
In this appendix, we make two remarks on excited states functionals. Among other works, the problem was tackled by Lieb in \cite{Lieb85}. 
We take $\Omega = \R^d$ for simplicity. First, we notice that the inner problem in the Levy-Lieb functional \eqref{lly} is finite and has an optimizer.

 \begin{lemma}\label{inters}
	 Take $k \in \N$ with $k \ge 1$, $A \subset L^2\ind{a}(\R^{dN})$, $\dim_{\C} A = k$ and $\ro \in L^1(\R^d,\R_+)$ such that $\int_{\R^d} \ro = N$ and $\sqrt{\ro} \in H^1(\R^d)$. There exists $\p \in H^1\ind{a}(\R^{dN},\C)$ such that $\p \perp A$ and $\ro_{\p} = \ro$, and the infimum 
 \begin{align*}
	 \myinf{\p \in A^{\perp} \\ \ro_{\p} = \ro} \ps{\p,\hn(0)\p}
 \end{align*}
is finite and attained.
  \end{lemma}
\begin{proof}
We take a frame $\vect \pa{\phi_1,\dots,\phi_k} = A$, and then we consider the Harriman-Lieb \cite{Harriman81,Lieb83b} orbitals $\vp_1,\dots,\vp_N \in H^1(\R^d,\C)$, which are such that $\ro_{\wedge_{j=1}^N \vp_j} = \ro$. We then use the orthonormalization procedure \cite[Corollary 1.3]{LazLie13} with the functions $\phi_1,\dots\phi_k$, $\vp_1,\dots,\vp_N$ so that there exists functions $g_1,\dots,g_k$, $f_1,\dots,f_N \in \cC^{\ii}\ind{c}(\R,\R)$ such that $\phi_1 e^{ig_1(x_1)}$,$\dots\phi_k e^{ig_k(x_1)}$, $\vp_1 e^{if_1(x_1)}$, $\dots$, $\vp_N e^{if_N(x_1)}$ is an orthonormal familly. Eventually, by defining $\p := \wedge_{j=1}^N \vp_j e^{i f_j(x_1)} \in H^1(\R^{dN},\C)$, we have $\ro_{\p} = \ro$ and $\p \in A$. We conclude that the set $\acs{\p \in H^1(\R^{dN},\C) \cap A^{\perp} \st \ro_{\p} = \ro}$ is not empty.

The minimum is attained by an adaptation of the proof of the case $k=0$ which is \cite[Theorem 3.3]{Lieb83b}.
 \end{proof}

 Then, we present a remark about the other possibility of defining the Levy-Lieb functional for excited states, which is
 \begin{align*}
	 \widetilde{F}\ex{k}(\ro) :=  \myinf{A\subset H^1\ind{a}(\R^{dN}) \\ \dim_{\C} A = k+1 \\ \exists \p \in A, \ro_{\p} = \ro}  \mymax{\p \in A \\ \ro_{\p} = \ro} \ps{\p,\hn(0)\p}.
 \end{align*}
The advantage of this one is that we can directly prove that $\widetilde{F}\ex{k}(\ro)$ is finite. However, we now show why it seems to be not of much use. The first Levy-Lieb functional provides an upper bound to the energy
 \begin{align*}
\exc{k}(v) \le  \myinf{\ro \ge 0, \int_{\R^d} \ro = N\\ \sqrt{\ro} \in H^1(\R^d)}  \pa{ \widetilde{F}\ex{k}(\ro) + \int_{\R^d} v\ro}.
 \end{align*}
 However, there are potentials $v \in L^p+L^{\ii}$ such that
 \begin{align*}
\exc{k}(v)  < \mysup{\ro \ge 0, \int_{\R^d} \ro = N\\ \sqrt{\ro} \in H^1(\R^d)}  \pa{ \widetilde{F}\ex{k}(\ro) + \int_{\R^d} v\ro}
 \end{align*}
 and hence $\widetilde{F}\ex{k}(\ro)$ does not provide a lower bound to the energy. Indeed, otherwise we would have 
\begin{align*}
 \widetilde{F}\ex{k}(\ro) \le E\ex{k}(v) - \int_{\R^d} v \ro
\end{align*}
 for all $v \in (L^p+L^{\ii})(\R^d,\R)$ and all $\ro \ge 0$ such that $\int_{\R^d}  \ro =N$, $\sqrt{\ro} \in H^1(\R^d)$, hence 
\begin{align*}
 \widetilde{F}\ex{k}(\ro) \le \inf_{v \in (\lpi)(\R^d,\R)} \pa{ \exc{k}(v) - \int_{\R^d} v \ro}.
\end{align*}
We assumed that $\int_{B_r} \ro > 0$ without loss of generality. With $v_n = n\indic_{B_r}$ for instance, then $\exc{k}(v_n) - \int_{\R^d} v_n \ro \sim - n \int_{B_r} \ro \ra -\ii$ when $n \ra +\ii$, and this would imply that $\widetilde{F}\ex{k}(\ro) = -\ii$ for all $\ro$.

\section*{Appendix 2: maximizing ${^+}\delta_v \ger{k}(u)$ in the two-fold\\degenerate case}
Here we show that in the case $\dim \kerr =2=p=q$, we can reduce the $3$-dimensional optimization problem \eqref{lao} of maximizing ${^+}\delta_v \ger{k}$ to a 1-dimensional problem. The interaction $w$ is general. This would enable to further accelerate the ODA algorithm.

Take $\p$ and $\Phi$ real such that they form a real orthonormal basis of the degenerate level $\kerrr$, we have $M_k^v = m_k^v +1$ and let us define $\ep := 1$ if $k = m_k^v$ and $\ep := -1$ if $k= M_k^v$. Then
 \begin{align}
	 {^+}\delta_v \ger{m_k^v} (u)  = \mymin{a, b \in \C \\ \ab{a}^2+\ab{b}^2 = 1 } \int_\Omega u \pa{\ro_{a \p + b \Phi} -\ro}
 \end{align}
 and ${^+}\delta_v \ger{M_k^v} (u)$ has the same formula but with a maximization. With the parametrization $a = \pa{\cos \alpha} e^{i \eta}$, $b = \pa{ \sin \alpha} e^{i(\eta + \beta)}$, we have
 \begin{align*}
	 \ro_{(\cos \alpha) \p + e^{i \beta} (\sin \alpha) \Phi} & = (\cos \alpha)^2 \ro_{\p} + (\sin \alpha)^2 \ro_{\Phi} + 2 \sin \alpha \cos \alpha \cos \beta \ro_{\p,\Phi} \\
	 & = \f{\ro_{\p}+\ro_{\Phi}}{2} +  \f{\ro_{\p}  -\ro_{\Phi}}{2} \cos 2 \alpha + \sin 2 \alpha  \cos \beta \ro_{\p,\Phi}
 \end{align*}
 and
 \begin{align*}
	 \int_\Omega u \ro_{(\cos \alpha) \p + e^{i \beta} (\sin \alpha) \Phi} = \int_\Omega u \f{\ro_{\p}+\ro_{\Phi}}{2} +  A \cos 2 \alpha + B \sin 2 \alpha  \cos \beta
 \end{align*}
where 
 \begin{align*}
	 A_u := \ud \int_\Omega u \pa{\ro_{\p}-\ro_{\Phi}}, \bhs B_u := \ps{\p,\vv{u}\Phi} = \int_\Omega u \ro_{\p,\Phi}.
 \end{align*}
 This yields
 \begin{multline*}
 ^+\delta_v \ger{m_k^v}(u)  =   \int_\Omega u \pa{\f{\ro_{\p} + \ro_{\Phi}}{2} - \ro} \\
	 + \mymin{\alpha,\beta \in [0,2 \pi] } \pa{A_u \cos 2\alpha  + B_u \cos \beta \sin 2\alpha}.
 \end{multline*}
 Optimizing over $\alpha$ yields the optimal value $2\alpha_u \in \pi \N + \arctan\pa{B_u (\cos \beta)/A_u}$ and using the classical formulas for $\cos \arctan$ and $\sin \arctan$, we get
 \begin{align*}
	 A_u \cos 2\alpha_u  + B_u \cos \beta \sin 2\alpha_u = 
	  \pm \sqrt{ A_u^2 + (\cos \beta)^2 B_u^2}.
 \end{align*}
Finally optimizing over $\beta$ gives $\beta_u = 0$, so
 \begin{align}\label{dunn}
	  {^+}\delta_v \ger{k} (u) 
	 = \int_\Omega \pa{\f{\ro_{\p} + \ro_{\Phi}}{2} -\ro} u - \ep\sqrt{A_u^2 + B_u^2},
 \end{align}
and for $B \neq 0$,
 \begin{align*}
\cos 2\alpha_u =  - \f{\ep A_u}{\sqrt{A_u^2+B_u^2}}, \bhs \sin 2\alpha_u =  -\f{\ep B_u}{\sqrt{A_u^2+B_u^2}}
 \end{align*}

 In order to compute the supremum over directions $u$, we consider the Lagrangian $\cL(u,\lambda) := {^+}\delta_v \ger{k} (u) - \lambda \pa{ \int_\Omega u^2 - 1}$, which is $\cC^{\ii}$ when $u$ is not a constant, and the Euler-Lagrange equation is
 \begin{align*}
	 2\lambda u^* & = \f{\ro_{\p} + \ro_{\Phi}}{2} - \ro - \ep\f{\pa{\ro_{\p}-\ro_{\Phi}}A_{u^*} + 2B_{u^*}\ro_{\p,\Phi}}{2\sqrt{A_{u^*}^2+B_{u^*}^2}} \\
	 & = \f{\ro_{\p} + \ro_{\Phi}}{2} - \ro +  (\cos 2\alpha_{u^*}) \f{\ro_{\p} - \ro_{\Phi}}{2} + (\sin 2\alpha_{u^*}) \ro_{\p,\Phi},
 \end{align*}
hence the optimal direction belongs to the directions
 \begin{align*}
	 u_{\theta} := c_{\theta} \pa{\f{\ro_{\p} + \ro_{\Phi}}{2} - \ro +  (\cos 2\theta) \f{\ro_{\p} - \ro_{\Phi}}{2} + (\sin 2\theta) \ro_{\p,\Phi}},
 \end{align*}
 where $\theta \in \seg{0,2\pi}$. The prefactor $c_{\theta} \ge 0$ is chosed such that $\int_\Omega u^2 = 1$. We could then reintroduce this expression back into \eqref{dunn} and reoptimize over $\theta$ to find an equation that $\theta$ has to verify, to maximize ${^+}\delta_v \ger{k} (u_{\theta})$, but this relation is quite involved. Denoting by $\theta^*$ the optimizing angle, which is the same for $k= m_k^v$ and $k=M_k^v$, we reduced the problem to a circle search
 \begin{align*}
	 \mysup{u \in L^2(\Omega,\R) \\ \nor{u}{L^2}=1} {^+}\delta_v \ger{k} (u) & = \mymax{\theta \in \seg{0,2\pi}} {^+}\delta_v \ger{k} (u_{\theta}) ={^+}\delta_v \ger{k} (u_{\theta^*}) \\
	 & = \int_\Omega \pa{\f{\ro_{\p} + \ro_{\Phi}}{2} -\ro} u_{\theta^*} - \ep\sqrt{A_{u_{\theta^*}}^2 + B_{u_{\theta^*}}^2}.
 \end{align*}
 

\section*{Appendix 3: computations for the implementation}

Here we provide complementary computations for the description of the algorithm of Section \ref{simus}.

\subsection{Optimal direction over mixed states} 
We now provide the details of the computations linked to the problem \eqref{passt}. We express the cost function in terms of the parameters $\Gamma_{IJ}$ and provide its gradient, needed in the implementation of ODA.

Once again we take the notations of Section \ref{defdu}, where we approximate 
 \begin{align*}
	 \kerr \tx{ by } \acs{ \p \in H^1\ind{a}(\Omega^N,\C) \st \ab{\cE_v\pa{\p}-\exc{k}(v)} \le t},  
 \end{align*}
having dimension $\ab{\cI\ind{tot}} = \ab{\cI\ind{out}}$, and we define $Q\ind{out} := \vect_{\R} \bpa{\wedge_{i \in I} \p_I}_{I \in \cI\ind{out}}$. Mixed states are decomposed into $\Gamma = \sum_{I,J \in \cI\ind{tot}} \Gamma_{IJ} \ketbra{\p_I}{\p_J}$, where $(\Gamma_{IJ})_{IJ}$ is a real positive matrix with unit trace, and
 \begin{align*}
	 \ro_{\Gamma} & = N\sum_{I,J \in \cI\ind{tot}} \Gamma_{IJ}\int_{\Omega^{N-1}} \p_I \p_J = \ro\ind{in} + \sum_{I,J \in \cI\ind{out}}\Gamma_{IJ} \ro\expo{out}_{IJ},
 \end{align*}
where $\ro\expo{out}_{IJ} := N\ind{out}\int_{\Omega^{N\ind{out}-1}} \p_I \p_J$. The problem \eqref{passt} can be reformulated by
 \begin{align*}
 \cP(v) = \mymin{\Gamma \in \cS(Q\ind{out}) \\ \Gamma \ge 0, \tr \Gamma = 1} e^{ \f{\pa{e\ind{in}+ \cE_v(\Gamma) - \exc{k}(v)}^2}{2T^2}} \int_{\Omega} \pa{\ro_{\Gamma} + \ro\ind{in} - \ro}^2,
 \end{align*}
where $e\ind{in} := \sum_{i \in I\ind{in}} E_i$. 
For any function $g : \cM \ra \R$, where $\cM$ is a smooth manifold modeled on a Hilbert space $\cH$ with scalar product $\ps{\ps{\cdot,\cdot}}$, we recall that we can define the gradient $\na_x g \in \Td_x \cM \simeq \cH$ by Riesz' theorem, via $(\d_x g) a = \ps{\ps{\na_x g,a}}$.
The energy of a mixed state is
 \begin{align*}
	 \cE_v \pa{\Gamma}  = \sum_{I,J \in \cI\ind{tot} } \Gamma_{IJ} \ps{\p_J, \sum_{i=1}^N \pa{-\Delta_i + v(x_i)} \p_I} 
	 = e\ind{in} + \mysum{I,J \in \cI\ind{out}}{} \Gamma_{IJ} e^v_{IJ},
 \end{align*}
 where for $I,J \in \cI\ind{out}$,
 \begin{align*}
	 e^v_{IJ} &:= \ps{\p_J, \sum_{i=1}^{N\ind{out}} \pa{-\Delta_i + v(x_i)} \p_I} \\
	 & =\delta_{IJ}\sum_{i \in I}E_i + \delta_{I \cupdot J = \acs{i,j}} \int_\Omega  \vp_j (-\Delta + v)\vp_i
 \end{align*}
Hence $\pa{\na_{\Gamma} \cE_v}_{IJ} = e^v_{IJ}$ which does not depend on $\Gamma$. The function we optimize is
 \begin{align*}
	 &  f\pa{\Gamma}= e^{ \f{\pa{e\ind{in} - \exc{k}(v) + \mysum{K,L\in \cI\ind{out}}{} \Gamma_{KL} e^v_{KL} }^2}{2T^2}} \bigg( \int_\Omega (\ro\ind{in} - \ro)^2 \\
	 & \bhs \bhs + 2\mysum{I,J\in \cI\ind{out}}{} \Gamma_{IJ}\int_\Omega (\ro\ind{in} - \ro) \ro_{KL} +  \mysum{I,J \in \cI\ind{out}  \\ K,L\in \cI\ind{out}}{} \Gamma_{IJ}\Gamma_{KL} \int_\Omega \ro_{IJ} \ro_{KL} \bigg),
 \end{align*}
 and has gradient
 \begin{align*}
	 & \bhs \pa{\na_{\Gamma} f}_{IJ} 
	 = e^{ \f{\pa{e\ind{in} - \exc{k}(v) + \mysum{K,L\in \cI\ind{out}}{} \Gamma_{KL} e^v_{KL} }^2}{2T^2}} \Bigg( 2\int_\Omega (\ro\ind{in} - \ro) \ro\expo{out}_{IJ} \\
	 & +  2 \mysum{K,L\in \cI\ind{out}}{} \Gamma_{KL} \int_\Omega \ro\expo{out}_{IJ} \ro\expo{out}_{KL}+ T^{-2} e_{IJ}^v \Big( e\ind{in} - \exc{k}(v) + \mysum{K,L\in \cI\ind{out}}{} \Gamma_{KL} e^v_{KL} \Big) \\
	 & \times \bigg( \int_\Omega (\ro\ind{in}-\ro)^2 + 2\mysum{K,L\in \cI\ind{out}}{} \Gamma_{KL} \int_\Omega (\ro\ind{in}-\ro)\ro\expo{out}_{KL} \\
	 & \bhs\bhs\bhs \bhs \bhs \bhs + \mysum{K,L\in \cI\ind{out} \\ X,Y \in \cI\ind{out}}{} \Gamma_{KL}\Gamma_{XY}\int_\Omega \ro\expo{out}_{KL}\ro\expo{out}_{XY} \bigg) \Bigg).
 \end{align*}
Hence, to launch the ODA, one has to compute
\begin{align*}
e\ind{in} - \exc{k}(v), \int_\Omega (\ro\ind{in}-\ro)^2,   \int_\Omega (\ro\ind{in} - \ro) \ro\expo{out}_{KL}, e_{KL}^v, \int_\Omega \ro\expo{out}_{KL}\ro\expo{out}_{XY}.
\end{align*}

\subsection{Optimization over pure states}
We take the notations of Section \ref{defdu}, and here give an extra computation concerning \eqref{passtp}. The familly $(\wedge_{i \in I\ind{in}} \vp_i \wedge \p_I)_{I \in \cI\ind{out}}$ is a basis of $\kerrr$. We can represent the eigenfunctions of the $k\expo{th}$ $N$-body level by complex columns vectors $C$, 
 \begin{align*}
	 \p = \sum_{I \in \cI\ind{tot}} C_I \p_I = \bigwedge_{j \in I\ind{in}} \vp^j  \sum_{I \in \cI\ind{out}} C_I \bigwedge_{i\in I} \vp^i.
 \end{align*}
The coefficients verify $\sum_{I \in \cI\ind{out}} \ab{C_I}^2 = 1$ and the set of such $C$'s forms a Grassmann manifold. Then the density is
\begin{align*}
	\ro_{\p} = \ro\ind{in} + \sum_{I \in \cI\ind{out}} \ab{C_I}^2 \sum_{i \in I} \vp_i^2+  2 \mysum{I,J \in \cI\ind{out} \\ \exists \ell,i,j \\ I \cupdot J = \acs{i,j}}{}  \pa{\re C_I \ov{C_J}}\vp_i \vp_j,
\end{align*}
and is invariant under $U\bpa{\ab{\cI\ind{out}}}$ transformations of $C$.




\bibliographystyle{siam}
\bibliography{build/inverse_pot_static.bbl}
\end{document}